\documentclass[11pt]{article}

\usepackage[margin=1in]{geometry}
\usepackage{setspace}
\usepackage{amsmath}
\usepackage{amsfonts}
\usepackage{amssymb}
\usepackage{amsthm}
\usepackage{mathtools}

\usepackage{microtype}  

\usepackage{graphicx}
\usepackage{float}
\usepackage{placeins}

\usepackage{booktabs}
\usepackage{caption}
\usepackage{subcaption}

\usepackage[round]{natbib}

\usepackage{algorithm}
\usepackage{algpseudocode}

\usepackage{xcolor}
\usepackage[colorlinks=true,
            linkcolor=blue,
            citecolor=blue,
            urlcolor=blue]{hyperref}

\usepackage{accents}
\usepackage{enumitem}

\theoremstyle{plain}
\newtheorem{theorem}{Theorem}
\newtheorem{lemma}{Lemma}

\theoremstyle{definition}

\newtheorem{assumption}{Assumption}

\theoremstyle{remark}
\newtheorem{remark}{Remark}

\newcommand{\indep}{\mathrel{\perp\!\!\!\perp}}

\title{Conformalized Method for Empirical Bayes Normal Mean Inference Problem with Heteroscedastic Variance}
\date{}

\author{
Kwangok Seo\thanks{Department of Statistics, Seoul National University, Seoul, Korea}
\and
Johan Lim\thanks{Department of Statistics, Seoul National University, Seoul, Korea}
\thanks{Corresponding author: johanlim@snu.ac.kr}
}

\begin{document}

\maketitle

\begin{abstract}
    We study the normal mean inference problem, which involves simultaneous testing of the means of many normal distributions. This problem has been extensively studied within the empirical Bayes (EB) framework. However, the reliability of most EB methods heavily depends on two key conditions: (i) the prior distribution is correctly specified, and (ii) it can be accurately estimated. In practice, both conditions are difficult to satisfy, and it is often unclear whether they hold in a given application.
    To overcome these limitations, we propose a new algorithm, called COIN (COnformal Inference for Normal mean inference problem). Unlike traditional empirical Bayes approaches, COIN produces decision rules whose validity does not depend on the correct specification or accurate estimation of the prior. We theoretically prove that COIN asymptotically controls the false discovery rate at the nominal level, even in the presence of prior misspecification or estimation errors.
    Since the COIN algorithm requires an external training dataset to estimate the prior distribution and conformity score function, we introduce two data-splitting strategies---sample-splitting and feature-splitting---for the case where such external data are unavailable. We provide theoretical guarantees for the data-splitting strategies and demonstrate their effectiveness through extensive numerical studies and three real data examples.
\medskip

\end{abstract}

\section{Introduction}
\subsection{Normal Mean Inference Problem}
We consider the problem of simultaneously testing multiple hypotheses based on paired summary statistics $(X_i, S_i^2)$, for $i = 1, \ldots, m$. Each pair $(X_i, S_i^2)$ is assumed to follow the joint distribution
\begin{equation*}
    (X_i, S_i^2) \sim N(\mu_i, \sigma_i^2) \otimes \frac{\sigma_i^2}{\nu}\chi^2_\nu,
\end{equation*}
where $X_i$ is an estimator of the parameter of interest $\mu_i$, and $S_i^2$ is an estimator of its variance. Here, $N(\mu, \sigma^2)$ denotes the normal distribution with mean $\mu$ and variance $\sigma^2$, while $\chi^2_\nu$ denotes the chi-squared distribution with $\nu$ degrees of freedom. 
The hypotheses of interest are given by
\begin{equation*}
    H_{i,0}: \mu_i = 0 
    \quad \text{vs.} \quad 
    H_{i,1}: \mu_i \neq 0, 
    \quad \text{for } i \in [m] \coloneqq \{1, 2, \ldots, m\}.
\end{equation*}
Each hypothesis is associated with an observed pair of summary statistics $(X_i, S_i^2)$, and the goal is to perform simultaneous inference across all $m$ hypotheses using these statistics. This setting is commonly referred to as the \emph{normal mean inference problem} (NMIP), as each hypothesis pertains to the mean parameter of a normally distributed random variable.

There are two main approaches to the NMIP: the empirical Bayes (EB) \citep{lonnstedt2002replicated, lu2016variance, stephens2017false, lu2019empirical, zheng2021mixtwice, fu2022heteroscedasticity, seo2025empirical} and the empirical partially Bayes (EPB) \citep{smyth2004linear, ignatiadis2025empirical} frameworks. The key distinction between the two lies in how the parameter of interest, $\mu_i$, is treated. In the EB framework, both the parameter of interest $\mu_i$ and the nuisance parameter $\sigma_i^2$ are regarded as random variables, and a prior is placed on their joint distribution. In contrast, the EPB framework, originally discussed by \citet{cox1975note}, treats the parameter of interest $\mu_i$ as an unknown but fixed constant, assigning a prior only to the nuisance parameter $\sigma_i^2$. This conceptual difference leads to distinct modeling assumptions and inference procedures. Throughout this paper, we primarily focus on the EB framework.

\subsection{Existing Empirical Bayes Normal Mean Inference Methods}
Empirical Bayes normal mean inference (EBNMI) methods consist of three essential components: (i) a \textit{prior specification} for the parameter pair $(\mu_i, \sigma_i^2)$; (ii) a \textit{conformity score function}; and (iii) an \textit{estimation algorithm} for the prior distribution and the conformity score function. Distinct variants of EBNMI methods are determined by particular choices of these components.

\subsubsection{General Procedure}
Most existing EBNMI methods share a common four-step procedure:

\textit{Step 1: Specify the Prior Distribution and the Conformity Score Function.} 
The procedure begins by specifying a prior distribution for the parameters $(\mu_i, \sigma_i^2)$ and a corresponding conformity score function, which plays a central role in the final decision rule. The joint prior distribution is typically represented as a hierarchical model, consisting of a marginal distribution for $\sigma_i^2$ and a conditional distribution for $\mu_i$ given $\sigma_i^2$. We adopt this hierarchical representation throughout the paper. The conformity score function, denoted by
\begin{equation*}
    u: \mathbb{R} \times [0, \infty) \to [0, \infty),
\end{equation*}
quantifies the strength of evidence against the null hypothesis. Specifically, it maps an observed summary statistic $(X, S^2)$ to a nonnegative real number. Without loss of generality, we assume that the function is defined so that smaller values of $u(X, S^2)$ indicate stronger evidence against the null hypothesis.

\textit{Step 2: Estimate the Prior Distribution and the Conformity Score Function.} Given the prior specification, the next step is to estimate the prior distribution using the observed summary statistics $\mathcal{D} \coloneqq \{(X_i, S_i^2)\}_{i=1}^m$. This estimation is typically carried out through an algorithm tailored to the assumed prior structure. Since the conformity score function depends on the prior distribution, this step also produces an estimated conformity score function, denoted by $\hat{u}(\cdot, \cdot)$.

\textit{Step 3: Compute Conformity Scores.} Using the estimated conformity score function $\hat{u}(\cdot, \cdot)$, we compute the conformity score for each hypothesis as
\begin{equation*}
    u_i = \hat{u}(X_i, S_i^2), \quad i \in [m].
\end{equation*}
Each hypothesis is therefore associated with a conformity score, which summarizes the evidence against the null hypothesis.

\textit{Step 4: Derive a Data-Adaptive Threshold and Define the Decision Rule.} 
A data-adaptive threshold $\tau$ is determined according to the target level of type~I error. Based on this threshold, the decision rule is defined as
\begin{equation*}
    \delta_i = \mathbb{I}\{u_i \leq \tau\}, \quad i \in [m],
\end{equation*}
where $\delta_i = 1$ indicates rejection of the $i$th null hypothesis and $\delta_i = 0$ otherwise. This decision rule is reasonable, as smaller conformity scores indicate stronger evidence against the null hypothesis.

To clarify how this framework is instantiated in a concrete setting, Section~\ref{ex_MixTwice} of the Supplementary Material provides a detailed exposition using the \textit{MixTwice} method of \cite{zheng2021mixtwice} as an illustrative example, explicitly demonstrating how each step of the general procedure is implemented.

\subsubsection{Limitations of the Existing EBNMI Methods}
A multiple testing procedure is said to be ``valid'' if it controls the type~I error at a pre-specified level. The validity of existing EBNMI methods heavily depends on the extent of the discrepancy between the estimated prior distribution and the true prior distribution. When this discrepancy is small, the resulting decision rule typically achieves type~I error control at the nominal level. However, as the discrepancy increases, the validity of this guarantee deteriorates, potentially leading to a type~I error rate that substantially exceeds the nominal level.

The extent of of the discrepancy between the estimated and the true prior distributions can largely be attributed to two sources: \textit{prior misspecification} and \textit{estimation error}. Prior misspecification occurs when the assumed form of the prior fails to accurately reflect the true prior distribution. Estimation error, on the other hand, may arise even when the prior is correctly specified, due to factors such as a finite number of experimental units, suboptimal optimization procedures, or inherent non-identifiability.

To mitigate the risk of misspecification, one may employ a more flexible prior that accommodates a broader class of distributions. However, such flexibility often comes at the expense of less stable and more challenging estimation, thereby increasing estimation error. Conversely, adopting a simpler prior specification increases the risk of model misspecification, as it may fail to capture key features of the true prior distribution. This trade-off highlights the need for an EBNMI procedure that is robust against discrepancies between the estimated and the true prior distributions.

\subsubsection{An Illustrative Numerical Example}
We examine the sensitivity of EBNMI methods to prior misspecification. To this end, we compare existing empirical Bayes methods (\textit{LS}; \citealp{lu2019empirical}, \textit{MixTwice}; \citealp{zheng2021mixtwice}, and \textit{gg-Mix}; \citealp{seo2025empirical}) with one of our proposed methods (\textit{COIN-FS}), focusing on their ability to control the FDR under both correctly specified and misspecified priors.

We generate $m = 20{,}000$ independent observations $(X_i, S_i^2)$ from the hierarchical model
\begin{equation*}
\begin{gathered}
    \sigma_i^2 \overset{\text{i.i.d.}}{\sim} 6/\chi^2_6, 
    \quad 
    \mu_i \mid \sigma_i^2 \overset{\text{i.i.d.}}{\sim} (1-\pi)\,\delta_0(\cdot) + \pi\,f(\cdot \mid \sigma_i^2), \\
    (X_i, S_i^2) \mid \mu_i, \sigma_i^2 
    \overset{\text{ind.}}{\sim} 
    N(\mu_i, \sigma_i^2)\;\otimes\; \frac{\sigma_i^2}{18}\chi^2_{18}.
\end{gathered}
\end{equation*}
The non-null proportion $\pi$ takes values in $\{0.5, 0.6, 0.7, 0.8, 0.9\}$.
We consider two simulation scenarios that differ in the specification of the non-null distribution $f(\cdot \mid \sigma_i^2)$:
\begin{enumerate}[label=(\Roman*)]
    \item Scenario 1. 
    The non-null distribution $f(\cdot \mid \sigma_i^2)$ is specified as a Gaussian scale mixture,
    \begin{equation*}
        f(\cdot \mid \sigma_i^2) = 0.1\, N(0,1) + 0.2\, N(0,4) + 0.7\, N(0,16).
    \end{equation*}
    This specification yields a symmetric and unimodal density independent of $\sigma_i^2$.

    \item Scenario 2.
    The non-null distribution $f(\cdot \mid \sigma_i^2)$ is specified as a Gaussian location mixture,
    \begin{equation*}
        f(\cdot \mid \sigma_i^2) = 0.3\, N(-3\sigma_i^2, \sigma_i^2) + 0.7\, N(4\sigma_i^2, \sigma_i^2).
    \end{equation*}
    This specification yields an asymmetric and bimodal density that depends on $\sigma_i^2$.
\end{enumerate}

All methods rely on modeling assumptions for the unknown prior distribution. In our comparison, all methods assume that the non-null distribution $f(\cdot \mid \sigma_i^2)$ is unimodal and symmetric, with $\mu_i$ independent of $\sigma_i^2$. This assumption is consistent with the data-generating mechanism in Scenario~1, where the true non-null distribution is a symmetric Gaussian scale mixture independent of $\sigma_i^2$. In contrast, it is violated in Scenario~2, where the true non-null distribution is an asymmetric Gaussian location mixture that depends on $\sigma_i^2$. Apart from the specification of $f(\cdot \mid \sigma_i^2)$, all other components of the prior distribution are correctly specified for all methods. In particular, we assume independence between $\mathbb{I}(\mu_i = 0)$ and $\sigma_i^2$, and model the prior distribution of $\sigma_i^2$ either under a scaled inverse-$\chi^2$ distribution or nonparametrically with support on $[0,\infty)$. Accordingly, Scenario~1 represents a correctly specified setting, whereas Scenario~2 corresponds to a misspecified setting in which only $f(\cdot \mid \sigma_i^2)$ is misspecified.

The simulation results are summarized in Figure~\ref{fig:Simulation2}. Under Scenario~1, \textit{LS}, \textit{MixTwice}, and \textit{gg-Mix} successfully control the FDR at the nominal level of $0.1$ and consistently achieve higher power than \textit{COIN-FS}. 
In contrast, under Scenario~2 the FDR of \textit{LS}, \textit{MixTwice}, and \textit{gg-Mix} substantially exceeds the nominal level across most settings. This inflation becomes more pronounced as the proportion of non-null hypotheses approaches $0.5$, a regime in which the impact of prior misspecification is more severe. 
These findings highlight a fundamental limitation of existing EBNMI methods and underscore the need for approaches whose validity remains robust to prior misspecification, or more generally, to discrepancies between the estimated and the true prior distributions. 
In contrast, our proposed method, \textit{COIN-FS}, is specifically designed to be robust to prior misspecification and consistently maintains valid FDR control even when the prior is misspecified.

\begin{figure}[!htb]
    \centering
    \includegraphics[width=\textwidth]{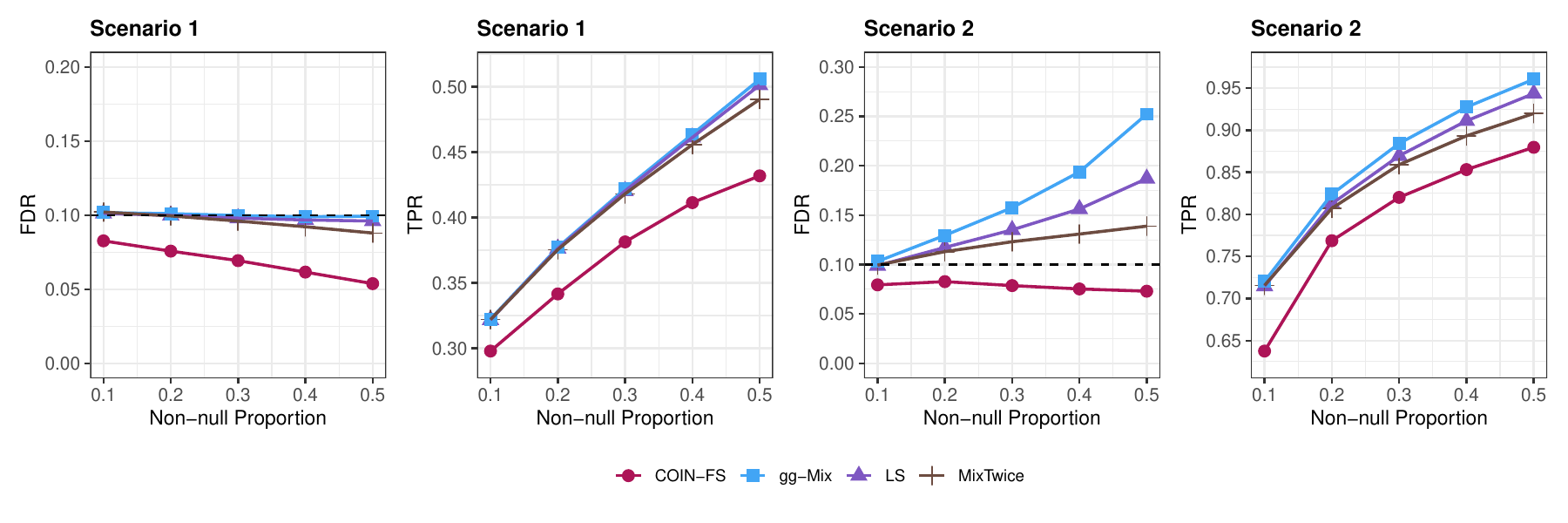}
    \caption{False discovery rates (FDRs) and true positive rates (TPRs) under Scenarios~1 and 2. Each panel compares the four methods \textit{LS}, \textit{MixTwice}, \textit{gg-Mix}, and \textit{COIN-FS} across varying non-null proportions $\pi$, with results averaged over 200 simulation replicates. The dashed horizontal line in the FDR panels indicates the target FDR level 0.1.}
    \label{fig:Simulation2}
\end{figure}

\subsection{Our Contributions}
The main contributions of this paper are as follows.

\begin{enumerate}[label=(\Roman*)]
    \item To address the limitations of existing EBNMI methods, we propose an algorithm called COIN (COnformal Inference for the Normal mean inference problem) that guarantees asymptotic FDR control regardless of discrepancies between the estimated and true prior distributions. The robustness of COIN allows greater flexibility in specifying the prior and choosing estimation procedures without compromising inferential validity.

    \item While COIN provides strong theoretical guarantees, it requires an external training dataset in addition to the test dataset, which may not always be available in practice. To address this limitation, we introduce two data-splitting strategies—\textit{sample-splitting} and \textit{feature-splitting}. The key idea is to construct an internal training dataset from the test dataset itself, enabling the use of the COIN framework even without external training data.

    \item We show that both data-splitting methods are theoretically valid, in the sense that they asymptotically control the FDR at the target level regardless of discrepancies between the estimated and true prior distributions. The validity of the sample-splitting method follows directly from that of the COIN framework. The validity of the feature-splitting method is established by leveraging properties of the $e$BH procedure \citep{wang2022false} and is motivated by the concept of asymptotic compound $e$-variables introduced by \citet{ignatiadis2024asymptotic}. To the best of our knowledge, these are the first EBNMI methods that achieve provable FDR control while remaining robust to prior misspecification.

    \item We evaluate two data-splitting variants of the COIN framework—\textit{sample-splitting} and \textit{feature-splitting}—under several forms of prior misspecification, including:
    \begin{enumerate}[label=(\roman*)]
    \item incorrectly assuming independence between $\mu_i$ and $\sigma_i^2$, i.e., $\mu_i \indep \sigma_i^2$;
    
    \item misspecifying the functional relationship between $\mu_i$ and $\sigma_i^2$;
    
    \item imposing restrictive assumptions on the distribution class of $G$ or the density $f$.
    \end{enumerate}
    
    The results show that the proposed methods reliably control the FDR across all forms of prior misspecification, even in finite-sample settings.
    \end{enumerate}

\subsection{Related Works}
\cite{fu2022heteroscedasticity} study the NMIP under an empirical Bayes framework and establish asymptotic optimality of their procedure, in the sense of achieving maximal detection power while controlling the marginal false discovery rate. These guarantees, however, rely on the assumptions that the variances $\sigma_i^2$ are known and that $\mu_i$ and $\sigma_i^2$ are independent. However, we consider a setting where $\sigma_i^2$ are unknown and only their estimates $s_i^2$ are available, and we allow for potential dependence between $\mu_i$ and $\sigma_i^2$. As a result, the asymptotic theory in \cite{fu2022heteroscedasticity} does not directly apply to our setting.
Recently, \cite{gang2025large} extend the work of \cite{fu2022heteroscedasticity} to composite nulls and allow arbitrary dependence between $\mu_i$ and $\sigma_i^2$, while still assuming that $\sigma_i^2$ are known. 
Although they suggest replacing $\sigma_i^2$ with $s_i^2$ for practical implementation, this approximation may be less accurate when the sample size for variance estimation is limited, as is common in practice. 

\cite{zhao2025conformalized} consider a multiple testing problem with side information, encompassing the NMIP as a special case, and provide a general framework that guarantees finite-sample FDR control, regardless of the accuracy of the assumed empirical Bayes working model. Their theory relies on the availability of calibration variables satisfying an exchangeability condition, as defined in (10) of \cite{zhao2025conformalized}. In practice, however, such variables may not be available, or it may be unclear whether the variables at hand satisfy the required condition. When the null distribution is independent of the side information and known, valid calibration variables can be obtained as discussed in \cite{zhao2025conformalized}. In our setting, however, the null distribution depends on $\sigma_i^2$, which serve as side information. Moreover, the variances $\sigma_i^2$ are unknown. Thus, this approach is not directly applicable to our setting.

A similar challenge arises in the knockoff filter for high-dimensional variable selection \citep{barber2015controlling, candes2018panning}, where finite-sample FDR control requires generating knockoff variables satisfying a similar exchangeability condition. Constructing such variables requires knowledge of true distribution of the features \citep{bates2021metropolized}; otherwise, one can only construct variables that approximately satisfy the required condition \citep{romano2020deep}. To quantify the impact of deviations from exchangeability on FDR control, \cite{barber2020robust} use the Kullback--Leibler divergence as a measure of deviation and establish an upper bound on the FDR as a function of this divergence.

Our work is closely related to \cite{barber2020robust} in the sense described above, but differs in several key aspects. \cite{barber2020robust} study high-dimensional variable selection in a regression setting and derive an FDR bound for given knockoff variables that holds regardless of the relationship between the response and the features. In contrast, we consider a multiple testing problem for normal means, propose a method for constructing calibration variables that asymptotically satisfy the required property in an appropriate sense, and establish asymptotic FDR control regardless of discrepancies between the estimated and true prior distributions.

\paragraph{Outline.}
The remainder of this paper is organized as follows. Section~\ref{sec_2} defines the problem setting. Section~\ref{sec_3} introduces our proposed method, COIN, and establishes its theoretical guarantees, showing that the resulting decision rule asymptotically controls the FDR at the nominal level. Section~\ref{sec_4} describes two data-splitting strategies—sample-splitting and feature-splitting—for implementing COIN when an external training dataset is unavailable. Section~\ref{sec_5} presents numerical studies evaluating the empirical performance of the two data-splitting methods and comparing them with existing approaches. Section~\ref{sec_6} applies our method to three real datasets. Section~\ref{sec_7} concludes the paper with a discussion.

\section{Problem Settings}\label{sec_2}
Consider two datasets, $\mathcal{D}_1$ and $\mathcal{D}_2$, each containing observations of paired statistics:
\begin{equation*}
    \mathcal{D}_k = \left\{(X_{i,k}, S_{i,k}^2) : i \in [m^{\mathcal{D}_k}] \right\}, \quad k = 1, 2,
\end{equation*}
where $m^{\mathcal{D}_k}$ denotes the number of pairs in dataset $\mathcal{D}_k$. 
Throughout, $\mathcal{D}_1$ and $\mathcal{D}_2$ are treated as the test and (external) training datasets, respectively.
We assume that each pair $(X_{i,k}, S_{i,k}^2)$ arises from the following hierarchical model:
\begin{equation}\label{hierar_model}
\begin{split}
    \sigma_{i,k}^2 
    &\overset{\text{i.i.d.}}{\sim} G(\cdot), \\
    \mu_{i,k} \mid \sigma_{i,k}^2 
    &\overset{\text{ind.}}{\sim} (1 - \pi) \delta_0(\cdot) + \pi f(\cdot; \sigma_{i,k}^2), \\
    X_{i,k}, S_{i,k}^2 \mid \mu_{i,k}, \sigma_{i,k}^2 
    &\overset{\text{ind.}}{\sim} N(\mu_{i,k}, \sigma_{i,k}^2) \otimes \frac{\sigma_{i,k}^2}{\nu} \chi^2_\nu.
\end{split}
\end{equation}
In this model, the variance component $\sigma_{i,k}^2$ is independently drawn from an arbitrary distribution $G$ supported on $[0,\infty)$. Conditional on $\sigma_{i,k}^2$, the mean parameter $\mu_{i,k}$ follows a two-component mixture: with probability $1-\pi$ it equals zero (the null), and with probability $\pi$ it is drawn from a non-null distribution with density $f(\cdot;\sigma_{i,k}^2)$, which may depend on $\sigma_{i,k}^2$. 
Given $\mu_{i,k}$ and $\sigma_{i,k}^2$, the observed effect size $X_{i,k}$ follows a normal distribution with mean $\mu_{i,k}$ and variance $\sigma_{i,k}^2$, while the corresponding variance estimate $S_{i,k}^2$ follows a scaled chi-squared distribution with $\nu$ degrees of freedom.

The problem considered here is the NMIP, in which we simultaneously test $m^{\mathcal{D}_1}$ hypotheses of the form
\begin{equation}\label{hypo}
    H_{i,0}: \mu_{i,1} = 0 
    \quad \text{vs.} \quad 
    H_{i,1}: \mu_{i,1} \neq 0, 
    \quad \text{for } i \in [m^{\mathcal{D}_1}].
\end{equation}
Our objective is to develop a provably valid EBNMI method by leveraging both the test dataset $\mathcal{D}_1$ and the training dataset $\mathcal{D}_2$. In particular, we seek valid control of the type~I error even when the estimated prior distribution differs from the true one.

To assess the validity of multiple testing methods, we adopt the false discovery rate (FDR; \citealp{benjamini1995controlling}) as the type~I error criterion. The FDR is defined as
\begin{equation*}
    \mathrm{FDR}(\boldsymbol{\delta}) = \mathbb{E}\big[\mathrm{FDP}(\boldsymbol{\delta})\big], 
    \quad 
    \mathrm{FDP}(\boldsymbol{\delta}) = 
    \frac{\sum_{i = 1}^{m^{\mathcal{D}_1}} (1 - \theta_i)\, \delta_i}
         {\left\{ \sum_{i = 1}^{m^{\mathcal{D}_1}} \delta_i \right\} \vee 1},
\end{equation*}
where $a \vee b = \max(a,b)$. Here $\theta_i$ denotes the true state of the $i$th hypothesis, with $\theta_i = 0$ if the null hypothesis $H_{i,0}$ is true and $\theta_i = 1$ otherwise. The decision rule is $\boldsymbol{\delta} = \{\delta_i : i \in [m^{\mathcal{D}_1}]\}$, where $\delta_i = 1$ indicates rejection of $H_{i,0}$ and $\delta_i = 0$ otherwise.

\begin{remark}[Availability of the External Training Dataset $\mathcal{D}_2$]
    A key requirement for the theoretical validity of our procedure is that the prior distribution $G$, used to generate calibration variables, and the conformity score function $u(\cdot,\cdot)$ be estimated independently of the test dataset $\mathcal{D}_1$.
    To ensure this independence, we assume the availability of an external training dataset $\mathcal{D}_2$ that is independent of $\mathcal{D}_1$. In many practical situations, however, such a dataset may not be available. To address this limitation, we introduce two data-splitting strategies in Section~\ref{sec_4}.
\end{remark}

\begin{remark}[Hierarchical Model]
The hierarchical model in~\eqref{hierar_model} is highly flexible and encompasses many existing prior specifications as special cases. For example, if the prior distribution $G$ is restricted to a scaled inverse chi-squared family and, for some positive constant $a$, the conditional distribution of $\mu_{i,k}$ given $\sigma_{i,k}^2$ satisfies
\begin{equation*}
    \frac{\mu_{i,k}}{\sigma_{i,k}^a} \mid \sigma_{i,k}^2 
    \sim 
    (1 - \pi)\, \delta_0(\cdot) + \pi\, f(\cdot),
    \quad \text{where } f \text{ is a unimodal density centered at zero},
\end{equation*}
then the model reduces to that considered in~\citet{lu2019empirical}. 
In contrast, if no restriction is imposed on $G$ and only the conditional distribution of $\mu_{i,k}$ given $\sigma_{i,k}^2$ is assumed to satisfy
\begin{equation*}
    \mu_{i,k} \mid \sigma_{i,k}^2 
    \sim 
    (1 - \pi)\, \delta_0(\cdot) + \pi\, f(\cdot),
    \quad \text{where } f \text{ is a unimodal density centered at zero},
\end{equation*}
then the model reduces to that of~\citet{zheng2021mixtwice}.
\end{remark}

We conclude this section by introducing some notation that will be used throughout the paper. Let $\mathcal{H}_0^{\mathcal{D}_1}$, $\mathcal{H}_1^{\mathcal{D}_1}$, and $\mathcal{H}^{\mathcal{D}_1}$ denote the index sets of true null, true non-null, and all hypotheses, respectively. Define $m_0^{\mathcal{D}_1} = |\mathcal{H}_0^{\mathcal{D}_1}|$, $m_1^{\mathcal{D}_1} = |\mathcal{H}_1^{\mathcal{D}_1}|$, and $m^{\mathcal{D}_1} = |\mathcal{H}^{\mathcal{D}_1}|$, where $|A|$ denotes the cardinality of a set $A$. For sequences of real numbers $\{x_n\}$ and $\{y_n\}$, we write $y_n = o(x_n)$ if $y_n / x_n \to 0$ as $n \to \infty$, and $y_n = O(x_n)$ if there exist constants $M > 0$ and $n_0 \in \mathbb{N}$ such that $|y_n| \le M |x_n|$ for all $n \ge n_0$.

\section{Methodology}\label{sec_3}

\subsection{Construction of Calibration Variables}\label{sub_sec_3.1}
Building on ideas from conformal inference \citep{vovk2005algorithmic, angelopoulos2024theoretical, zhao2025conformalized} and the knockoff filter \citep{barber2015controlling, candes2018panning}, COIN constructs calibration variables that act as negative controls, thereby mitigating the impact of discrepancies between the estimated and true prior distributions on FDR control. We begin by describing how the calibration variables are constructed.

\paragraph{Conditional Exchangeability.} 
A desirable property of calibration variables is \textit{conditional exchangeability}. Let $\tilde{X}_{i,1}$ denote the calibration variable corresponding to the $i$th hypothesis. For each hypothesis, we consider the triplet $(X_{i,1}, \tilde{X}_{i,1}, S_{i,1}^2)$. We say that the calibration variables satisfy conditional exchangeability if
\begin{equation}\label{cond_exch}
    (X_{i,1}, \tilde{X}_{i,1}) \mid S_{i,1}^2 
    \overset{d}{=} 
    (\tilde{X}_{i,1}, X_{i,1}) \mid S_{i,1}^2, 
    \quad \forall i \in \mathcal{H}_0^{\mathcal{D}_1}.
\end{equation}
This condition implies that, under the null hypothesis, 
$X_{i,1}$ and $\tilde{X}_{i,1}$ are conditionally exchangeable given $S_{i,1}^2$. Consequently, comparing $X_{i,1}$ with its calibration counterpart $\tilde{X}_{i,1}$ provides a meaningful way to assess how extreme the observed statistic is under the null.

\paragraph{Oracle Construction.} 
We first consider the oracle setting in which the prior distribution $G$ is known. In this case, the conditional distribution of $X$ given $\mu = 0$ and $S^2 = s^2$ is  
\begin{equation*}
\begin{split}
    X \mid \mu = 0, S^2 = s^2 
    \sim h_G(x \mid \mu = 0, S^2 = s^2) 
    \coloneqq 
    \frac{\int p(x \mid \mu = 0, \sigma^2)\, p(s^2 \mid \sigma^2)\, dG(\sigma^2)}
         {\int p(s^2 \mid \sigma^2)\, dG(\sigma^2)},
\end{split}
\end{equation*}
where the subscript $G$ emphasizes that this conditional distribution depends on the prior $G$. Based on this distribution, the calibration variables are generated as  
\begin{equation}\label{oracle_cal}
    \tilde{X}_{i,1}^* \overset{\text{ind.}}{\sim} h_G(x \mid \mu = 0, S^2 = S_{i,1}^2), 
    \quad i \in [m^{\mathcal{D}_1}],
\end{equation}
where the superscript $*$ indicates that $\tilde{X}_{i,1}^*$ is generated under the oracle setting. By construction, the calibration variables in~\eqref{oracle_cal} satisfy the conditional exchangeability property.

\paragraph{Data-Adaptive Construction.} 
The oracle construction relies on knowledge of the prior distribution $G$, which is typically unavailable in practice. To address this limitation, we approximate the oracle procedure in a data-adaptive manner by estimating $G$ from the training dataset $\mathcal{D}_2$. 
Specifically, we employ the nonparametric maximum likelihood estimator (NPMLE):
\begin{equation}\label{NPMLE_G}
    \hat{G} 
    = \underset{G \in \mathcal{G}}{\arg\max} 
      \sum_{i = 1}^{m^{\mathcal{D}_2}} 
      \log \left\{\int p(S^2 = S_{i,2}^2 \mid \sigma^2)\, dG(\sigma^2)\right\},
\end{equation}
where $\mathcal{G}$ denotes the set of all distributions supported on $[0, \infty)$. Using the plug-in principle, we estimate the conditional density by replacing the unknown prior $G$ with its NPMLE $\hat{G}$:
\begin{equation*}
    h_{\hat{G}}(x \mid \mu = 0, S^2 = s^2) 
    = 
    \frac{\int p(x \mid \mu = 0, \sigma^2)\, p(s^2 \mid \sigma^2)\, d\hat{G}(\sigma^2)}
         {\int p(s^2 \mid \sigma^2)\, d\hat{G}(\sigma^2)}.
\end{equation*}
Based on this estimated density, the calibration variables are generated as
\begin{equation}\label{pseudo_cal}
    \tilde{X}_{i,1} \overset{\text{ind.}}{\sim} 
    h_{\hat{G}}(x \mid \mu = 0, S^2 = S_{i,1}^2), 
    \quad i \in [m^{\mathcal{D}_1}].
\end{equation}
Unlike the oracle case, these calibration variables do not satisfy the conditional exchangeability property, due to estimation errors in $\hat{G}$.

\begin{remark}[Practical Implementation of the NPMLE]
   To construct the calibration variables in a data-adaptive manner, we need to solve the optimization problem described in~\eqref{NPMLE_G}. For practical implementation, rather than directly solving~\eqref{NPMLE_G}, we follow the approach of \citet{koenker2014convex}, which restricts the class $\mathcal{G}$ to discrete distributions with fixed finite support. Maximizing the objective in~\eqref{NPMLE_G} over this restricted class yields an estimate of $G$, serving as an approximation to the NPMLE. 
\end{remark}

\begin{remark}
     When the prior distribution $G$ is known, calibration variables can be constructed so as to satisfy the conditional exchangeability property in~\eqref{cond_exch}, which is key to finite-sample FDR control (see Section~\ref{oracle_case} of the Supplementary Material for details). When the calibration variables are constructed using an estimated prior $\hat{G}$, this conditional exchangeability property generally fails to hold, creating a fundamental obstacle to finite-sample FDR control. Nevertheless, asymptotic FDR control remains achievable even when $G$ is unknown, as established in Theorem~\ref{theorem1}.
\end{remark}

\subsection{COIN Algorithm}\label{sub_sec_3.2}
In this section we introduce the COIN algorithm for solving the NMIP. The procedure is designed to be robust to discrepancies between the estimated and true prior distributions. The overall algorithm is summarized in Algorithm~\ref{alg_1} in Section~\ref{algorithms} of the Supplementary Material and consists of four main steps.

\textit{Step 1: Construct Calibration Variables.} For each hypothesis $i \in \mathcal{H}^{\mathcal{D}_1}$, we construct a calibration variable $\tilde{X}_{i,1}$ following~\eqref{pseudo_cal}. By pairing each calibration variable with its corresponding variance estimate $S_{i,1}^2$ from the test dataset $\mathcal{D}_1$, we obtain the pseudo calibration dataset
\begin{equation*}
    \tilde{\mathcal{D}}_1 = \left\{ (\tilde{X}_{i,1}, S_{i,1}^2) : i \in [m^{\mathcal{D}_1}] \right\}.
\end{equation*}
This dataset is referred to as the \textit{pseudo calibration dataset} to distinguish it from the \textit{oracle calibration dataset} that would be obtained if the true prior $G$ were known.

\textit{Step 2: Specify and Estimate the Working Prior and Conformity Score Function.} A working prior distribution for $(\mu, \sigma^2)$ is specified, together with a conformity score function $u(\cdot, \cdot)$ that maps each data pair $(X, S^2)$ to a non-negative score. These quantities are then estimated using the external training dataset $\mathcal{D}_2$, which is independent of the test dataset $\mathcal{D}_1$. This independence is essential to ensure the validity of the procedure. A variety of prior specifications, conformity score functions, and estimation algorithms can be employed at this stage, including those proposed by \citet{lu2016variance}, \citet{stephens2017false}, \citet{lu2019empirical}, \citet{zheng2021mixtwice} and \citet{seo2025empirical}, among others.

\textit{Step 3: Compute Conformity Scores.} 
The estimated conformity score function $\hat{u}(\cdot, \cdot)$ is evaluated on both the test dataset and the pseudo calibration dataset:
\begin{equation*}
    \mathcal{U} \coloneqq \{u_i = \hat{u}(X_{i,1}, S_{i,1}^2) : i \in \mathcal{H}^{\mathcal{D}_1}\}, \quad 
    \tilde{\mathcal{U}} \coloneqq \{\tilde{u}_i = \hat{u}(\tilde{X}_{i,1}, S_{i,1}^2) : i \in \mathcal{H}^{\mathcal{D}_1}\}.
\end{equation*}
Throughout, we assume that smaller conformity scores indicate stronger evidence against the null hypothesis.

\textit{Step 4: Derive a Data-Adaptive Threshold and Define the Decision Rule.} Define a data-adaptive threshold $\tau$ as
\begin{equation}\label{thresh_unknownG}
    \tau = \max \left\{ t \in \mathcal{U} \cup \tilde{\mathcal{U}} :
    \frac{1 + \sum_{i \in \mathcal{H}^{\mathcal{D}_1}} \mathbb{I}(\tilde{u}_i \leq u_i \wedge t)}
    {\left[\sum_{i \in \mathcal{H}^{\mathcal{D}_1}} \mathbb{I}(u_i \leq \tilde{u}_i \wedge t)\right] \vee 1}
    \leq \alpha \right\},
\end{equation}
where $\alpha$ denotes the target FDR level, and $a \wedge b$ denotes $\min(a, b)$. By convention, we set $\tau = -\infty$ if the maximization set in~\eqref{thresh_unknownG} is empty, in which case no hypotheses are rejected. Otherwise, the decision rule is given by
\begin{equation*}
    \delta_i = \mathbb{I}\!\left(u_i \le \tilde{u}_i \wedge \tau\right),
    \qquad i \in \mathcal{H}^{\mathcal{D}_1}.
\end{equation*}

\begin{remark}[Data-Adaptive Threshold $\tau$]
    To understand the rationale behind the definition of the data-adaptive threshold $\tau$ in~\eqref{thresh_unknownG}, consider the oracle setting in which the prior distribution $G$ is known and the calibration variables satisfy the conditional exchangeability property in~\eqref{cond_exch}. In this case, for any threshold $t$, we have
    \begin{equation*}
    \begin{split}
        \text{FDP}(t) 
        &= \frac{\sum_{i \in \mathcal{H}_0^{\mathcal{D}_1}} \mathbb{I}(u_i \leq \tilde{u}_i \wedge t)}
        {\left[\sum_{i \in \mathcal{H}^{\mathcal{D}_1}} \mathbb{I}(u_i \leq \tilde{u}_i \wedge t)\right] \vee 1}
        \approx 
        \frac{\sum_{i \in \mathcal{H}_0^{\mathcal{D}_1}} \mathbb{I}(\tilde{u}_i \leq u_i \wedge t)}
        {\left[\sum_{i \in \mathcal{H}^{\mathcal{D}_1}} \mathbb{I}(u_i \leq \tilde{u}_i \wedge t)\right] \vee 1} \\[3pt]
        &\leq 
        \frac{1 + \sum_{i \in \mathcal{H}^{\mathcal{D}_1}} \mathbb{I}(\tilde{u}_i \leq u_i \wedge t)}
        {\left[\sum_{i \in \mathcal{H}^{\mathcal{D}_1}} \mathbb{I}(u_i \leq \tilde{u}_i \wedge t)\right] \vee 1}
        \eqqcolon \widehat{\text{FDP}}(t),
    \end{split}
    \end{equation*}
    where the approximation is justified by conditional exchangeability, which implies
    \begin{equation*}
        \sum_{i \in \mathcal{H}_0^{\mathcal{D}_1}} \mathbb{I}(u_i \leq \tilde{u}_i \wedge t)
        \overset{d}{=} 
        \sum_{i \in \mathcal{H}_0^{\mathcal{D}_1}} \mathbb{I}(\tilde{u}_i \leq u_i \wedge t).
    \end{equation*}
    Hence, the data-adaptive threshold $\tau$ can be equivalently written as
    \begin{equation*}
        \tau = \max \left\{ t \in \mathcal{U} \cup \tilde{\mathcal{U}} : \widehat{\text{FDP}}(t) \leq \alpha \right\}.
    \end{equation*}
\end{remark}

\begin{remark}(Decision Rule $\boldsymbol{\delta}$)
    Since smaller conformity scores indicate stronger evidence against the null hypothesis, a natural rejection condition is $u_i \leq \tilde{u}_i$, that is, the observed score is smaller than its corresponding calibration value. However, this condition may frequently hold even under the null hypothesis; if the prior distribution $G$ were known, then by conditional exchangeability, approximately half of the true nulls would satisfy $u_i \leq \tilde{u}_i$. To guard against such false discoveries, we further impose a global threshold~$\tau$. Hence, a hypothesis is rejected only when both $u_i \leq \tilde{u}_i$ and $u_i \leq \tau$ are satisfied.
\end{remark}

\subsection{Asymptotic FDR Control}\label{sub_sec_3.3}
Following \citet{ignatiadis2025empirical}, we assume that the support of $G$ is bounded away from both $0$ and $\infty$.

\begin{assumption}[Assumption 8 in \citealp{ignatiadis2025empirical}]\label{assum_1}
    Let $0 < \underaccent{\bar}{L}  \leq L \leq U \leq \bar{U} < \infty$. We assume that $\sigma_{i,2}^2 \in [L, U]$ for all $i$ (almost surely) and that we compute NPMLE $\hat{G}$ as in \eqref{NPMLE_G} under the additional condition $G[\underaccent{\bar}{L}, \bar{U}] = 1$.
\end{assumption}

Theorem~\ref{theorem1} establishes that the decision rule derived from COIN asymptotically controls the FDR at the nominal level~$\alpha$. We emphasize that the asymptotic statement is made with respect to the training sample size, i.e., $m^{\mathcal{D}_2} \to \infty$. The proof is provided in Section~\ref{proof_sec_3} of the Supplementary Material.

\begin{theorem}\label{theorem1}
    Suppose that the hierarchical model described in \eqref{hierar_model} and Assumption~\ref{assum_1} hold.  
    Also assume that $ \nu \geq 2 $ and that there are no ties between $ u_i $ and $ \tilde{u}_i $ for all $ i \in \mathcal{H}^{\mathcal{D}_1} $.  
    Then, the FDR of the decision rule $\boldsymbol{\delta}$, obtained from COIN, is:
    \begin{equation*}
        \text{FDR}(\boldsymbol{\delta}) = \alpha + O\left( \left(\frac{(\log{m^{\mathcal{D}_2}})^5}{m^{\mathcal{D}_2}}\right)^{1/4}\right),
    \end{equation*}
    where $\alpha$ denotes a pre-specified FDR level.
\end{theorem}

In the oracle case, where the prior distribution~$G$ is known, the FDR is controlled at the nominal level~$\alpha$ in finite samples (see Theorem~\ref{theorem4} in Section~\ref{oracle_case} of the Supplementary Material). By contrast, when $G$ is unknown, the pairwise exchangeability condition generally fails to hold, leading to an additional $O(\cdot)$ term. This term represents the statistical cost incurred by estimating~$G$ from data.

\begin{remark}
Theorem~\ref{theorem1} is closely related to the result of \citet{ignatiadis2025empirical}, who derived a similar FDR bound under the compound setting. Let $\boldsymbol{\delta}_{\mathrm{IS}}$ denote their proposed decision rule. They showed that, under the compound setting,
\begin{equation*}
    \mathrm{FDR}(\boldsymbol{\delta}_{\mathrm{IS}}) 
    \leq \alpha + \frac{C}{\kappa_m} \!\left( \frac{(\log m)^5}{m} \right)^{1/4} + \eta_m,    
\end{equation*}
where $m$ is the number of tested hypotheses, and $\kappa_m$ and $\eta_m$ are tuning parameters such that $\mathbb{P}(R < m\kappa_m) \leq \eta_m$, with $R$ denoting the number of rejections made by the procedure. If the rejection ratio $R/m$ remains constant as $m \to \infty$, meaning that the procedure consistently rejects a fixed proportion of hypotheses, then for any $0 < \kappa_m < R/m$ we have $\eta_m = 0$. In this case, the bound simplifies to
\begin{equation*}
    \mathrm{FDR}(\boldsymbol{\delta}_{\mathrm{IS}}) 
    = \alpha + O\!\left( \left( \frac{(\log m)^5}{m} \right)^{1/4} \!\right),
\end{equation*}
which matches the order of the error term in Theorem~\ref{theorem1}.
\end{remark}

\section{Data Splitting Methods}\label{sec_4}

\subsection{Setting without External Training Dataset}
Thus far, our framework has assumed the availability of an external training dataset~$\mathcal{D}_2$. In many practical scenarios, however, such a dataset may not be available. To accommodate this more realistic setting, we introduce two data-splitting strategies, sample-splitting and feature-splitting. The key idea is to construct an internal training dataset from the test dataset itself, thereby enabling the application of the COIN algorithm even in the absence of external training data. To this end, we extend the definition of the test dataset. Specifically, we consider two forms of test datasets:

\begin{enumerate}[label=(\Roman*)]
    \item The \textit{individual-level test dataset}, denoted by $\mathcal{D}^{\text{raw}} \in \mathbb{R}^{m^{\mathcal{D}} \times n}$, is a data matrix in which each row corresponds to a feature (e.g., a gene) and each column corresponds to a sample (e.g., a subject), containing the full set of individual observations.
 
    \item The \textit{summary-level test dataset}, denoted by $\mathcal{D} \coloneqq \{(X_i, S_i^2)\}_{i \in \mathcal{H}^\mathcal{D}}$, consists of summary statistics derived from the individual-level dataset. Since we focus on the setting where only the test dataset is available, we simplify the notation by omitting subscripts that distinguish between the training and test datasets; for instance, we write $X_i$ instead of $X_{i,1}$. Compared with the individual-level dataset, the summary-level dataset is less informative because it retains only aggregated summary statistics rather than the original observations.
\end{enumerate}

\noindent In what follows, we distinguish between these two types of test datasets and use the corresponding terminology.

\subsection{Sample-Splitting Method}\label{sub_sec_4.1}
In this section, we consider the setting in which the individual-level test dataset $\mathcal{D}^{\text{raw}}$ is available. In other words, rather than assuming that the summary-level test dataset 
$\mathcal{D} = \{(X_i, S_i^2)\}_{i \in \mathcal{H}^\mathcal{D}}$ 
has been pre-computed and provided, we assume access to its underlying source data.

The approach begins by randomly dividing the columns (i.e., samples) of the individual-level test dataset $\mathcal{D}^{\text{raw}}$ into two disjoint subsets of equal size, denoted by $\mathcal{D}_1^{\text{raw}}$ and $\mathcal{D}_2^{\text{raw}}$, while taking into account the underlying data structure (see Remark~\ref{remark_ss} for concrete examples). We then compute summary statistics for each subset to obtain two datasets of the form
\begin{equation*}
    \mathcal{D}_1^{\text{split}} = \{(X_{i,1}^{(\text{split})}, S_{i,1}^{2(\text{split})})\}_{i \in \mathcal{H}^\mathcal{D}}, 
    \quad 
    \mathcal{D}_2^{\text{split}} = \{(X_{i,2}^{(\text{split})}, S_{i,2}^{2(\text{split})})\}_{i \in \mathcal{H}^\mathcal{D}}.
\end{equation*}
Each summary pair $(X_{i,k}^{(\text{split})}, S_{i,k}^{2(\text{split})})$ is computed from $\mathcal{D}_k^{\text{raw}}$ for $k = 1, 2$. The superscript ``split'' is used to explicitly distinguish between the test and training datasets considered in Section~\ref{sec_3} and those obtained by splitting the individual-level test dataset. Once the summary datasets are constructed, we apply COIN algorithm, treating $\mathcal{D}_1^{(\text{split})}$ as the test dataset and $\mathcal{D}_2^{(\text{split})}$ as the training dataset. The overall procedure of the sample-splitting method is summarized in Algorithm~\ref{alg_2}, presented in Section~\ref{algorithms} of the Supplementary Material.

Under the common assumption of independence among samples, sample-splitting yields two independent subsets, $\mathcal{D}_1^{\text{raw}}$ and $\mathcal{D}_2^{\text{raw}}$. Consequently, the corresponding summary-level datasets, $\mathcal{D}_1^{\text{split}}$ and $\mathcal{D}_2^{\text{split}}$, are also independent. Hence, $\mathcal{D}_2^{\text{split}}$ serves as an external training dataset independent of $\mathcal{D}_1^{\text{split}}$, and the theoretical guarantees established in Section~\ref{sub_sec_3.3} continue to hold for the decision rule derived from the sample-splitting method.

\begin{remark}\label{remark_ss}
    The sample-splitting method applies to a wide range of experimental designs. For illustration, consider first the one-group setting, where all samples are drawn from a single population and the goal is to infer properties of this common distribution. In this case, the total set of samples is divided into two subsets of equal size: one used to construct $\mathcal{D}_1^{\text{raw}}$ and the other to construct $\mathcal{D}_2^{\text{raw}}$. 
    As another example, consider the two-group setting, where samples are independently drawn from two distinct populations, typically representing different experimental conditions such as treatment and control. Here, the raw test data matrix has dimension $m \times (n_1 + n_2)$, where $n_1$ and $n_2$ denote the numbers of samples from each population, respectively. We independently split the samples within each group into two equal-sized parts: one part from each group is used to construct $\mathcal{D}_1^{\text{raw}}$, and the remaining parts are used to construct $\mathcal{D}_2^{\text{raw}}$.
\end{remark}

\subsection{Feature-Splitting Method}\label{sub_sec_4.2}
Although the sample-splitting method described in the previous section is statistically natural and conceptually simple, it suffers from a loss of power since only half of the original samples are used to construct the training dataset. To address this limitation, we propose an alternative splitting strategy designed to enhance power. This approach splits the features rather than the samples and, importantly, can be implemented using only the summary-level test dataset $\mathcal{D} = \{(X_i, S_i^2)\}_{i \in \mathcal{H}^{\mathcal{D}}}$ 
without requiring access to the individual-level data.

We adopt a $K$-fold splitting scheme inspired by \citet{ignatiadis2021covariate}. Specifically, we randomly partition the index set $\mathcal{H}^{\mathcal{D}}$ into $K$ disjoint subsets $\mathcal{H}^{(1)}, \dots, \mathcal{H}^{(K)}$ of approximately equal size, and define the corresponding data splits as
\begin{equation*}
    \mathcal{D}^{(k)} \coloneqq \{(X_i, S_i^2) : i \in \mathcal{H}^{(k)}\}, 
    \quad k = 1, \dots, K.
\end{equation*}
For each fold $k$, we treat $\mathcal{D}^{(k)}$ as the test dataset and define the corresponding training dataset as
\begin{equation*}
    \mathcal{D}^{(-k)} \coloneqq \{(X_i, S_i^2) : i \in \mathcal{H}^{\mathcal{D}} \setminus \mathcal{H}^{(k)}\}.
\end{equation*}
We then follow the steps of COIN algorithm up to the computation of the threshold. Specifically, we first estimate the prior distribution $G$ using $\mathcal{D}^{(-k)}$, yielding $\hat{G}^{(k)}$. 
Based on $\hat{G}^{(k)}$, we generate calibration variables $\tilde{X}_i^{(k)}$ for each $i \in \mathcal{H}^{(k)}$. 
We then specify a working prior and a conformity score function $u(\cdot,\cdot)$, which are estimated using $\mathcal{D}^{(-k)}$. 
Let $\hat{u}^{(k)}(\cdot,\cdot)$ denote the resulting estimated conformity score function. We then compute the conformity scores $u_i^{(k)} = \hat{u}^{(k)}(X_i^{(k)}, S_i^{2(k)})$ and $\tilde{u}_i^{(k)} = \hat{u}^{(k)}(\tilde{X}_i^{(k)}, S_i^{2(k)})$, for all $i \in \mathcal{H}^{(k)}$. Next, we determine a threshold $\tau^{(k)}$ as in~\eqref{thresh_unknownG}, using a fold-specific target FDR level $\alpha^{(k)}$. 
To maintain clarity, we explicitly index all fold-specific quantities by $k$, including $\hat{G}^{(k)}$, $X_i^{(k)}$, $\tilde{X}_i^{(k)}$, $S_i^{2(k)}$, $\hat{u}^{(k)}(\cdot, \cdot)$, $u_i^{(k)}$, $\tilde{u}_i^{(k)}$, $\tau^{(k)}$, and $\alpha^{(k)}$. Using these quantities, we define a test statistic for each $i \in \mathcal{H}^{(k)}$ as
\begin{equation}\label{e_variable}
    E_i^{(k)} \coloneqq 
    \frac{|\mathcal{H}^{(k)}| \cdot \xi_i^{(k)} 
    \cdot \mathbb{I}\!\left(s_i^{(k)} \leq \tau^{(k)}\right)}
    {1 + \sum_{j \in \mathcal{H}^{(k)}} 
    \left(1 - \xi_j^{(k)}\right) 
    \cdot \mathbb{I}\!\left(s_j^{(k)} \leq \tau^{(k)}\right)},
\end{equation}
where $\xi_i^{(k)} = \mathbb{I}\!\left(u_i^{(k)} \leq \tilde{u}_i^{(k)}\right)$ and $s_i^{(k)} = u_i^{(k)} \wedge \tilde{u}_i^{(k)}$.

By repeating this procedure for each fold $k = 1, \dots, K$, we obtain fold-specific collections of test statistics:
\begin{equation*}
    \mathcal{E}^{(k)} \coloneqq \{E_i^{(k)} : i \in \mathcal{H}^{(k)}\}, 
    \quad k = 1, \dots, K.
\end{equation*}
We then aggregate the test statistics across all folds into a single collection: 
\begin{equation*}
    \mathcal{E} \coloneqq \bigcup_{k = 1}^K \mathcal{E}^{(k)}.    
\end{equation*}
Although the random variables in $\mathcal{E}$ are not generalized/compound $e$-variables in the theoretical sense, 
we treat them as such and apply the $e$BH procedure \citep{wang2022false} with a target FDR level $\alpha_{e\text{BH}}$ to construct the rejection set. This constitutes the basic version of the feature-splitting method.

To further improve detection power, the U-$e$BH procedure proposed by \citet{xu2023more} can be used as an alternative to the standard $e$BH procedure. In this approach, an external random variable $U \sim \mathrm{Unif}(0,1)$, independent of all elements in $\mathcal{E}$, is introduced. Each statistic in $\mathcal{E}$ is then scaled by $U$, yielding
\begin{equation*}
    \mathcal{E}^* \coloneqq 
    \left\{ E_i^{(k)} / U : i \in \mathcal{H}^{(k)},\ k \in [K] \right\}.
\end{equation*}
The standard $e$BH procedure is subsequently applied to the scaled collection $\mathcal{E}^*$. When $\mathcal{E}$ consists of valid generalized/compound $e$-variables, this randomized scaling step enhances power while maintaining valid FDR control. The complete feature-splitting procedure is summarized in Algorithm~\ref{alg_3}, presented in Section~\ref{algorithms} of the Supplementary Material.

\begin{remark}
    If the prior distribution $G$ were known, the calibration variables would satisfy the conditional exchangeability property described in~\eqref{cond_exch}. In this oracle setting, the resulting set $\mathcal{E}$ forms a valid collection of generalized/compound $e$-variables regardless of the choice of the fold-specific target FDR levels $\alpha^{(k)} \in (0,1)$ (see Section~\ref{oracle_case} of the Supplementary Material for details). Consequently, applying the $e$BH or U-$e$BH procedure to $\mathcal{E}$ guarantees exact finite-sample FDR control \citep{wang2022false, xu2023more}.
    In practice, however, the prior distribution $G$ is unknown and must be estimated from data. As a result, the conditional exchangeability property generally fails to hold, and exact finite-sample FDR control is no longer guaranteed. Nevertheless, Theorems~\ref{theorem2} and~\ref{theorem3} show that, for fixed $K$, the feature-splitting method asymptotically controls the FDR at the target level as $m^{\mathcal{D}} \to \infty$. Empirical evidence supporting this theoretical result is provided in Section~\ref{sec_5}.
\end{remark}

\subsubsection{Asymptotic FDR Control of the Feature-Splitting Method}
We establish the asymptotic validity of the feature-splitting method, building on the concept of asymptotic compound $e$-variables introduced by \citet{ignatiadis2024asymptotic}. The main results are summarized in Theorems~\ref{theorem2} and~\ref{theorem3}.

Theorem~\ref{theorem2} establishes that the decision rule obtained by applying the $e$BH procedure to the collection $\mathcal{E}$ asymptotically controls the FDR at the target level~$\alpha$. A proof is given in Section~\ref{proof_sec_4} of the Supplementary Material.

\begin{theorem}\label{theorem2}
Let $\boldsymbol{\delta}_{e\text{BH}} = \big\{ \delta_{i, e\text{BH}}^{(k)} : i \in \mathcal{H}^{(k)},\ k \in [K] \big\}$ 
denote the collection of decision rules obtained by applying 
the $e$BH procedure to the collection $\mathcal{E} = \big\{ E_i^{(k)} : i \in \mathcal{H}^{(k)},\ k \in [K] \big\}$. For a fixed number of folds~$K$, as the total number of hypotheses $|\mathcal{H}^{\mathcal{D}}| \to \infty$, the procedure asymptotically controls the FDR at the pre-specified level~$\alpha_{e\text{BH}}$. Formally,
\begin{equation*}
    \mathrm{FDR}(\boldsymbol{\delta}_{e\text{BH}}) 
    = \alpha_{e\text{BH}} + o(1).
\end{equation*}
\end{theorem}

Theorem~\ref{theorem3} shows that replacing the $e$BH procedure 
with the U-$e$BH procedure yields a decision rule that also 
asymptotically controls the FDR. The proof is given in Section~\ref{proof_sec_4} of the Supplementary Material.

\begin{theorem}\label{theorem3}
Let $\boldsymbol{\delta}_{\text{U-}e\text{BH}} = 
\{\delta_{i, \text{U-}e\text{BH}}^{(k)} : i \in \mathcal{H}^{(k)},\ k \in [K]\}$ 
denote the collection of decision rules obtained by applying 
the U-$e$BH procedure to the collection 
$\mathcal{E}^* = \{ E_i^{(k)}/U : i \in \mathcal{H}^{(k)},\ k \in [K] \}$, 
where $U \sim \mathrm{Unif}(0,1)$ is an external random variable 
independent of all $E_i^{(k)}$. 
For a fixed number of folds~$K$, as the total number of hypotheses 
$|\mathcal{H}^{\mathcal{D}}| \to \infty$, 
the procedure asymptotically controls the FDR at the pre-specified level~$\alpha_{e\text{BH}}$. 
Formally,
\begin{equation*}
    \mathrm{FDR}(\boldsymbol{\delta}_{\text{U-}e\text{BH}}) 
    = \alpha_{e\text{BH}} + o(1).
\end{equation*}
\end{theorem}

\begin{remark}[Uniform Power Improvement via Threshold Refinement]
    Following \cite{ren2024derandomised}, the feature-splitting method attains a uniform power improvement through a refinement of the fold-specific data-adaptive threshold $\tau^{(k)}$ without compromising type I error control. This refinement is applied across all numerical experiments and real-data applications in this paper. See Section~\ref{implementation_details_data_split} of the Supplementary Material for details.
\end{remark}

\section{Numerical Study}\label{sec_5}
\subsection{Simulated Data Generation}
To implement the sample-splitting method, access to the individual-level test dataset is required. We begin by describing the construction of the individual-level test dataset $\mathcal{D}^{\text{raw}} \in \mathbb{R}^{m^{\mathcal{D}} \times n}$ used in our simulation study. From this dataset, we derive the corresponding summary-level dataset $\mathcal{D} = \{(X_i, S_i^2)\}_{i=1}^{m^{\mathcal{D}}}$ and verify that it satisfies the hierarchical model structure specified in~\eqref{hierar_model}.

Without loss of generality, we consider a two-group comparison setting in which the objective is to test for differences in means between a treatment group and a control group across multiple features. Let $n = n_1 + n_2$ denote the total number of subjects, where $n_1$ and $n_2$ are the sample sizes of the treatment and control groups, respectively, and let $m^{\mathcal{D}}$ denote the number of features.

\paragraph{Construction of Individual-level Test Dataset.} 
For each feature $i \in [m^{\mathcal{D}}]$, we generate the latent parameters $(\mu_i, \sigma_i^2)$ from a pre-specified prior under the hierarchical model
\begin{equation*}
    \sigma_{i}^2 \overset{\text{i.i.d.}}{\sim} G(\cdot), 
    \quad
    \mu_i \mid \sigma_i^2 \overset{\text{ind.}}{\sim} 
    (1 - \pi)\, \delta_0(\cdot) + \pi\, f(\cdot;\sigma_i^2).
\end{equation*}
Let $r_i$ denote the $i$-th row of the individual-level test dataset:
\begin{equation*}
    r_i^\top = (a_{i1}, \ldots, a_{in_1},\, b_{i1}, \ldots, b_{in_2}),
\end{equation*}
where the first $n_1$ elements correspond to the treatment group and the remaining $n_2$ to the control group. Conditional on $(\mu_i, \sigma_i^2)$, the entries of the $i$-th row are generated as
\begin{equation*}
\begin{split}
    a_{ij} = \mu_{ai} + \epsilon_{aij}, 
    &\quad \epsilon_{aij} \overset{\text{i.i.d.}}{\sim} \mathcal{N}(0, \eta_i^2), 
    \quad j = 1, \ldots, n_1, \\
    b_{ij} = \mu_{bi} + \epsilon_{bij}, 
    &\quad \epsilon_{bij} \overset{\text{i.i.d.}}{\sim} \mathcal{N}(0, \eta_i^2), 
    \quad j = 1, \ldots, n_2,
\end{split}    
\end{equation*}
with the following parameter specifications:
\begin{equation*}
    \mu_{ai} = 
    \begin{cases}
        \mu_i, & \text{if } \mu_i \ne 0, \\
        0,     & \text{if } \mu_i = 0,
    \end{cases}
    \quad 
    \mu_{bi} = 0, 
    \quad \text{and} \quad 
    \eta_i^2 = \sigma_i^2 \left( \frac{1}{n_1} + \frac{1}{n_2} \right)^{-1}.
\end{equation*}
Repeating this process for all features yields the complete individual-level test dataset $\mathcal{D}^{\text{raw}}$.

\paragraph{Derivation of the Summary-Level Test Dataset.} 
From the individual-level test dataset $\mathcal{D}^{\text{raw}}$ obtained through the procedure described above, we compute summary statistics for each feature $i$. Since the goal is to test $H_{i,0} : \mu_i = 0$ for $i = 1, 2, \ldots, m^{\mathcal{D}}$, we adopt the sample mean difference as a natural estimator of $\mu_i$, defined by
\begin{equation*}
    X_i = \bar{a}_i - \bar{b}_i, 
    \quad \text{where} \quad 
    \bar{a}_i = \frac{1}{n_1} \sum_{j=1}^{n_1} a_{ij}, 
    \quad 
    \bar{b}_i = \frac{1}{n_2} \sum_{j=1}^{n_2} b_{ij}.
\end{equation*}
The corresponding variance estimator is given by
\begin{equation*}
    S_i^2 
    = 
    \left(
        \frac{1}{n_1 + n_2 - 2} 
        \big[
            (n_1 - 1) s_{ai}^2 + (n_2 - 1) s_{bi}^2
        \big]
    \right)
    \left( 
        \frac{1}{n_1} + \frac{1}{n_2} 
    \right),
\end{equation*}
where the within-group sample variances are defined as
\begin{equation*}
    s_{ai}^2 = \frac{1}{n_1 - 1} \sum_{j=1}^{n_1} (a_{ij} - \bar{a}_i)^2, 
    \quad 
    s_{bi}^2 = \frac{1}{n_2 - 1} \sum_{j=1}^{n_2} (b_{ij} - \bar{b}_i)^2.
\end{equation*}
By construction, the resulting summary statistics follow the distribution
\begin{equation*}
    X_i,\, S_i^2 \mid \mu_i, \sigma_i^2 
    \overset{\text{ind.}}{\sim} 
    N(\mu_i, \sigma_i^2) 
    \otimes 
    \frac{\sigma_i^2}{\nu} \chi^2_\nu,
    \quad \text{where } \nu = n_1 + n_2 - 2.
\end{equation*}
Thus, the summary-level test dataset $\mathcal{D} = \{(X_i, S_i^2)\}_{i=1}^{m^{\mathcal{D}}}$ follows the hierarchical model specified in~\eqref{hierar_model}.

Unless otherwise specified, in all simulation settings we set $m^{\mathcal{D}} = 20{,}000$ and assume equal sample sizes $n_1 = n_2 = 10$, yielding degrees of freedom $\nu = n_1 + n_2 - 2 = 18$.
\subsection{Performance Metric}
We use the FDR as our Type~I error metric and evaluate power using the true positive rate (TPR). The TPR is defined as the expected value of the true positive proportion (TPP), that is, the ratio of correctly rejected hypotheses among the truly non-null hypotheses:
\begin{equation*}
    \mathrm{TPR}(\boldsymbol{\delta}) = \mathbb{E}[\mathrm{TPP}(\boldsymbol{\delta})], 
    \quad
    \mathrm{TPP}(\boldsymbol{\delta}) = 
    \frac{\sum_{i=1}^m \theta_i \delta_i}
         {\left( \sum_{i=1}^m \theta_i \right) \vee 1}.
\end{equation*}
Among procedures that control the FDR, those attaining higher TPR are considered more powerful. In practice, we report empirical estimates of these metrics, computed as the averages of the observed FDP and TPP, respectively, over 200 independent replications for each simulation setting. The target FDR level is fixed at $\alpha = 0.1$ throughout.

\subsection{The Methods Considered in the Comparison}
We evaluate and compare the performance of five EBNMI methods: three existing approaches, \textit{LS} \citep{lu2019empirical}, \textit{MixTwice} \citep{zheng2021mixtwice}, and \textit{gg-Mix} \citep{seo2025empirical}, and two proposed methods, \textit{COIN-SS} and \textit{COIN-FS}:
\begin{itemize}
    \item \textit{LS} (\citealp{lu2019empirical}):  
    This method assumes a hierarchical prior structure of the form
    \begin{equation*}
        \sigma_i^2 \overset{\text{i.i.d.}}{\sim} \text{Scaled-Inv-}\chi^2, 
        \quad
        \frac{\mu_i}{\sigma_i^a} \mid \sigma_i^2 
        \overset{\text{ind.}}{\sim} 
        (1 - \pi)\, \delta_0(\cdot) + \pi\, f(\cdot),
    \end{equation*}
    where $ f $ is restricted to be a unimodal density and $ a $ is a nonnegative constant. For practical implementation, it has been suggested to approximate the density $ f $ using either a uniform or a half-uniform mixture. In our simulation studies, we consistently employ a uniform mixture to approximate the true underlying density $ f $. Although both $ a = 0 $ and $ a = 1 $ have been suggested as reasonable choices for the constant $ a $, we fix $ a = 0 $ throughout all simulation settings. The \textit{LS} is implemented using the \texttt{R} package \texttt{ashr}.

    \item \textit{MixTwice (\citealp{zheng2021mixtwice})}: 
    This method can be viewed as a generalization of \textit{LS} (with $a=0$), replacing the parametric prior on $\sigma_i^2$ with a nonparametric one. It assumes the following hierarchical model
    \begin{equation*}
        \sigma_i^2 \overset{\text{i.i.d.}}{\sim} G, 
        \quad
        \mu_i \mid \sigma_i^2 \overset{\text{i.i.d.}}{\sim} 
        (1 - \pi)\, \delta_0 + \pi\, f,
    \end{equation*}
    where $G$ is an arbitrary distribution supported on $[0,\infty)$ and $f$ is a unimodal density. The \textit{MixTwice} is implemented using the \texttt{R} package \texttt{MixTwice}.

    \item \textit{gg-Mix (\citealp{seo2025empirical})}:  
    This method can be viewed as a generalization of \textit{MixTwice} in that it relaxes the structural assumptions imposed on the prior for non-null effects $\mu_i$. Specifically, it assumes the hierarchical model
    \begin{equation*}
        \sigma_i^2 \overset{\text{i.i.d.}}{\sim} G, 
        \quad
        \mu_i \mid \sigma_i^2 \overset{\text{i.i.d.}}{\sim} 
        (1 - \pi)\, \delta_0 + \pi\, f,
    \end{equation*}
    where $G$ is an arbitrary distribution supported on $[0,\infty)$ and $f$ is left unconstrained (no unimodality assumption is imposed). For practical implementation, it has been suggested to approximate $f$ using flexible mixture distributions. In our simulations, unless otherwise stated, we consistently employ a Gaussian location-scale mixture to approximate the true underlying density $f$.

    \item \textit{COIN-SS} and \textit{COIN-FS}:      
    These are our proposed methods, which differ in their data-splitting strategies: \textit{COIN-SS} employs sample splitting, whereas \textit{COIN-FS} employs feature splitting. Both methods use the same working prior as \textit{gg-Mix}.\footnote{The working prior is not required to match the prior assumptions of \textit{gg-Mix} and may be chosen arbitrarily. The choice of the working prior does not affect the validity of the procedure, but influences its power.} Additional implementation details are provided in Section~\ref{implementation_details_data_split} of the Supplementary Material.
\end{itemize}

\noindent
We additionally consider the method of \cite{ignatiadis2025empirical} for comparison. However, as this method falls outside the EB framework, we do not report its results in the main manuscript and instead provide them in Section~\ref{add_sim} of the Supplementary Material.
 
\subsection{Simulation Setup}
We examine two distinct simulation scenarios, each characterized by a different dependence structure between the effect size $\mu_i$ and the variance component $\sigma_i^2$.

\begin{enumerate}[label=(\Roman*)]
    \item {Scenario 1 (Independent  $\sigma_i^2$ and $\mu_i$):}  
    We assume that $ \mu_i $ and $ \sigma_i^2 $ are independent. Under this assumption, the hierarchical prior structure is specified as
    \begin{equation*}
        \sigma_i^2 \overset{\text{i.i.d.}}{\sim} G(\cdot), \quad \mu_i \mid \sigma_i^2\overset{\text{i.i.d.}}{\sim} (1-\pi) \delta_0(\cdot) + \pi f(\cdot).
    \end{equation*}
    
    \item {Scenario 2 (dependent $\sigma_i^2$ and $\mu_i$):}  
    We assume the following hierarchical prior structure:
    \begin{equation*}
        \sigma_i^2 \overset{\text{i.i.d.}}{\sim} G(\cdot), \quad \frac{\mu_i}{\sigma_i} \mid \sigma_i^2 \overset{\text{ind.}}{\sim} (1-\pi) \delta_0(\cdot) + \pi f(\cdot),
    \end{equation*}
    which implies that $\sigma_i^2$ is independent of the signal indicator $\mathbb{I}(\mu_i \neq 0)$, but allows dependence between $\mu_i$ and $\sigma_i^2$.
\end{enumerate}
Scenarios~1 and~2 reflect the standard assumptions commonly adopted in many EBNMI methods.

Given the dependence structures described above, we specify the following prior distributions to define the simulation settings.
\begin{enumerate}[label=(\roman*)]
    \item {Prior distributions for $\sigma_i^2$, $G$:}
    \begin{enumerate}
        \item Scaled inverse-$\chi^2$ (SIC): $6/\chi^2_6$,
        \item Point mass at 1 (PM): $\delta_1$,
        \item Two-point discrete (TPD): $0.7 \delta_1 + 0.3 \delta_{10}$.
    \end{enumerate}
    
    \item {Distributions for non-zero $\mu_i$, $f$:}
    \begin{enumerate}
        \item Unimodal (Unimodal): $N(0, 16)$,
        \item Symmetric bimodal (Sym-Bimodal): $0.5 N(-4, 1) + 0.5 N(4, 1)$,
        \item Asymmetric bimodal (Asym-Bimodal): $0.3 N(-3, 1) + 0.7 N(4, 1)$.
    \end{enumerate}
    
    \item {Proportions of non-nulls:} $\pi \in \{0.1, 0.2, 0.3, 0.4, 0.5\}$.
\end{enumerate}

For Scenario~1, where $\mu_i$ and $\sigma_i^2$ are independent, we consider all three prior distributions for $G$ (SIC, PM, and TPD). In contrast, for Scenario~2, the PM prior for $G$ is excluded. This exclusion is due to the fact that, under the PM prior, all $\sigma_i^2$ values are fixed at~1 without variability, making it infeasible to induce meaningful dependence between $\mu_i$ and $\sigma_i^2$. Therefore, only the SIC and TPD priors are used for $G$ in Scenario~2. Combining all configurations, Scenario~1 comprises $3 \times 3 \times 5 = 45$ distinct simulation settings, whereas Scenario~2 comprises $2 \times 3 \times 5 = 30$ settings. 

\subsection{Results}

\paragraph{Simulation Results under Scenario 1.}
Figure~\ref{fig:sim1_fdr} reports the FDR of each method across all simulation settings, from which several key observations emerge.  First, \textit{LS} and \textit{MixTwice} achieve valid FDR control only when their prior assumptions hold (the top two panels of the left column for \textit{LS} and the left column for \textit{MixTwice}). When these assumptions are violated, both methods exhibit substantial FDR inflation, particularly in settings where the non-null density $f$ is bimodal and the non-null proportion $\pi$ is close to~0.5. Second, the remaining methods—\textit{gg-Mix}, \textit{COIN-SS}, and \textit{COIN-FS}—consistently maintain valid FDR control across all simulation settings. This result is expected, as these methods rely on a common prior specification that is sufficiently flexible to accommodate all true priors considered in the data-generating process.

Figure~\ref{fig:sim1_tpr} displays the TPR. Since \textit{LS} and \textit{MixTwice} fail to control the FDR in most cases and offer power comparable to \textit{gg-Mix} when FDR control is achieved, we restrict attention to \textit{gg-Mix}, \textit{COIN-FS}, and \textit{COIN-SS}. As expected, \textit{gg-Mix} uniformly attains the highest power, reflecting its use of the full dataset without splitting. The proposed methods exhibit some loss of power due to data splitting; however, the reduction for \textit{COIN-FS} is modest. In contrast, \textit{COIN-SS} is substantially less powerful and is consistently dominated by \textit{COIN-FS}. These findings demonstrate that \textit{COIN-FS} effectively mitigates the power loss inherent in \textit{COIN-SS}, achieving a marked improvement in detection performance.

\begin{figure}[!htb]
    \centering
    \includegraphics[width=\textwidth]{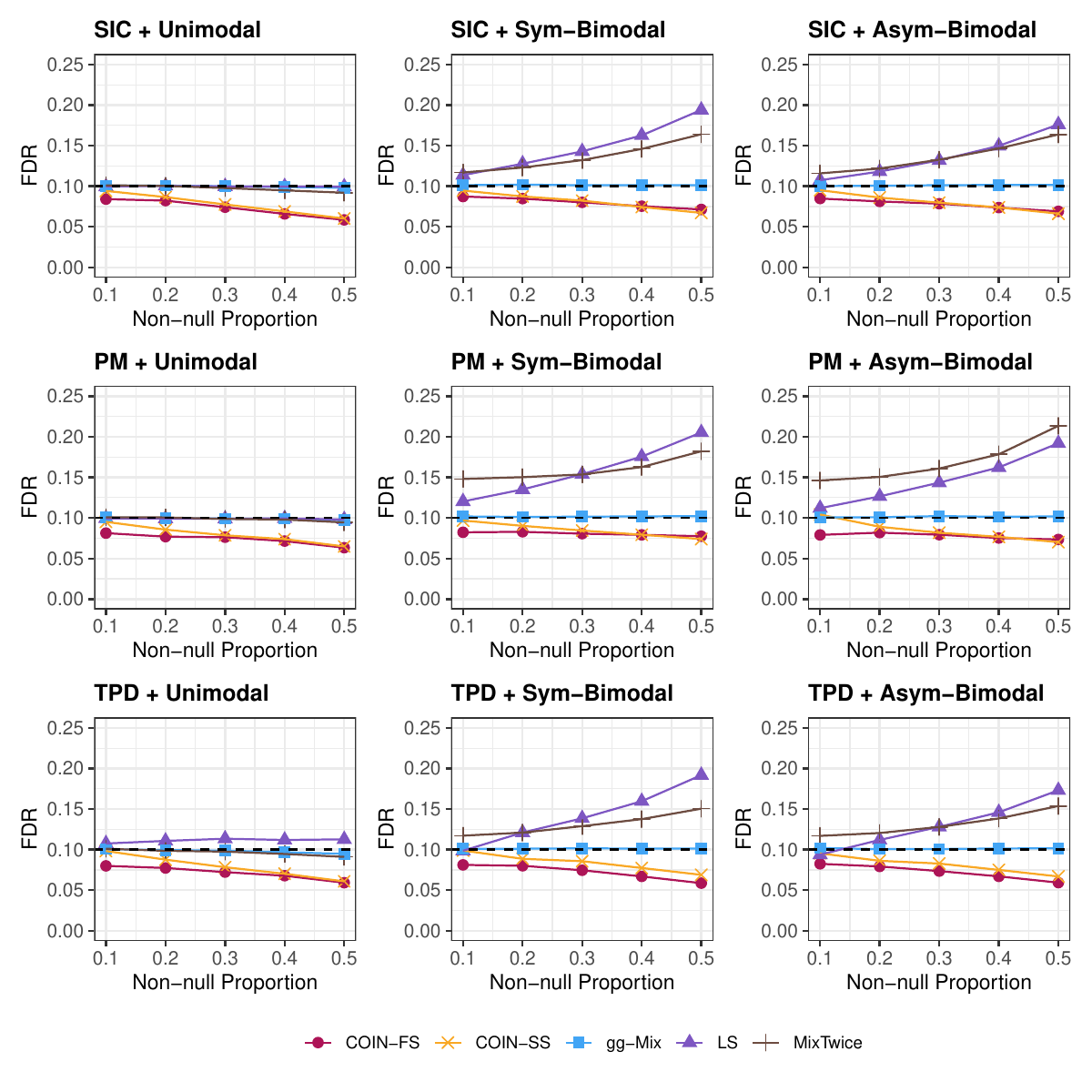}
    \caption{False discovery rates (FDRs) of five multiple testing methods---\textit{LS}, \textit{MixTwice}, \textit{gg-Mix}, \textit{COIN-SS}, and \textit{COIN-FS}---are evaluated under various prior specifications. In all simulation settings, $\mu_i$ and $\sigma_i^2$ are assumed to be independent. Each row corresponds to a different distribution for the variance $\sigma^2$, while each column represents a different distribution for the non-zero effect $\mu$. The x-axis indicates the proportion of non-null hypotheses $\pi$, and each point represents the average over 200 replications. The dashed horizontal line marks the target FDR level of 0.1.}
    \label{fig:sim1_fdr}
\end{figure}

\begin{figure}[!htb]
    \centering
    \includegraphics[width=\textwidth]{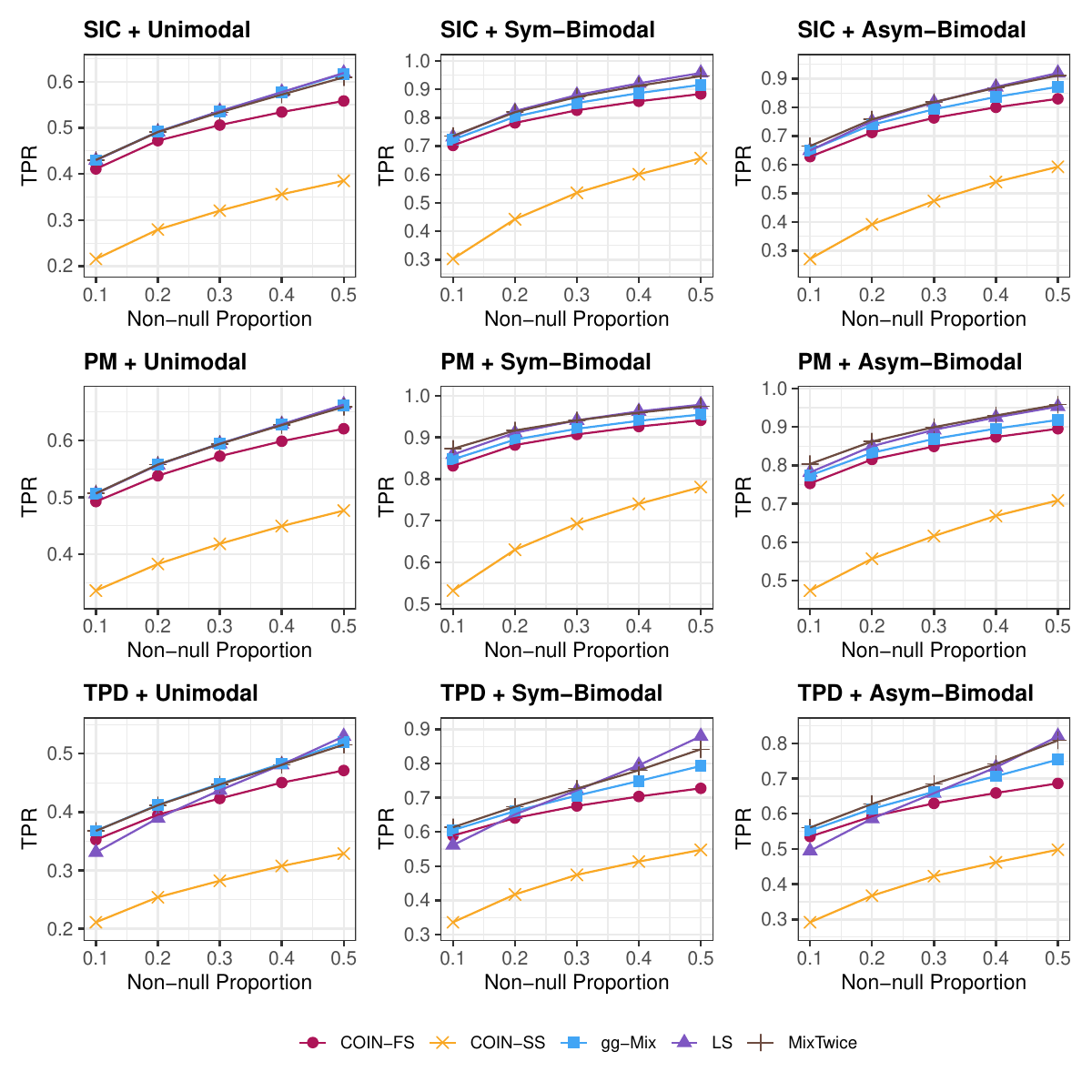}
    \caption{
    True positive rates (TPRs) of the five multiple testing methods. The simulation setup and layout are identical to those in Figure~\ref{fig:sim1_fdr}.
    }
    \label{fig:sim1_tpr}
\end{figure}

\paragraph{Simulation Results under Scenario 2.}
Figure~\ref{fig:sim2_fdr} reports the FDR, while Figure~\ref{fig:sim2_tpr} reports the TPR, across all simulation settings. It is important to note that all five methods are evaluated under prior misspecification, since $\mu_i$ and $\sigma_i^2$ are dependent in the data-generating process, whereas all methods assume independence between them. The following patterns are observed.

First, \textit{LS}, \textit{MixTwice}, and \textit{gg-Mix} often fail to control the FDR, with the extent of inflation varying across simulation settings, underscoring their sensitivity to prior misspecification. In contrast, \textit{COIN-SS} and \textit{COIN-FS} consistently achieve valid FDR control in all settings, demonstrating robustness to violations of prior assumptions. Second, although some existing methods appear to control the FDR under certain misspecified scenarios, such behavior is purely empirical and lacks theoretical justification. By comparison, the proposed methods are accompanied by formal guarantees for FDR control. Finally, \textit{COIN-FS} uniformly dominates \textit{COIN-SS} in terms of power across all simulation settings, as intended.

\begin{figure}[!htb]
    \centering
    \includegraphics[width=\textwidth]{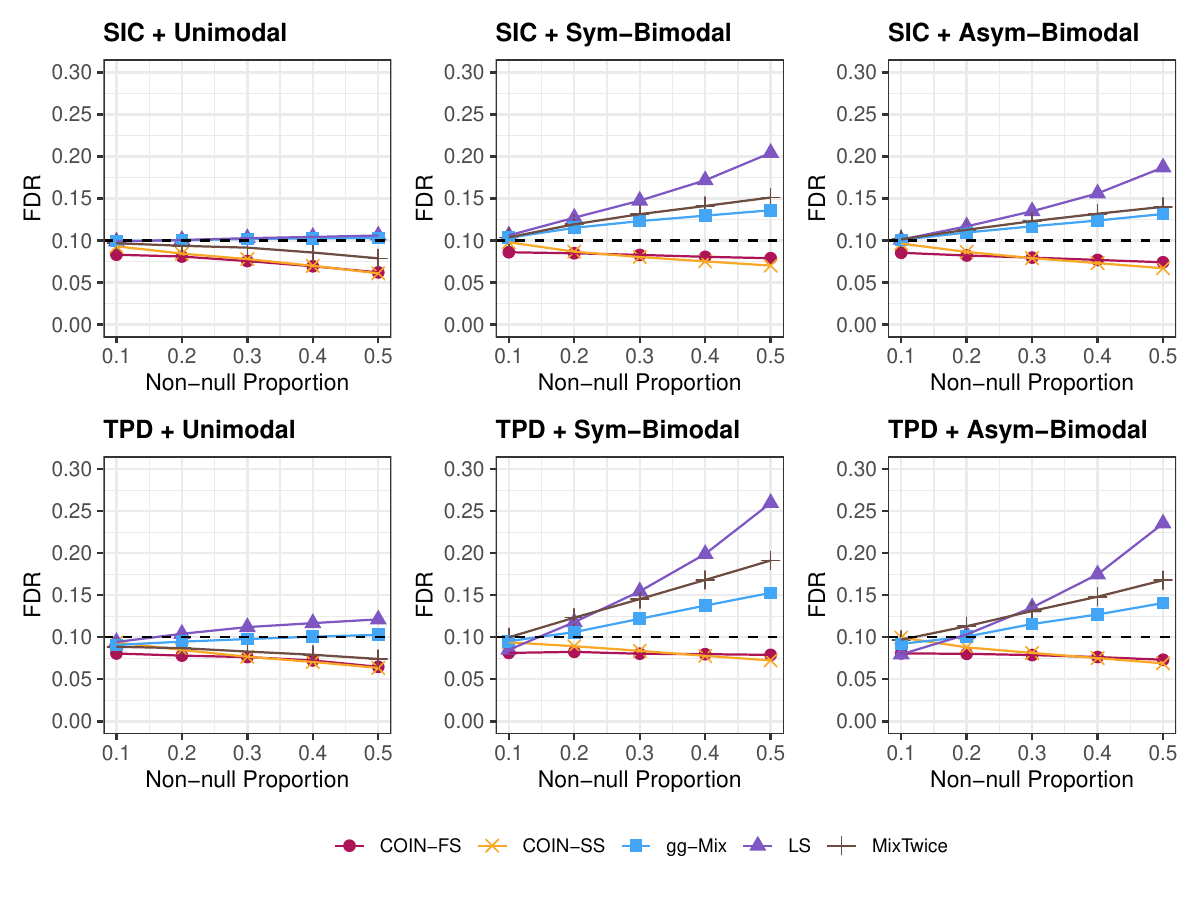}
    \caption{False discovery rates (FDRs) of five multiple testing methods---\textit{LS}, \textit{MixTwice}, \textit{gg-Mix}, \textit{COIN-SS}, and \textit{COIN-FS}---are evaluated under various prior specifications. In all simulation settings, $\mu_i$ and $\sigma_i^2$ are assumed to be dependent. Each row corresponds to a different distribution for the variance $\sigma^2$, while each column represents a different distribution for the non-zero effect $\mu$. The x-axis indicates the proportion of non-null hypotheses $\pi$, and each point represents the average over 200 replications. The dashed horizontal line marks the target FDR level of 0.1.}
    \label{fig:sim2_fdr}
\end{figure}

\begin{figure}[!htb]
    \centering
    \includegraphics[width=\textwidth]{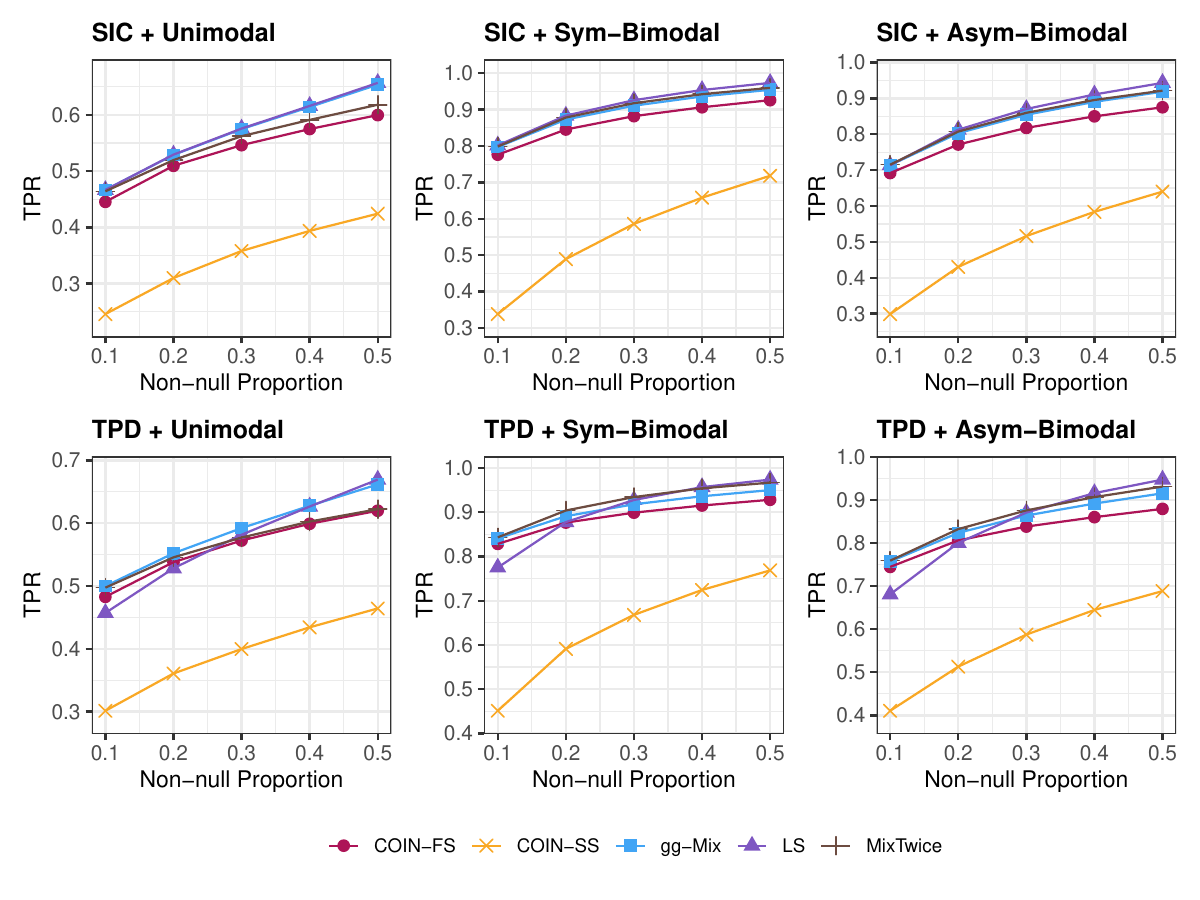}
    \caption{True positive rates (TPRs) of five multiple testing methods. The simulation setup and layout are identical to those in Figure~\ref{fig:sim2_fdr}.
    }
    \label{fig:sim2_tpr}
\end{figure}

\section{Real Data Analysis}\label{sec_6}
We evaluate the proposed methods on three real datasets. In all analyses we compare four procedures, \textit{LS}, \textit{gg-Mix}, \textit{MixTwice}, and \textit{COIN-FS}, at target FDR levels of 0.05 and 0.1. Because only summary-level test statistics are available, \textit{COIN-SS} is excluded from the comparison. We briefly summarize the datasets below; detailed descriptions are provided in Section~\ref{detail_real} of the Supplementary Material.

The first dataset consists of genome-wide DNA methylation measurements at 439,918 CpG sites from immune cell samples \citep{zhang2013genome}. The analysis aims to identify CpG sites exhibiting differential methylation between naive T-cells and antigen-activated naive T-cells. Summary statistics $(X_i, S_i^2)$ are obtained from site-wise linear regression models, with residual degrees of freedom $\nu = 4$.

The second dataset contains gene expression measurements for 48,107 genes from 12 patients with chronic lymphocytic leukemia (CLL), profiled before and after treatment with Ibrutinib \citep{dietrich2018drug}. The objective is to detect differentially expressed genes between the pre- and post-treatment conditions. Summary statistics are constructed from paired differences, with $\nu = 11$.

The third dataset comprises proteomic measurements for 6,763 proteins from 34 breast cancer tumor interstitial fluid samples across multiple experimental batches \citep{terkelsen2021high}. The analysis focuses on identifying proteins differentially expressed between luminal and Her2 subtypes. Site-wise linear regression models are fitted to obtain $(X_i, S_i^2)$, with residual degrees of freedom $\nu = 28$.

Table~\ref{tab:real_data_results} reports the number of discoveries for each method across the three real datasets at target FDR levels of 0.05 and 0.1. Several patterns emerge consistently across datasets and significance levels.
First, substantial differences are observed among \textit{LS}, \textit{MixTwice}, and \textit{gg-Mix}, with \textit{gg-Mix} uniformly producing the largest number of rejections. In light of our simulation results, where \textit{LS} and \textit{MixTwice} perform comparably to \textit{gg-Mix} under correct prior specification, these discrepancies suggest possible prior misspecification and raise concerns regarding the validity of FDR control for the former two methods. 
Second, although \textit{gg-Mix} yields the most discoveries, its validity relies on the independence assumption between $\mu_i$ and $\sigma_i^2$, which is typically unverifiable in practice; accordingly, its findings should be interpreted with caution.
Finally, \textit{COIN-FS} consistently produces the fewest discoveries. Because it maintains valid FDR control under substantially weaker prior assumptions than existing EBNMI methods, these discoveries are expected to be more reliable.

\begin{table}[!htb]
\centering
\small
\caption{Numbers of discoveries for \textit{LS}, \textit{gg-Mix}, \textit{MixTwice}, and \textit{COIN-FS} across the three real data examples at target FDR levels 0.05 and 0.1.}
\label{tab:real_data_results}
\begin{tabular}{lcccc}
\toprule
\textbf{Dataset (FDR Level)} & \textbf{\textit{LS}} & \textbf{\textit{gg-Mix}} & \textbf{\textit{MixTwice}} & \textbf{\textit{COIN-FS}} \\
\midrule
Data Example 1 (0.05) & 1,838 & 4,340 & 1,264 & 918 \\
Data Example 1 (0.1) & 7,594 & 24,583 & 2,423 & 1,524 \\
\midrule
Data Example 2 (0.05) & 105 & 130 & 77 & 0 \\
Data Example 2 (0.1) & 187 & 272 & 135 & 0 \\
\midrule
Data Example 3 (0.05) & 310 & 420 & 300 & 185 \\
Data Example 3 (0.1) & 740 & 1,019 & 738 & 230 \\
\bottomrule
\end{tabular}
\end{table}

\section{Discussion}\label{sec_7}
In this paper, we propose a new class of empirical Bayes normal mean inference methods, termed COIN, that asymptotically control the false discovery rate at the target level, even when the estimated prior distribution differs from the true one. To accommodate settings where an external training dataset is unavailable, we introduce two data-splitting strategies---sample splitting and feature splitting. We establish that both approaches achieve asymptotic FDR control at the target level regardless of such discrepancies. Extensive simulations demonstrate that the proposed methods maintain valid FDR control across a wide range of settings, including under various forms of prior misspecification. We further illustrate the practical utility of our approach through applications to real-world datasets.

A promising direction for future research is to reduce the inherent randomness in sample-splitting and feature-splitting methods. By design, these methods introduce randomness from two main sources: the random partitioning of the data and the stochastic generation of calibration variables. Such randomness can hinder reproducibility, as it may lead to different testing outcomes when the same dataset is analyzed repeatedly. Recent studies have proposed derandomization techniques in related settings. Some focus on mitigating the randomness introduced by sample splitting \citep{wasserman2020universal, dai2023false}, while others aim to reduce the variability arising from the generation of calibration variables \citep{ren2023derandomizing, ren2024derandomised, zhao2025false, zhao2025conformalized}. Extending such approaches to our framework is an important direction for future work.

\section*{Acknowledgement}
The authors declare that there are no conflicts of interest.

\bibliographystyle{apalike}
\bibliography{ref}

\clearpage
\appendix
\numberwithin{equation}{section}
\numberwithin{table}{section}
\numberwithin{figure}{section}

\numberwithin{theorem}{section}
\numberwithin{lemma}{section}
\numberwithin{proposition}{section}
\numberwithin{corollary}{section}
\begin{center}
    {\Large Supplementary Material for ``Conformalized Methods for Empirical Bayes Normal Mean Inference Problem with Heteroscedastic Variances''}
\end{center}

\section{Proofs for Section \ref{sec_3}}\label{proof_sec_3}

\subsection{Auxiliary Lemmas for the Proof of Theorem~\ref{theorem1}}
\begin{lemma}\label{lemma1}
    Suppose the hierarchical model described in \eqref{hierar_model} holds. Then, the conditional density of $X$ given $\mu = 0$ and $S^2 = s^2$ can be expressed as
    \begin{equation*}
    \begin{split}
        h_G(x) 
        \mathrel{\coloneqq}& ~p(x \mid \mu = 0, S^2 = s^2) \\
        =&  ~\dfrac{1}{\sqrt{2\pi}} \dfrac{\Gamma((\nu+1)/2)}{\Gamma(\nu/2)} \left(\dfrac{\nu}{2}\right)^{\frac{\nu}{2}} \left(\dfrac{\nu+1}{2}\right)^{-\frac{\nu+1}{2}} \\ 
        &\quad \times (s^2)^{\frac{\nu}{2}-1} \left(\dfrac{x^2 + \nu s^2}{\nu+1}\right)^{-\frac{\nu-1}{2}} \dfrac{f_G\left(\dfrac{x^2 + \nu s^2}{\nu+1}; \nu+1\right)}{f_G(s^2; \nu)},
    \end{split}
    \end{equation*}
    where $f_G(\cdot; \nu)$ denotes the marginal distribution of $s^2$, i,e., 
    \begin{equation*}
        f_G(s^2 ; \nu) \coloneqq \int p(s^2 \mid \sigma^2) ~dG(\sigma^2).
    \end{equation*}
\end{lemma}

\begin{proof} 
    By definition, the conditional density $h_G(x)$ is given by
    \begin{equation*}
    \begin{split}
        h_G (x) &= \frac{\int p(x \mid \mu = 0, \sigma^2) \cdot p(s^2 \mid \sigma^2) ~dG(\sigma^2)}{\int p(s^2 \mid \sigma^2)~dG(\sigma^2)}\\
        &= \int p(x \mid \mu = 0, \sigma^2) \cdot p(s^2 \mid \sigma^2) ~dG(\sigma^2) \cdot \frac{1}{f_G(s^2; \nu)}\\
        &= \int \frac{1}{\sqrt{2\pi \sigma^2}} e^{-\frac{1}{2\sigma^2}x^2} \cdot \frac{1}{\Gamma(\nu/2)} \left(\frac{\nu}{2\sigma^2}\right)^{\frac{\nu}{2}} (s^2)^{\frac{\nu}{2}-1} e^{-\frac{\nu}{2\sigma^2}s^2} dG(\sigma^2)\cdot \frac{1}{f_G(s^2; \nu)}\\
        &= \frac{1}{f_G(s^2;\nu)} \dfrac{\left(\frac{\nu}{2}\right)^{\frac{\nu}{2}} (s^2)^{\frac{\nu}{2}-1}}{\sqrt{2\pi} ~\Gamma(\nu/2)} \Gamma((\nu+1)/2) \left(\frac{\nu + 1}{2}\right)^{-\frac{\nu+1}{2}}\left(\frac{x^2 + \nu s^2}{\nu+1}\right)^{-\frac{\nu+1}{2} + 1}\\
        &\quad\quad \times \int \frac{1}{\Gamma((\nu+1)/2)} \left(\frac{\nu + 1}{2\sigma^2}\right)^{\frac{\nu+1}{2}} \left(\frac{x^2 + \nu s^2}{\nu+1}\right)^{\frac{\nu+1}{2} - 1} e^{-\frac{\nu+1}{2\sigma^2} \frac{x^2 + \nu s^2}{\nu + 1}}dG(\sigma^2)\\
        &= \dfrac{1}{\sqrt{2\pi}} \dfrac{\Gamma((\nu+1)/2)}{\Gamma(\nu/2)} \left(\dfrac{\nu}{2}\right)^{\frac{\nu}{2}} \left(\dfrac{\nu+1}{2}\right)^{-\frac{\nu+1}{2}} \\
        &\quad\quad \times (s^2)^{\frac{\nu}{2}-1} \left(\dfrac{x^2 + \nu s^2}{\nu+1}\right)^{-\frac{\nu-1}{2}} \dfrac{f_G\left(\dfrac{x^2 + \nu s^2}{\nu+1}; \nu+1\right)}{f_G(s^2; \nu)}
    \end{split}
    \end{equation*}
\end{proof}

\begin{lemma} \label{lemma2}
    Suppose that $\nu \geq 2$. For any positive constants $ a $ and $ b $, there exists a positive constant $ c(\nu) $, depending only on $\nu$, that satisfies the following inequality:
    \begin{equation*}
    \begin{split}
        &\int_{s^2 \leq b} \int_{|x|>a} |h_G (x)- h_{\hat{G}} (x)| ~dx\, f_G(s^2; \nu) \,ds^2 \\
        &\quad  \leq 2\left\{c(\nu) ~\frac{b}{a}~ ||f_G(\cdot ; \nu+1) -f_{\hat{G}}(\cdot ; \nu+1)||_{L_2} + ||f_G(\cdot ; \nu) -f_{\hat{G}}(\cdot ; \nu)||_{L_1}\right\}.
    \end{split}
    \end{equation*}
\end{lemma}

\begin{proof}
    Throughout this proof, we simplify notation by writing $ f_G $ and $ h_G $ (resp. $f_{\hat{G}}$ and $h_{\hat{G}}$) instead of $ f_G(s^2; \nu) $ and $ h_G(x) $ (resp. $f_{\hat{G}}(s^2; \nu)$ and $h_{\hat{G}}(x)$), respectively. 
    
    Define $f_G^* \coloneqq \frac{f_G + f_{\hat{G}}}{2}$. Then, we have
    \begin{equation*}
    \begin{split}
        &\int_{s^2 \leq b}\int_{|x|>a} |h_G - h_{\hat{G}}|\,dx f_G \,ds^2 \\
        &\quad= \int_{s^2 \leq b}\int_{|x|>a}
        \left|\frac{h_G f_G}{f_G} - \frac{h_G f_G}{f_G^*}
        + \frac{h_G f_G}{f_G^*} - \frac{h_{\hat{G}} f_{\hat{G}}}{f_G^*}
        + \frac{h_{\hat{G}} f_{\hat{G}}}{f_G^*} - \frac{h_{\hat{G}} f_{\hat{G}}}{f_{\hat{G}}}\right|\,dx\,f_G\,ds^2.\\
        &\quad\overset{(\star)}{\leq}  \int_{s^2 \leq b}\int_{|x|>a}
        \frac{|h_Gf_G - h_{\hat{G}}f_{\hat{G}}|}{f_G^*} 
        + h_G \left|1-\frac{f_G}{f_G^*}\right|
        + h_{\hat{G}} \left|1-\frac{f_{\hat{G}}}{f_G^*}\right| \,dx\, f_G \,ds^2.\\
        &\quad\overset{(\star\star)}{\leq}  \int_{s^2 \leq b}\int_{|x|>a}
        2|h_Gf_G - h_{\hat{G}}f_{\hat{G}}| + |f_G - f_{\hat{G}}|(h_G + h_{\hat{G}}) \,dx \,ds^2.\\
        &\quad\overset{(\star\star\star)}{\leq} 2 \int_{s^2 \leq b}\underbrace{\int_{|x|>a} |h_Gf_G - h_{\hat{G}}f_{\hat{G}}| \, dx}_{\eqqcolon  A} \, ds^2 + 2 ||f_G -f_{\hat{G}}||_{L_1}
    \end{split}
    \end{equation*}
    The first inequality (marked by $(\star)$) follows from the triangle inequality; the second inequality (marked by $(\star\star)$) follows from the definition of $f_G^*$ and the fact that $f_G^* \geq f_G/2$; and the last inequality (marked by $(\star\star\star)$) follows from the fact that $h_G$ and $h_{\hat{G}}$ are densities.

    Now, let consider the term A. Using Lemma~\ref{lemma1}, change-of-variable technique, Hölder's inequality, and $\nu \geq 2$, we obtain
    \begin{equation*}
    \begin{split}
        A 
        &= \int_{|x|>a} C(\nu) (s^2)^{\frac{\nu}{2} -1} \left(\frac{x^2 + \nu s^2}{\nu+1}\right)^{-\frac{\nu-1}{2}} \\
        &\quad\quad \times\left|f_G\left(\frac{x^2 + \nu s^2}{\nu+1};\nu+1\right) - f_{\hat{G}}\left(\frac{x^2 + \nu s^2}{\nu+1}; \nu+1\right)\right|\,dx\\
        &= C^\prime(\nu) (s^2)^{\frac{\nu}{2} -1}  \int \dfrac{(t^2)^{-\frac{\nu-1}{2}}}{\sqrt{(\nu+1)t^2 - \nu s^2}} \mathbb{I}\left(t^2 \geq \frac{a^2 + \nu s^2}{\nu+1}\right) \\
        &\quad\quad \times|f_G(t^2; \nu+1) - f_{\hat{G}}(t^2; \nu+1)|\, dt^2\\
        &\leq  C^\prime(\nu) (s^2)^{\frac{\nu}{2} -1} \left[\int \dfrac{(t^2)^{-(\nu-1)}}{(\nu+1)t^2 - \nu s^2} \mathbb{I}\left(t^2 \geq \frac{a^2 + \nu s^2}{\nu+1}\right)\, dt^2\right]^{\frac{1}{2}}\\
        &\quad\quad \times \left[\int |f_G(t^2; \nu+1) - f_{\hat{G}}(t^2; \nu+1)|^2\, dt^2\right]^{\frac{1}{2}}\\
        &\leq  C^{\prime\prime}(\nu) (s^2)^{\frac{\nu}{2} -1}  \left[ \frac{1}{a^2(a^2 + \nu s^2)^{\nu-2}}\right]^{\frac{1}{2}} \left[\int |f_G(t^2; \nu+1) - f_{\hat{G}}(t^2; \nu+1)|^2\, dt^2\right]^{\frac{1}{2}}\\
        &\leq C^{\prime\prime}(\nu) ~\dfrac{1}{a}~  ||f_G(t^2; \nu+1) - f_{\hat{G}}(t^2; \nu+1)||_{L_2},
    \end{split}
    \end{equation*}
    where $ C(\nu) $, $ C'(\nu) $, and $ C''(\nu) $ are positive constants depending only on $ \nu $. It is important to note that the last term does not depends on the value of $s^2$. By combining all the results established above, we obtain the desired result.
\end{proof}

\begin{lemma}\label{lemma3}
    Suppose that Assumption~\ref{assum_1} holds and $\nu \geq 2$. Given the dataset $\mathcal{D}_2$ for $m^{\mathcal{D}_2} \geq e$,
    \begin{equation*}
    \begin{split}
        &\int \int |h_G(x) - h_{\hat{G}}(x)|~dx f_G(s^2; \nu)\,ds^2 \\
        &\quad= 2 \mathbb{I}\left(H(f_G(\cdot; \nu), f_{\hat{G}}(\cdot ; \nu)) \geq C \frac{\log{m^{\mathcal{D}_2}}}{\sqrt{m^{\mathcal{D}_2}}}\right) + O\left(\left(\frac{(\log{m^{\mathcal{D}_2}})^5}{m^{\mathcal{D}_2}}\right)^{\frac{1}{4}}\right),
    \end{split}
    \end{equation*}
    where $H(f, g)$ denotes the Hellinger distance between $f$ and $g$ defined as 
    \begin{equation*}
        H(f, g) = \left(\frac{1}{2} \int \left(\sqrt{f(x)} - \sqrt{g(x)}\right)^2\;dx\right)^{\frac{1}{2}},
    \end{equation*}
    and the constant $C = C(\nu, \underaccent{\bar}{L}, \bar{U})$ is defined as in Theorem 9 of \cite{ignatiadis2025empirical}.
\end{lemma}

\begin{proof}
    Throughout this proof, we assume that the training dataset $\mathcal{D}_2$ is given, which implies that the estimated prior distribution $\hat{G}$ is a deterministic function.
    
    We begin by introducing the following event:    
    \begin{equation*}
    E = \left\{ H(f_G(\cdot; \nu), f_{\hat{G}}(\cdot; \nu)) \geq C \frac{\log{m^{\mathcal{D}_2}}}{\sqrt{m^{\mathcal{D}_2}}} \right\}.
    \end{equation*}
    Using this event, we decompose the integral as follows:
    \begin{equation*}
    \begin{split}
        &\int \int |h_G(x) - h_{\hat{G}}(x)|\,dx\, f_G(s^2;\nu)\, ds^2 \\
        &\quad\leq \int \int |h_G(x) - h_{\hat{G}}(x)|\,dx\, f_G(s^2;\nu)\, ds^2 \cdot \mathbb{I}(E)\\
        &\quad\quad\quad + \int_{s^2 > b} \int |h_G(x) - h_{\hat{G}}(x)|\,dx\, f_G(s^2;\nu)\, ds^2 \cdot \mathbb{I}(E^c) \\
        &\quad\quad\quad + \int_{s^2 \leq b} \int_{|x| \leq a} |h_G(x) - h_{\hat{G}}(x)|\,dx\, f_G(s^2;\nu)\, ds^2 \cdot \mathbb{I}(E^c)\\
        &\quad\quad\quad + \int_{s^2 \leq b} \int_{|x| > a} |h_G(x) - h_{\hat{G}}(x)|\,dx\, f_G(s^2;\nu)\, ds^2 \cdot \mathbb{I}(E^c)\\
        &\quad\overset{(\star)}{\leq} 2 \left\{ \mathbb{I}(E) + \mathbb{P}(S^2 > b) + \underset{|x| \leq a}{\max} |h_G(x) - h_{\hat{G}}(x)| \,a \right. \\
        &\quad\quad\quad + \left. c(\nu) \frac{b}{a} \, \| f_G(\cdot ; \nu+1) - f_{\hat{G}}(\cdot ; \nu+1) \|_{L_2} \cdot \mathbb{I}(E^c) + \| f_G(\cdot ; \nu) - f_{\hat{G}}(\cdot ; \nu) \|_{L_1} \cdot \mathbb{I}(E^c) \right\}.
    \end{split}
    \end{equation*}
    The last inequality (marked by $(\star)$) follows from the fact that $ h_G $, $ h_{\hat{G}} $, and $f_G$ are density functions, as well as Lemma \ref{lemma2}.

    We choose the following values for $a$ and $b$:
    \begin{equation*}
        a = \left(\frac{(\log{m^{\mathcal{D}_2}})^5}{m^{\mathcal{D}_2}}\right)^{\frac{1}{4}}, \quad b = 4U\left|\log{\left(\frac{1}{m^{\mathcal{D}_2}}\right)}\right|.
    \end{equation*}
    Then, by using Lemma 12, Lemma S1 and Lemma S6 in \cite{ignatiadis2025empirical}, we establish the following bounds: For all $m^{\mathcal{D}_2} \geq e,$
    \begin{enumerate}[label=\roman*.]
        \item $ \mathbb{P}(S^2 > b) \leq \frac{1}{m^{\mathcal{D}_2}}$ 
        \item $ \underset{|x| \leq a}{\max} ~|h_G(x) - h_{\hat{G}}(x)|~ a = O\left(\left(\frac{(\log{m^{\mathcal{D}_2}})^5}{m^{\mathcal{D}_2}}\right)^{\frac{1}{4}} \right)$.
        \item $ c(\nu) ~\dfrac{b}{a}~ \| f_G(\cdot ; \nu+1) - f_{\hat{G}}(\cdot ; \nu+1) \|_{L_2} \cdot \mathbb{I}(E^c) = O\left(\left(\frac{(\log{m^{\mathcal{D}_2}})^5}{m^{\mathcal{D}_2}}\right)^{\frac{1}{4}}\right) $.
        \item $ \| f_G(\cdot ; \nu) - f_{\hat{G}}(\cdot ; \nu) \|_{L_1} \cdot \mathbb{I}(E^c) = O\left(\frac{\log{m^{\mathcal{D}_2}}}{\sqrt{m^{\mathcal{D}_2}}}\right)$.
    \end{enumerate}
    By combining all the results established above, we obtain the desired result.
\end{proof}

\begin{lemma}\label{lemma4}
    Assume that there are no ties between $ u_i $ and $ \tilde{u}_i $ for all $ i \in \mathcal{H}^{\mathcal{D}_1} $. Given the dataset $\mathcal{D}_2$, for any Borel measurable set $E$ with positive measure and $i \in \mathcal{H}_0^{\mathcal{D}_1}$, we have
    \begin{equation*}
    \begin{split}
        &\left|\left\{\mathbb{P}(\text{sgn}(u_i - \tilde{u}_i) = 1 \mid u_i\wedge \tilde{u}_i \in E) - \frac{1}{2}\right\} \cdot \mathbb{P}(u_i \wedge \tilde{u}_i \in E) \right| \\
        &\quad\leq 2\int\int |h_G(x) - h_{\hat{G}}(x)|\,dx\, f_G(s^2; \nu)\,ds^2 
    \end{split}
    \end{equation*}
\end{lemma}

\begin{proof}
    Throughout this proof, we assume that the dataset $\mathcal{D}_2$ is given, which implies that both the estimated prior $\hat{G}$ and the estimated conformity score function $\hat{u}(\cdot, \cdot)$ are deterministic.

    For each $i \in \mathcal{H}^{\mathcal{D}_1}$, let $\tilde{X}_{i,1}^*$ denote the oracle calibration variable, which satisfies the conditional exchangeability property stated in \eqref{cond_exch}. Based on this oracle calibration variable, we define the corresponding conformity score:
    \begin{equation*}
        \tilde{u}_i^* = \hat{u}(\tilde{X}_{i,1}^*, S_{i,1}^2).    
    \end{equation*}
    By construction, the pair $(u_i, \tilde{u}_i^*)$ is exchangeable for each $i \in \mathcal{H}_0^{\mathcal{D}_1}$ (since the conformity score function is fixed and $\tilde{X}_{i,1}^*$ and $X_{i,1}$ are conditionally exchangeable given $S_{i,1}^2$). Under the additional assumption that there are no ties between $u_i$ and $\tilde{u}_i^*$, the exchangeability of $(u_i, \tilde{u}_i^*)$ implies that, for any Borel measurable set $E$ with positive measure,
    \begin{equation*}
        \mathbb{P}\left( \operatorname{sgn}(u_i - \tilde{u}_i^*) = 1 \,\middle|\, u_i \wedge \tilde{u}_i^* \in E \right) = \frac{1}{2}, \quad \forall i \in \mathcal{H}_0^{\mathcal{D}_1}.
    \end{equation*}
    This fact allows us to establish the key bound:
    \begin{equation}\label{eqn:lemma4_1}
    \begin{split}
        &\left|\left\{ \mathbb{P}(\text{sgn}(u_i - \tilde{u}_i) = 1 \mid u_i \wedge \tilde{u}_i \in E) - 1/2 \right\} \cdot \mathbb{P}(u_i \wedge \tilde{u}_i \in E)\right| \\
        &\quad = \left|\{\mathbb{P}(\text{sgn}(u_i - \tilde{u}_i) = 1, u_i \wedge \tilde{u}_i \in E) -\mathbb{P}(\text{sgn}(u_i - \tilde{u}_i^*) = 1, u_i \wedge \tilde{u}_i^* \in E)\right. \\
        &\quad\quad\quad\left. +~ \mathbb{P}(\text{sgn}(u_i - \tilde{u}_i^*) = 1, u_i \wedge \tilde{u}_i^* \in E) - 1/2 \cdot \mathbb{P}(u_i \wedge \tilde{u}_i \in E)  \}\right|\\
        &\quad\overset{(\star)}{\leq} \left| \mathbb{P}(\text{sgn}(u_i - \tilde{u}_i) = 1, u_i \wedge \tilde{u}_i \in E) - \mathbb{P}(\text{sgn}(u_i - \tilde{u}_i^*) = 1, u_i \wedge \tilde{u}_i^* \in E) \right| \\
        &\quad\quad\quad + 1/2 \left| \mathbb{P}(u_i \wedge \tilde{u}_i \in E) - p(u_i \wedge \tilde{u}_i^* \in E) \right| \\
        &\quad\overset{(\star\star)}{\leq} 3/2 \left| \mathbb{P}(\text{sgn}(u_i - \tilde{u}_i) = 1, u_i \wedge \tilde{u}_i \in E) - \mathbb{P}(\text{sgn}(u_i - \tilde{u}_i^*) = 1, u_i \wedge \tilde{u}_i^* \in E) \right| \\
        &\quad\quad\quad + 1/2 \left| \mathbb{P}(\text{sgn}(u_i - \tilde{u}_i) = -1, u_i \wedge \tilde{u}_i \in E) - \mathbb{P}(\text{sgn}(u_i - \tilde{u}_i^*) = -1, u_i \wedge \tilde{u}_i^* \in E) \right|,
    \end{split}
    \end{equation}
    where the first inequality (marked by $(\star)$) is a direct consequence of the triangle inequality, and the second inequality (marked by $(\star\star)$) follows from the absence of ties between $u_i$ and $\tilde{u}_i$ (or $\tilde{u}_i^*$), together with another application of the triangle inequality.

    We next analyze the joint probabilities appearing on the right-hand side. First, observe that
    \begin{equation*}
    \begin{split}
        \mathbb{P}(\text{sgn}(u_i - \tilde{u}_i) = 1, u_i \wedge \tilde{u}_i \in E) 
        &= \mathbb{P}(\text{sgn}(u_i - \tilde{u}_i) = 1, \tilde{u}_i \in E)\\
        &= \mathbb{E}\left[\mathbb{P}(\text{sgn}(u_i - \tilde{u}_i) = 1, \tilde{u}_i \in E\mid S_{i,1}^2)\right]\\
        &= \int\int\int_{B(y, s^2, E)} h_{\hat{G}}(x) \,dx\, h_G(y) \,dy\, f_G(s^2;\nu) \,ds^2,
        \end{split}
    \end{equation*}
    where $B(y, s^2, E)$ denotes a Borel measurable set that depends on $y$, $s^2$ and $E$. Similarly,
    \begin{equation*}
    \begin{split}
        \mathbb{P}(\text{sgn}(u_i - \tilde{u}_i) = 1, u_i \wedge \tilde{u}_i^* \in E) 
        &= \mathbb{P}(\text{sgn}(u_i - \tilde{u}_i) = 1, \tilde{u}_i^* \in E)\\
        &= \mathbb{E}[\mathbb{P}(\text{sgn}(u_i - \tilde{u}_i^*) = 1, \tilde{u}_i^* \in E\mid S_{i,1}^2)]\\
        &= \int\int\int_{B(y, s^2, E)} h_G(x) \,dx\, h_G(y) \,dy\, f_G(s^2;\nu) \,ds^2.
        \end{split}
    \end{equation*}
    These identities immediately imply that
    \begin{equation}\label{eqn:lemma4_2}
    \begin{split}
        &\left| \mathbb{P}(\text{sgn}(u_i - \tilde{u}_i) = 1, u_i \wedge \tilde{u}_i \in E) - \mathbb{P}(\text{sgn}(u_i - \tilde{u}_i^*) = 1, u_i \wedge \tilde{u}_i^* \in E) \right| \\
        &\quad \leq \int\int\int |h_G(x) - h_{\hat{G}}(x)|\,dx\, h_G(y) \,dy\, f_G(s^2;\nu) \,ds^2 \\
        &\quad = \int\int |h_G(x) - h_{\hat{G}}(x)|\,dx\, f_G(s^2;\nu) \,ds^2.
    \end{split}
    \end{equation}
    By an entirely similar argument, we also obtain
    \begin{equation}\label{eqn:lemma4_3}
    \begin{split}
        &\left| \mathbb{P}(\text{sgn}(u_i - \tilde{u}_i) = -1, u_i \wedge \tilde{u}_i \in E) - \mathbb{P}(\text{sgn}(u_i - \tilde{u}_i^*) = -1, u_i \wedge \tilde{u}_i^* \in E) \right| \\
        &\quad\leq \int\int |h_G(x) - h_{\hat{G}}(x)|\,dx\, f_G(s^2;\nu) \,ds^2.
    \end{split}
    \end{equation}
    Combining inequalities \eqref{eqn:lemma4_1}--\eqref{eqn:lemma4_3} completes the proof.
\end{proof}

\begin{lemma}\label{lemma5}
    Suppose the hierarchical model described in \eqref{hierar_model} holds, and assume that there are no ties between $ u_i $ and $ \tilde{u}_i $ for all $ i \in \mathcal{H}^{\mathcal{D}_1} $. Given the dataset $ \mathcal{D}_2 $, we have
    \begin{equation*}
        \left(\{1+\text{sgn}(u_i - \tilde{u}_i)\}/2 \mid u_k \wedge \tilde{u}_k :  k \in \mathcal{H}^{\mathcal{D}_1}\right) \overset{\text{i.i.d.}}{\sim} \text{Bernoulli}(\rho), \quad \forall i \in \mathcal{H}_0^{\mathcal{D}_1},
    \end{equation*}
    where $ \text{sgn}(\cdot) $ denotes the sign function, and $ \rho $ is a random variable defined as 
    \begin{equation*}
        \rho \coloneqq \mathbb{P}\left(\text{sgn}(u_i - \tilde{u}_i) = 1 \mid u_k \wedge \tilde{u}_k : k \in \mathcal{H}^{\mathcal{D}_1}\right) \quad \text{for} \quad i \in \mathcal{H}_0^{\mathcal{D}_1}.
    \end{equation*}
\end{lemma}

\begin{proof}
    Throughout this proof, we assume that the dataset $\mathcal{D}_2$ is given, which implies that both the estimated prior distribution $\hat{G}$ and the conformity score function $\hat{u}(\cdot, \cdot)$ are deterministic. Thus, all sources of randomness are due to the collection of triplets $\{(X_{i,1}, \tilde{X}_{i,1}, S_{i,1}^2)\}_{i \in \mathcal{H}^{\mathcal{D}_1}}$.
     
    It is also worth noting that, under the hierarchical model in \eqref{hierar_model}, the following properties hold due to the construction of the calibration variable $\tilde{X}_{i,1}$ and the definition of the conformity score:
    \begin{enumerate}[label=(\roman*)]
        \item The triplets $ \{(X_{i,1}, \tilde{X}_{i,1}, S_{i,1}^2)\}_{i \in \mathcal{H}^{\mathcal{D}_1}} $ are mutually independent. \label{prop:indep_all}
        \item The triplets $\{(X_{i,1}, \tilde{X}_{i,1}, S_{i,1}^2)\}_{i \in \mathcal{H}_0^{\mathcal{D}_1}}$ are i.i.d.. \label{prop:iid_null}
        \item Both $ \text{sgn}(u_i - \tilde{u}_i) $ and $ u_i \wedge \tilde{u}_i $ are functions of the triplet $ (X_{i,1}, \tilde{X}_{i,1}, S_{i,1}^2) $. \label{prop:func}
    \end{enumerate}
    Leveraging properties~\ref{prop:indep_all} and~\ref{prop:func}, we establish the following conditional independence relation for each $i \in \mathcal{H}^{\mathcal{D}_1}$:
    \begin{equation*}
        \mathbb{P}\left(\text{sgn}(u_i - \tilde{u}_i) = 1 \,\middle|\, u_k \wedge \tilde{u}_k : k \in \mathcal{H}^{\mathcal{D}_1}\right)
        = \mathbb{P}\left(\text{sgn}(u_i - \tilde{u}_i) = 1 \,\middle|\, u_i \wedge \tilde{u}_i\right).
    \end{equation*}
    Combining this with property~\ref{prop:iid_null}, we conclude that
    \begin{equation*}
        \left\{\mathbb{P}\left(\text{sgn}(u_i - \tilde{u}_i) = 1 \,\middle|\, u_k \wedge \tilde{u}_k : k \in \mathcal{H}^{\mathcal{D}_1}\right)\right\}_{ i \in \mathcal{H}_0^{\mathcal{D}_1}} \text{ are i.i.d.}.
    \end{equation*}
    Since there are no ties between $ u_i $ and $ \tilde{u}_i $, for all $ i \in \mathcal{H}_0^{\mathcal{D}_1} $,
    \begin{equation*}
        \left(\{1 + \text{sgn}(u_i - \tilde{u}_i)\}/2 \,\middle|\, u_k \wedge \tilde{u}_k : k \in \mathcal{H}^{\mathcal{D}_1}  \right)
        \overset{\text{i.i.d.}}{\sim} \text{Bernoulli}(\rho).
    \end{equation*}
\end{proof}

Before proceeding to the next result, we introduce the relevant notation. Let $ \xi_i = \mathbb{I}(u_i \leq \tilde{u}_i) $, $ s_i = u_i \wedge \tilde{u}_i $, and let $ s_{(i)} $ denote the $ i $-th smallest value among $ \{s_j\}_{j \in \mathcal{H}^{\mathcal{D}_1}} $. With these notations, define the cumulative counts denoted by $V_i$ and $\tilde{V}_i$:
\begin{equation*}
    V_i = \sum_{j \in \mathcal{H}_0^{\mathcal{D}_1}} \xi_j \cdot \mathbb{I}(s_j \leq s_{(i)}), \quad 
    \tilde{V}_i = \sum_{j \in \mathcal{H}_0^{\mathcal{D}_1}} (1 - \xi_j) \cdot \mathbb{I}(s_j \leq s_{(i)}).
\end{equation*}
Furthermore, for a target FDR level $\alpha$, define
\begin{equation}\label{k_hat}
    \hat{k} = \hat{k}(\alpha) \coloneqq \max \left\{k \in [m^{\mathcal{D}_1}] : \frac{1 + \sum_{i \in \mathcal{H}^{\mathcal{D}_1}} (1 - \xi_i) \cdot \mathbb{I}(s_i \leq s_{(k)})}{\sum_{i \in \mathcal{H}^{\mathcal{D}_1}} \xi_i \cdot \mathbb{I}(s_i \leq s_{(k)})} \leq \alpha \right\}.
\end{equation}
    
\begin{lemma}\label{lemma6}
    Suppose the hierarchical model described in \eqref{hierar_model} holds, and assume that there are no ties between $ u_i $ and $ \tilde{u}_i $ for all $ i \in \mathcal{H}^{\mathcal{D}_1} $. Given the dataset $ \mathcal{D}_2 $, the following inequality holds:
    \begin{equation*}
        \mathbb{E}\left[ \frac{V_{\hat{k}}}{M + \tilde{V}_{\hat{k}}} \right] \leq 1,
    \end{equation*}
    where  $M = \left(1 + \left(\frac{1 - 2\rho - (1-\rho)\rho^{m_0^{\mathcal{D}_1}}}{\rho}\right) m_0^{\mathcal{D}_1} \right) \vee 1$ and $\rho$ is a random variable defined in Lemma~\ref{lemma5}.
\end{lemma}

\begin{proof}
    Throughout this proof, we assume that the dataset $\mathcal{D}_2$ is given, so that both the estimated prior distribution $\hat{G}$ and the conformity score function $\hat{u}(\cdot, \cdot)$ are deterministic. 
    
    We claim that the process $\{M_i\}$, defined as
    \begin{equation*}
        M_i = \frac{V_i}{M + \tilde{V}_i},
    \end{equation*}
    is a time-reversible supermartingale with respect to the filtration $\{\mathcal{F}_i\}$, where
    \begin{equation*}
        \mathcal{F}_i = \sigma\left(\{\xi_j : j \in \mathcal{H}_1^{\mathcal{D}_1}\},\ \{s_j : j \in \mathcal{H}^{\mathcal{D}_1}\},\ \{V_j, \tilde{V}_j : i \leq j \leq m^{\mathcal{D}_1}\} \right).
    \end{equation*}
    By the definition of $\rho$, $ M_i $ and $ \mathcal{F}_i $, it follows that $ M_i $ is $ \mathcal{F}_i $-measurable for all $ i \in \mathcal{H}^{\mathcal{D}_1} $, i.e., $ M_i \in \mathcal{F}_i $ for all $ i \in \mathcal{H}^{\mathcal{D}_1} $. Thus, it suffices to show that $ \mathbb{E}(M_i \mid \mathcal{F}_{i+1}) \leq M_{i+1} $. Consider the following cases: (i) If $ \tilde{V}_{i+1} = 0 $, then by the non-negativity and monotonicity of $ V_i $ and $ \tilde{V}_i $, we have $ M_i \leq M_{i+1} $. (ii) If $ \tilde{V}_{i+1} \neq 0 $ and the original index corresponding to $ s_{(i+1)} $ belongs to $ \mathcal{H}_1^{\mathcal{D}_1} $, then since $ V_i $ and $ \tilde{V}_i $ depend only on indices in $ \mathcal{H}_0^{\mathcal{D}_1} $, it follows that $ V_i = V_{i+1} $ and $ \tilde{V}_i = \tilde{V}_{i+1} $, which implies $ M_i = M_{i+1} $. (iii) If $ \tilde{V}_{i+1} \neq 0 $ and the original index corresponding to $ s_{(i+1)} $ belongs to $ \mathcal{H}_0^{\mathcal{D}_1} $, then
    \begin{equation*}
    \begin{split}
        \mathbb{E}(M_i \mid \mathcal{F}_{i+1}) 
        &= \mathbb{E}(M_i \cdot \mathbb{I}(V_i = V_{i+1}, \tilde{V}_i = \tilde{V}_{i+1}-1) \mid \mathcal{F}_{i+1}) \\
        &\quad\quad + ~\mathbb{E}(M_i \cdot \mathbb{I}(V_i = V_{i+1} -1, \tilde{V}_i = \tilde{V}_{i+1}) \mid \mathcal{F}_{i+1})\\
        &= \frac{V_{i+1}}{M + \tilde{V}_{i+1}-1} \cdot \mathbb{P}(V_i = V_{i+1}, \tilde{V}_i = \tilde{V}_{i+1}-1 \mid \mathcal{F}_{i+1}) \\
        &\quad\quad +~\frac{V_{i+1}-1}{M + \tilde{V}_{i+1}}\cdot \mathbb{P}(V_i = V_{i+1} -1, \tilde{V}_i = \tilde{V}_{i+1} \mid \mathcal{F}_{i+1})\\
        &\overset{(\star)}{=} \frac{V_{i+1}}{M + \tilde{V}_{i+1}-1} \cdot \frac{\tilde{V}_{i+1}}{ V_{i+1} +  \tilde{V}_{i+1}} + \frac{V_{i+1}-1}{M + \tilde{V}_{i+1}}\cdot  \frac{ V_{i+1}}{ V_{i+1} + \tilde{V}_{i+1}}\\
        &\overset{(\star\star)}{=} M_{i+1} \cdot \left\{\frac{(M + \tilde{V}_{i+1})  \tilde{V}_{i+1} + (M + \tilde{V}_{i+1}-1)(V_{i+1}-1)}{(M+\tilde{V}_{i+1}-1)\{ V_{i+1} + \tilde{V}_{i+1}\}}\right\}.
    \end{split}
    \end{equation*}
    The third equality (marked by $(\star)$) is justified by the exchangeability of the triplets $\{(X_{i,1}, \tilde{X}_{i,1}, S_{i,1}^2)\}_{i \in \mathcal{H}_0^{\mathcal{D}_1}}$, which follows from the fact that they are i.i.d. Since the second term in the last equality (marked by $(\star\star)$) is at most 1 by the definition of $M$, it follows that $\mathbb{E}(M_i \mid \mathcal{F}_{i+1}) \leq M_{i+1}$, thereby establishing the first claim.

    Next, we claim that $\hat{k}$ is a stopping time with respect to the filtration $\{\mathcal{F}_i\}$. By definition, the event $\{\hat{k} = i\}$ can be written as
    \begin{equation*}
        \{\hat{k} = i\} = E_1(i) \cap E_2(i),    
    \end{equation*}
    where
    \begin{equation*}
        E_1(i) = \{Q(s_{(i)}) \leq \alpha\}, \quad
        E_2(i) = \{Q(s_{(j)}) > \alpha \text{ for all } j > i\}.    
    \end{equation*}
    Since $ Q(s_{(j)}) $ can be expressed as
    \begin{equation*}
    \begin{split}
        Q(s_{(k)}) 
        &= \frac{1 + \tilde{V}_k + \sum_{i \in \mathcal{H}_1^{\mathcal{D}_1}}(1-\xi_i) I(s_i \leq s_{(k)})}{[V_{k} + \sum_{i \in \mathcal{H}_1^{\mathcal{D}_1}} \xi_i I(s_i \leq s_{(k)})] \vee 1},
        \end{split}
    \end{equation*}
    It follows that both $E_1(i)$ and $E_2(i)$ are $\mathcal{F}_i$-measurable, and thus $\{\hat{k} = i\} \in \mathcal{F}_i$, which implies that $\hat{k}$ is a stopping time.

    Up to this point, we have established that the sequence $ \{M_i\} $ is a time-reversible supermartingale, and that $ \hat{k} $ is a stopping time with respect to the filtration $ \{\mathcal{F}_i\} $. Then, by the optional stopping theorem, we obtain
    \begin{equation}\label{eqn:lemma2_1}
    \begin{split}
        \mathbb{E}(M_{\hat{k}}) 
        \leq \mathbb{E}(M_m) 
        &= \mathbb{E}\left[\frac{\sum_{i \in \mathcal{H}_0^{\mathcal{D}_1}} \xi_i}{M + \sum_{i \in \mathcal{H}_0^{\mathcal{D}_1}} (1-\xi_i)}\right] \\
        &= \mathbb{E}\left[ \mathbb{E}\left[\frac{\sum_{i \in \mathcal{H}_0^{\mathcal{D}_1}} \xi_i}{M + \sum_{i \in \mathcal{H}_0^{\mathcal{D}_1}} (1-\xi_i)} \Bigg| s_i : i \in \mathcal{H}^{\mathcal{D}_1}\right]\right].
    \end{split}
    \end{equation}
    Let $ B = \sum_{i \in \mathcal{H}_0^{\mathcal{D}_1}} \xi_i $. Then, the inner expectation of the last term in \eqref{eqn:lemma2_1} is at most 1:
    \begin{equation*}
    \begin{split}
        &\mathbb{E}\left[ \frac{B}{M + m_0^{\mathcal{D}_1} - B} \bigg| s_i : i \in \mathcal{H}^{\mathcal{D}_1} \right]\\
        &\quad= \sum_{i = 1}^{m_0^{\mathcal{D}_1}} \mathbb{P}(B = i \mid s_i : i \in \mathcal{H}^{\mathcal{D}_1}) \frac{i}{M + m_0^{\mathcal{D}_1} - i} \\
        &\quad\overset{(\star)}{=} \sum_{i = 1}^{m_0^{\mathcal{D}_1}} \frac{m_0^{\mathcal{D}_1}!}{(i-1)! (m_0^{\mathcal{D}_1} - i + 1)!} (1-\rho)^{i-1} \rho^{m_0^{\mathcal{D}_1} - i + 1} \cdot \frac{m_0^{\mathcal{D}_1}-i+1}{M+m_0^{\mathcal{D}_1}-i} \cdot \frac{1-\rho}{\rho} \\
        &\quad \leq \frac{m_0^{\mathcal{D}_1}}{M + m_0^{\mathcal{D}_1}-1}\cdot \frac{1-\rho}{\rho} \cdot (1-\rho^{m_0^{\mathcal{D}_1}})\\
        &\quad\overset{(\star\star)}{\leq} 1.
        \end{split}
    \end{equation*}
    The first equality (marked by $(\star)$) holds by Lemma \ref{lemma5}, and the last inequality (marked by $(\star\star)$) hold by definition of $M$. Thus, we conclude that $ \mathbb{E}[M_{\hat{k}}] \leq 1 $.
\end{proof}

\begin{lemma}\label{lemma7}
    Suppose that Assumption~\ref{assum_1} holds and $\nu \geq 2$. For all $\epsilon > 0$,
    \begin{equation*}
       \mathbb{P}\left(\frac{1}{2} - \rho \geq \epsilon \right) = O\left(\left(\frac{(\log{m^{\mathcal{D}_2}})^5}{m^{\mathcal{D}_2}}\right)^{1/4}\right).
    \end{equation*}
\end{lemma}

\begin{proof}
    Note that for any index $i \in \mathcal{H}_0^{\mathcal{D}_1}$, there exist a measureable function $g$ such that
    \begin{equation*}
    \begin{split}
        \left\{\frac{1}{2} - \rho \geq \epsilon ~\middle|~ \mathcal{D}_2\right\}
        &=\left\{\frac{1}{2} - \mathbb{P}(\text{sgn}(u_i - \tilde{u}_i) = 1 \mid u_k \wedge \tilde{u}_k : k \in \mathcal{H}^{\mathcal{D}_1}) \geq \epsilon ~\middle|~ \mathcal{D}_2\right\}\\
        &=\left\{\frac{1}{2} - \mathbb{P}(\text{sgn}(u_i - \tilde{u}_i)= 1 \mid u_i \wedge \tilde{u}_i ) \geq \epsilon ~\middle|~ \mathcal{D}_2\right\} \\
        &= \left\{g(u_i \wedge \tilde{u}_i) \in \left[0, \frac{1}{2}- \epsilon\right] ~\middle|~ \mathcal{D}_2\right\}\\
        &= \left\{u_i \wedge \tilde{u}_i \in g^{-1}\left(\left[0, \frac{1}{2} - \epsilon\right]\right) ~\middle|~ \mathcal{D}_2\right\}.
        \end{split}
    \end{equation*}
    Let $E = g^{-1}\left(\left[0, \frac{1}{2} -\epsilon\right]\right)$. Then, by the definition of the set $E$, we have
    \begin{equation*}
    \begin{split}
         &\mathbb{P}\left(\frac{1}{2} - \rho \geq \epsilon ~\middle|~ \mathcal{D}_2\right)\\
         &\quad= \mathbb{E}\left[\mathbb{I}(u_i \wedge \tilde{u}_i \in E) ~\middle|~ \mathcal{D}_2\right]\\
         &\quad\leq \epsilon^{-1} \cdot \mathbb{E}\left[\left\{\frac{1}{2}-\mathbb{P}(\text{sgn}(u_i - \tilde{u}_i)= 1 \mid u_i \wedge \tilde{u}_i ) \right\}\cdot \mathbb{I}(u_i \wedge \tilde{u}_i \in E) ~\middle|~ \mathcal{D}_2\right].
         \end{split}
    \end{equation*}
    Therefore, it suffices to show that 
    \begin{equation*}
    \begin{split}
         &\mathbb{E}\left[\underbrace{\mathbb{E}\left[\left\{\frac{1}{2}-\mathbb{P}(\text{sgn}(u_i - \tilde{u}_i)= 1 \mid u_i \wedge \tilde{u}_i ) \right\} \cdot \mathbb{I}(u_i \wedge \tilde{u}_i \in E) \bigg| \mathcal{D}_2\right] }_{\eqqcolon A}\right]\\
         &\quad= O\left(\left(\frac{(\log{m^{\mathcal{D}_2}})^5}{m^{\mathcal{D}_2}}\right)^{1/4} \right).
    \end{split}
    \end{equation*}
    
    We now turn to the term $A$, and establish the following bound:
    \begin{equation*}
    \begin{split}
        A
        &= \frac{1}{2} ~\mathbb{P}(u_i \wedge \tilde{u}_i \in E \mid \mathcal{D}_2) - \mathbb{P}(\text{sgn}(u_i - \tilde{u}_i) = 1, u_i \wedge \tilde{u}_i \in E \mid \mathcal{D}_2)\\
        &\leq \left|\left\{ \mathbb{P}(\text{sgn}(u_i - \tilde{u}_i) = 1 \mid u_i \wedge \tilde{u}_i \in E, \mathcal{D}_2) - \frac{1}{2} \right\} \cdot \mathbb{P}(u_i \wedge \tilde{u}_i \in E \mid \mathcal{D}_2)\right|\\
        &\overset{(\star)}{\leq}  2\int\int |h_G(x) - h_{\hat{G}}(x)|\,dx\, f_G(s^2; \nu)\,ds^2 \\
        &\overset{(\star\star)}{=} 4 I\left(H(f_G, f_{\hat{G}}) \geq C \frac{\log{m^{\mathcal{D}_2}}}{\sqrt{m^{\mathcal{D}_2}}}\right) + O\left(\left(\frac{(\log{m^{\mathcal{D}_2}})^5}{m^{\mathcal{D}_2}}\right)^{\frac{1}{4}}\right).
        \end{split}
    \end{equation*}
    The second inequality (marked by $(\star)$) holds by Lemma \ref{lemma4}, and last equality (marked by $(\star\star)$) holds by Lemma \ref{lemma3}.
    By taking expectation with respect to $\mathcal{D}_2$, and using the result from Theorem 9 in \cite{ignatiadis2025empirical}, we conclude that
    \begin{equation*}
        \mathbb{E}[A] = O\left(\left(\frac{(\log{m^{\mathcal{D}_2}})^5}{m^{\mathcal{D}_2}}\right)^{\frac{1}{4}}\right).
    \end{equation*}
    This completes the proof.
\end{proof}

\subsection{Proof of Theorem~\ref{theorem1}}
\begin{proof}[Proof of Theorem~\ref{theorem1}]
    Recall the definition of $\hat{k}$ given in~\eqref{k_hat}. We begin the proof by showing that the threshold defined as $s_{(\hat{k})}$ coincides with the data-adaptive threshold $\tau$ introduced in \eqref{thresh_unknownG}. Recall that $\tau$ is defined as
    \begin{equation*}
        \tau \coloneqq \max \left\{ t \in \mathcal{U} \cup \tilde{\mathcal{U}} : Q(t) \leq \alpha \right\},
    \end{equation*}
    where the function $Q(t)$ is given by
    \begin{equation*}
        Q(t) 
        \coloneqq \frac{1 + \sum_{i \in \mathcal{H}^{\mathcal{D}_1}} \mathbb{I}(\tilde{u}_i \leq u_i \wedge t)}{\sum_{i \in \mathcal{H}^{\mathcal{D}_1}} \mathbb{I}(u_i \leq \tilde{u}_i \wedge t)}  
        = \frac{1 + \sum_{i \in \mathcal{H}^{\mathcal{D}_1}} (1 - \xi_i) \cdot \mathbb{I}(s_i \leq t)}{\sum_{i \in \mathcal{H}^{\mathcal{D}_1}} \xi_i \cdot \mathbb{I}(s_i \leq t)},
    \end{equation*}
    with $\xi_i = \mathbb{I}(u_i \leq \tilde{u}_i)$ and $s_i = u_i \wedge \tilde{u}_i$.
    Note that the function $Q(t)$ is piecewise constant and only changes at points in the set $\{s_i : i \in [m^{\mathcal{D}_1}]\} \subset \mathcal{U} \cup \tilde{\mathcal{U}}$.
    Hence, the maximization in the definition of $\tau$ can be restricted to this discrete set, so that
    \begin{equation*}
        \tau = \max \left\{ t \in \{s_i : i \in [m^{\mathcal{D}_1}] \} : Q(t) \leq \alpha \right\},
    \end{equation*}
    which implies that $\tau = s_{(\hat{k})}$.
    
    Based on the above result, we derive an upper bound for $\text{FDP}(\boldsymbol{\delta})$. To this end, consider the following event $E$ defined as
    \begin{equation*}
        E = \left\{ \frac{1}{2} - \rho < \epsilon \right\}.    
    \end{equation*}
    By the definition of $\text{FDP}(\boldsymbol{\delta})$, we have
    \begin{equation*}
    \begin{split}
        \text{FDP}(\boldsymbol{\delta}) 
        &= \frac{1 + \sum_{i \in \mathcal{H}^{\mathcal{D}_1}} (1-\xi_i) I(s_i \leq s_{(\hat{k})})}{\left[\sum_{i \in \mathcal{H}^{\mathcal{D}_1}} \xi_i I(s_i \leq s_{(\hat{k})})\right] \vee 1}\\
        &\quad\quad \times  \frac{1 + \sum_{i \in \mathcal{H}_0^{\mathcal{D}_1}} (1-\xi_i) I(s_i \leq s_{(\hat{k})})}{1 + \sum_{i \in \mathcal{H}^{\mathcal{D}_1}} (1-\xi_i) I(s_i \leq s_{(\hat{k})})} \\
        &\quad\quad \times \frac{\sum_{i \in \mathcal{H}_0^{\mathcal{D}_1}} \xi_i I(s_i \leq s_{(\hat{k})})}{1 + \sum_{i \in \mathcal{H}_0^{\mathcal{D}_1}} (1-\xi_i) I(s_i \leq s_{(\hat{k})})} \\
        &\leq \alpha \cdot 1 \cdot \frac{\sum_{i \in \mathcal{H}_0^{\mathcal{D}_1}} \xi_i I(s_i \leq s_{(\hat{k})})}{M + \sum_{i \in \mathcal{H}_0^{\mathcal{D}_1}} (1-\xi_i) I(s_i \leq s_{(\hat{k})})} \\
        &\quad\quad \times \left\{ \frac{M + \sum_{i \in \mathcal{H}_0^{\mathcal{D}_1}} (1-\xi_i) I(s_i \leq s_{(\hat{k})})}{1 + \sum_{i \in \mathcal{H}_0^{\mathcal{D}_1}} (1-\xi_i) I(s_i \leq s_{(\hat{k})})} \right\} \cdot \mathbb{I}(E) + \mathbb{I}(E^c).
    \end{split}
    \end{equation*}
    The inequality follows from the definition of $\hat{k}$ given in \eqref{k_hat}, the fact that $\mathcal{H}_0^{\mathcal{D}_1} \subset \mathcal{H}^{\mathcal{D}_1}$, and the trivial bound $\text{FDP}(\boldsymbol{\delta}) \leq 1$. 
    For sufficiently small $\epsilon$, we have $M = 1$ on the event $E$. Thus, we obtain
    \begin{equation*}
        \text{FDP}(\boldsymbol{\delta}) 
        \leq \alpha \cdot \frac{\sum_{i \in \mathcal{H}_0^{\mathcal{D}_1}} \xi_i \mathbb{I}(s_i \leq s_{(\hat{k})})}{M + \sum_{i \in \mathcal{H}_0^{\mathcal{D}_1}} (1-\xi_i) I(s_i \leq s_{(\hat{k})})} + \mathbb{I}(E^c).
    \end{equation*}
    Taking expectations on both sides and applying Lemma~\ref{lemma6} and Lemma~\ref{lemma7}, we obtain
    \begin{equation*}
        \text{FDR}(\boldsymbol{\delta}) 
        = \alpha + O\left( \left( \frac{(\log{m^{\mathcal{D}_2}})^5}{m^{\mathcal{D}_2}} \right)^{1/4} \right),    
    \end{equation*}
    which completes the proof.
\end{proof}

\section{Proofs for Section \ref{sec_4}}\label{proof_sec_4}

\subsection{Proof of Theorem~\ref{theorem2}}
\begin{proof}[Proof of Theorem~\ref{theorem2}]
   For each $k \in [K]$, we begin by defining notation. Given $\mathcal{D}^{(-k)}$, let
    \begin{equation*}
    \begin{split}
        \rho^{(k)} &= \mathbb{P}(\text{sgn}(u_i^{(k)} - \tilde{u}_i^{(k)}) = 1 \mid u_j^{(k)} \wedge \tilde{u}_j^{(k)} : j \in \mathcal{H}^{(k)}), \quad \text{and}\\
        M^{(k)} &= \left(1+ \frac{1 - 2\rho^{(k)} - (1 - \rho^{(k)})(\rho^{(k)})^{|\mathcal{H}_0^{(k)}|}}{\rho^{(k)}} |\mathcal{H}_0^{(k)}|\right) \vee 1.
    \end{split}
    \end{equation*}
    For $\epsilon^{(k)} > 0$, define the event $A^{(k)} = \left\{\frac{1}{2} - \rho^{(k)} < \epsilon^{(k)}\right\}$.

    We first claim that for sufficiently small $\epsilon^{(k)}$,
    \begin{equation*}
        \mathbb{E} \left[\sum_{i \in \mathcal{H}_0^{(k)}} E_i^{(k)} ~\mathbb{I}(A^{(k)}) \right] \leq |\mathcal{H}^{(k)}|.
    \end{equation*}
    To see this, note that
    \begin{equation}\label{eqn:theorem2_1}
    \begin{split}
        \mathbb{E}\left[\sum_{i \in \mathcal{H}_0^{(k)}} E_i^{(k)} ~\mathbb{I}(A^{(k)})\right]
        &\overset{(\star)}{=} \mathbb{E}\left[\sum_{i \in \mathcal{H}_0^{(k)}} \frac{|\mathcal{H}^{(k)}| \cdot \xi_i^{(k)} \cdot \mathbb{I}(s_i^{(k)} \leq \tau^{(k)})}{M^{(k)} + \sum_{j \in \mathcal{H}^{(k)}} (1 - \xi_j^{(k)}) \cdot \mathbb{I}(s_j^{(k)} \leq \tau^{(k)})}~ \mathbb{I}(A^{(k)})\right]\\
        &\leq |\mathcal{H}^{(k)}| \cdot \mathbb{E}\left[\frac{\sum_{i \in \mathcal{H}_0^{(k)}} \xi_i^{(k)} \cdot \mathbb{I}(s_i^{(k)} \leq \tau^{(k)})}{M^{(k)} + \sum_{j \in \mathcal{H}_0^{(k)}} (1 - \xi_j^{(k)}) \cdot \mathbb{I}(s_j^{(k)} \leq \tau^{(k)})}\right]\\
        &\overset{(\star\star)}{\leq} |\mathcal{H}^{(k)}|.
    \end{split}
    \end{equation}
    The equality (marked by $(\star)$) holds because, for sufficiently small $\epsilon^{(k)}$, $M^{(k)} = 1$ on the event $A^{(k)}$. The last inequality (marked by $(\star\star)$) holds by Lemma~\ref{lemma6}.

    Next, let $A = \bigcap_{k = 1}^K A^{(k)}$. Then, we have
    \begin{equation}\label{eqn:theorem2_2}
        \mathbb{E}\left[\sum_{k = 1}^K \sum_{i \in \mathcal{H}_0^{(k)}} E_i^{(k)} ~\mathbb{I}(A)\right]
        \leq \sum_{k = 1}^K \mathbb{E}\left[\sum_{i \in \mathcal{H}_0^{(k)}} E_i^{(k)} ~\mathbb{I}(A^{(k)})\right]
        \overset{(\star)}{\leq} \sum_{k = 1}^K |\mathcal{H}^{(k)}| = |\mathcal{H}^{\mathcal{D}}|.
    \end{equation}
    The inequality (marked by ($\star$)) follows from \eqref{eqn:theorem2_1}.
    Furthermore, we have
    \begin{equation*}
        \mathbb{P}(A^c) = \mathbb{P}\left(\bigcup_{k = 1}^K (A^{(k)})^c\right) \leq \sum_{k = 1}^K \mathbb{P}\left((A^{(k)})^c\right).
    \end{equation*}
    Since we fix the number of folds $K$ and by Lemma~\ref{lemma7}, the last term converges to 0 as $|\mathcal{H}^{\mathcal{D}}| \to \infty$.
    
    Let $R$ be the number of rejections obtained by applying the $e$BH procedure to the set $\mathcal{E}$, defined as
    \begin{equation*}
        R = \sum_{k = 1}^K \sum_{i \in \mathcal{H}^{(k)}} \delta_{i, e\text{BH}}^{(k)}.
    \end{equation*}
    We obtain the following upper bound on the FDP($\boldsymbol{\delta}_{e\text{BH}}$):
    \begin{equation*}
    \begin{split}
        \text{FDP}(\boldsymbol{\delta}_{e\text{BH}})
        &= \frac{\sum_{k = 1}^K \sum_{i \in \mathcal{H}_0^{(k)}} \mathbb{I}\left(E_i^{(k)} \geq \frac{|\mathcal{H}^{\mathcal{D}}|}{(\alpha_{e\text{BH}} R)}\right)}{R \vee 1}\\
        &\leq \frac{\alpha_{e\text{BH}}}{|\mathcal{H}^{\mathcal{D}}|} \sum_{k = 1}^K \sum_{i \in \mathcal{H}_0^{(k)}} E_i^{(k)} ~\mathbb{I}(A) + \mathbb{I}(A^c).
    \end{split}
    \end{equation*}
    The last inequality holds due to two facts: (i) the indicator function $\mathbb{I}\left(E_i^{(k)} \geq \frac{|\mathcal{H}^{\mathcal{D}}|}{\alpha_{e\text{BH}} R}\right)$ can be upper bounded by $\frac{\alpha_{e\text{BH}} R}{|\mathcal{H}^{\mathcal{D}}|} E_i^{(k)}$, and (ii) by definition, the false discovery proportion $\mathrm{FDP}(\boldsymbol{\delta}_{e\text{BH}})$ is always less than or equal to 1. Taking expectation on both sides, we conclude:
    \begin{equation*}
        \text{FDR}(\boldsymbol{\delta}_{e\text{BH}}) \leq \frac{\alpha_{e\text{BH}}}{|\mathcal{H}^{\mathcal{D}}|} ~\mathbb{E}\left[\sum_{k = 1}^K \sum_{i \in \mathcal{H}_0^{(k)}} E_i^{(k)} ~\mathbb{I}(A)\right] + \mathbb{P}(A^c) \overset{(\star)}{\leq} \alpha_{e\text{BH}} + \mathbb{P}(A^c).
    \end{equation*}
    The inequality (marked by ($\star$)) follows from \eqref{eqn:theorem2_2}.
    Since $\mathbb{P}(A^c) \to 0$ as $|\mathcal{H}^{\mathcal{D}}| \to \infty$, the desired result follows.
\end{proof}

\subsection{Proof of the Theorem~\ref{theorem3}}

\begin{proof}[Proof of Theorem~\ref{theorem3}]
    We begin by defining the necessary notation. 
    Let $\tilde{E}_i \in \mathcal{E}$ denote the test statistic corresponding to the $i$-th hypothesis. Then, the set $\mathcal{E}$ can be expressed as $\mathcal{E} = \{\tilde{E}_i : i \in \mathcal{H}^{\mathcal{D}}\}$. Similarly, the set $\mathcal{E}^*$ can be written as $\mathcal{E}^* = \{\tilde{E}_i / U : i \in \mathcal{H}^{\mathcal{D}}\}$. Let $\tilde{E}_{(i)}$ denote the $i$-th largest order statistic among $\mathcal{E}$ and define  
    \begin{equation*} 
        r[i] := \sum_{j \in \mathcal{H}^{\mathcal{D}}} \mathbb{I}(\tilde{E}_j   \geq \tilde{E}_i) \quad \text{and} \quad l(i) := \underset{j \geq r[i]}{\arg\max} ~ j \tilde{E}_{(j)}.
    \end{equation*}
    
    Now, consider the stochastic rounding function proposed by \cite{xu2023more}, defined as follows:
    \begin{equation*}
        S_{\alpha_{l(i)}}(\tilde{E}_i) := \tilde{E}_i \cdot \mathbb{I}(\tilde{E}_i \geq \alpha_{l(i)}^{-1}) 
        + \alpha_{l(i)}^{-1} \cdot \mathbb{I}(\alpha_{l(i)}^{-1} > \tilde{E}_i \geq U \alpha_{l(i)}^{-1}),
    \end{equation*}
    where $\alpha_{l(i)}^{-1} := \frac{|\mathcal{H}^{\mathcal{D}}|}{l(i) \alpha_{e\text{BH}}}$.
    \cite{xu2023more} have shown that applying the $e$BH procedure to the set $\mathcal{E}^*$ yields the same rejection set as applying it to the set $\{S_{\alpha_{l(i)}}(\tilde{E}_i) : i \in \mathcal{H}^{\mathcal{D}}\}$ (see Proposition~4 of \cite{xu2023more}). Therefore, it suffices to show that applying the $e$BH procedure to $\{S_{\alpha_{l(i)}}(\tilde{E}_i) : i \in \mathcal{H}^{\mathcal{D}}\}$ asymptotically controls the FDR as $|\mathcal{H}^{\mathcal{D}}| \to \infty$.

    Let $A$ denote the event defined in the proof of Theorem~\ref{theorem2}. 
    Then, we have
    \begin{equation*}
    \begin{split}
        \mathbb{E}\left[S_{\alpha_{l(i)}}(\tilde{E}_i) \cdot \mathbb{I}(A)\right]
        &= \mathbb{E}\left[\tilde{E}_i \cdot \mathbb{I}(\tilde{E}_i \geq \alpha_{l(i)}^{-1}) \cdot \mathbb{I}(A)\right] \\
        &\quad\quad + \mathbb{E}\left[\alpha_{l(i)}^{-1} \cdot \mathbb{I}(\alpha_{l(i)}^{-1} > \tilde{E}_i \geq U \alpha_{l(i)}^{-1}) \cdot \mathbb{I}(A)\right] \\
        &= \mathbb{E}\left[\tilde{E}_i \cdot \mathbb{I}(\tilde{E}_i \geq \alpha_{l(i)}^{-1}) \cdot \mathbb{I}(A)\right] \\
        &\quad\quad + \mathbb{E}\left[\alpha_{l(i)}^{-1} \cdot \mathbb{I}(\alpha_{l(i)}^{-1} > \tilde{E}_i) \cdot \mathbb{P}(U \leq \alpha_{l(i)} \tilde{E}_i \mid \alpha_{l(i)}, \tilde{E}_i, A) \cdot \mathbb{I}(A)\right] \\
        &\overset{(\star)}{=} \mathbb{E}\left[\tilde{E}_i \cdot \mathbb{I}(\tilde{E}_i \geq \alpha_{l(i)}^{-1}) \cdot \mathbb{I}(A)\right] \\
        &\quad\quad + \mathbb{E}\left[\alpha_{l(i)}^{-1} \cdot \mathbb{I}(\alpha_{l(i)}^{-1} > \tilde{E}_i) \cdot (\alpha_{l(i)} \tilde{E}_i \wedge 1) \cdot \mathbb{I}(A)\right] \\
        &\overset{(\star\star)}{=} \mathbb{E}\left[\tilde{E}_i \cdot \mathbb{I}(\tilde{E}_i \geq \alpha_{l(i)}^{-1}) \cdot \mathbb{I}(A)\right] 
        + \mathbb{E}\left[\tilde{E}_i \cdot \mathbb{I}(\tilde{E}_i < \alpha_{l(i)}^{-1}) \cdot \mathbb{I}(A)\right] \\
        &= \mathbb{E}\left[\tilde{E}_i \cdot \mathbb{I}(A)\right].
    \end{split}
    \end{equation*}
    The third equality (marked by $(\star)$) holds because $U$ is a uniform random variable independent of $\alpha_{l(i)}$, $\tilde{E}_i$, and the event $A$. The fourth equality (marked by $(\star\star)$) follows from the fact that, under the event $\{\alpha_{l(i)}^{-1} > \tilde{E}_i\}$, we have $\alpha_{l(i)} \tilde{E}_i \wedge 1 = \alpha_{l(i)} \tilde{E}_i$.
    Summing over all null hypotheses, we obtain:
    \begin{equation}\label{eqn:theorem3}
        \mathbb{E}\left[\sum_{i \in \mathcal{H}_0^{\mathcal{D}}} S_{\alpha_{l(i)}}(\tilde{E}_i) \cdot \mathbb{I}(A)\right] 
        = \mathbb{E}\left[\sum_{i \in \mathcal{H}_0^{\mathcal{D}}} \tilde{E}_i \cdot \mathbb{I}(A)\right] 
        \leq |\mathcal{H}^{\mathcal{D}}|,
    \end{equation}
    where the last inequality follows from~\eqref{eqn:theorem2_2}.

    Let $R$ denote the number of rejections from the $e$BH procedure applied to $\mathcal{E}^*$:
    \begin{equation*}
        R = \sum_{k = 1}^K \sum_{i \in \mathcal{H}^{(k)}} \delta_{i, \text{U-}e\text{BH}}^{(k)}.
    \end{equation*}
    Then the false discovery proportion satisfies:
    \begin{equation*}
    \begin{split}
        \text{FDP}(\boldsymbol{\delta}_{\text{U-}e\text{BH}})
        &\overset{(\star)}{=} \frac{\sum_{i \in \mathcal{H}_0^{\mathcal{D}}} \mathbb{I}\left(S_{\alpha_{l(i)}}(\tilde{E}_i) \geq \frac{|\mathcal{H}^{\mathcal{D}}|}{\alpha_{e\text{BH}} R} \right) }{R \vee 1} \\
        &\leq \frac{\alpha_{e\text{BH}}}{|\mathcal{H}^{\mathcal{D}}|} \sum_{i \in \mathcal{H}_0^{\mathcal{D}}} S_{\alpha_{l(i)}}(\tilde{E}_i) \cdot \mathbb{I}(A) + \mathbb{I}(A^c).
    \end{split}
    \end{equation*}
    The equality (marked by $(\star)$) follows from Proposition~4 of \cite{xu2023more}.
    Taking expectation on both sides yields:
    \begin{equation*}
        \text{FDR}(\boldsymbol{\delta}_{\text{U-}e\text{BH}}) 
        \leq \frac{\alpha_{e\text{BH}}}{|\mathcal{H}^{\mathcal{D}}|} ~\mathbb{E}\left[\sum_{i \in \mathcal{H}_0^{\mathcal{D}}} S_{\alpha_{l(i)}}(\tilde{E}_i) \cdot \mathbb{I}(A)\right] + \mathbb{P}(A^c) 
        \overset{(\star)}{\leq} \alpha_{e\text{BH}} + \mathbb{P}(A^c).
    \end{equation*}
    The last inequality (marked by $(\star)$) follows from~\eqref{eqn:theorem3}. Since $\mathbb{P}(A^c) \to 0$ as $|\mathcal{H}^{\mathcal{D}}| \to \infty$, the desired result follows.
\end{proof}

\section{An Illustrative Example of the General Procedure}\label{ex_MixTwice}
We use the \textit{MixTwice} method of \citet{zheng2021mixtwice} to demonstrate how the general procedure is implemented in practice.

\textit{Step 1: Specify the Prior Distribution and the Conformity Score Function.} The method assumes the following prior distributions:
\begin{equation*}
    \sigma_i^2 \overset{\text{i.i.d.}}{\sim} G(\cdot), 
    \quad 
    \mu_i \mid \sigma_i^2 \overset{\text{i.i.d.}}{\sim} (1-\pi)\,\delta_0(\cdot) + \pi\,f(\cdot),
\end{equation*}
where $G$ is an arbitrary distribution on $[0, \infty)$, $\pi \in (0,1)$ denotes the proportion of non-null effects, $\delta_0$ is the Dirac measure at zero (representing the null component), and $f$ is a unimodal distribution centered at zero (representing the non-null component). This prior specification implies independence between $\mu_i$ and $\sigma_i^2$. Given this prior, the conformity score function $u(\cdot, \cdot)$ is defined as
\begin{equation*}
    u(x, s^2) \coloneqq p(\mu = 0 \mid x, s^2) 
    = \frac{(1-\pi)\,p(x \mid \mu = 0, s^2)}{p(x \mid s^2)},
\end{equation*}
where 
\begin{equation*}
\begin{split}
    p(x \mid \mu = 0, s^2)
    &= \frac{\int p(x \mid \mu = 0, \sigma^2)\, p(s^2 \mid \sigma^2)\, dG(\sigma^2)}
             {\int p(s^2 \mid \sigma^2)\, dG(\sigma^2)}, 
    \\[6pt]
    p(x \mid s^2)
    &= \frac{(1-\pi)\!\int p(x \mid \mu = 0, \sigma^2)\, p(s^2 \mid \sigma^2)\, dG(\sigma^2)}
             {\int p(s^2 \mid \sigma^2)\, dG(\sigma^2)}
    \\
    &\quad + \frac{\pi \!\iint p(x \mid \mu, \sigma^2)\, p(s^2 \mid \sigma^2)\, f(\mu)\, d\mu\, dG(\sigma^2)}
                   {\int p(s^2 \mid \sigma^2)\, dG(\sigma^2)} .
\end{split}
\end{equation*}
The output of the conformity score function represents the posterior probability that the null hypothesis is true given $(x, s^2)$, which is commonly referred to as the local false discovery rate (Lfdr; \citealp{efron2001empirical}).

\textit{Step 2: Estimate the Prior Distribution and the Conformity Score Function.} The unknown prior components $(G, \pi, f)$ are estimated from the observed data $\mathcal{D} = \{(X_i, S_i^2)\}_{i=1}^m$. To make the estimation problem tractable, \citet{zheng2021mixtwice} approximate both $G$ and $f$ 
by finite discrete distributions that satisfy the assumed structural constraints (e.g., the unimodality of $f$). These discrete approximations lead to an explicit expression for the marginal likelihood of the observed data, which is then maximized to obtain estimates of the prior components. Having obtained these estimates, the conformity score function is estimated via the plug-in principle by replacing the unknown prior with its estimated counterpart. 

\textit{Step 3: Compute Conformity Scores.} Using the estimated conformity score function $\hat{u}(\cdot, \cdot)$, each hypothesis is assigned a conformity score:
\begin{equation*}
    u_i = \hat{u}(X_i, S_i^2), \quad i \in [m].
\end{equation*}
These scores, which approximate the Lfdr, are then used to rank the hypotheses. By construction, smaller conformity scores indicate stronger evidence against the null hypothesis.

\textit{Step 4: Derive a Data-Adaptive Threshold and Define the Decision Rule.} To control the false discovery rate (FDR; \citealp{benjamini1995controlling}) at a nominal level $\alpha$, one may adopt the data-adaptive thresholding procedure proposed by \citet{sun2007oracle}.

\section{Implementation Details of Data-Splitting Methods}\label{implementation_details_data_split}
Both data-splitting methods (\textit{COIN-SS} and \textit{COIN-FS}) are built upon the COIN algorithm, which requires specifying the following three core components:
\begin{itemize}
    \item a working prior distribution for $(\mu_i, \sigma_i^2)$,
    \item a conformity score function, and
    \item an estimation algorithm for the prior distribution and the conformity score function.
\end{itemize}
We now describe each of these components in detail. It is important to emphasize that the choices presented here serve as illustrative examples. Practitioners may modify or replace any of these components according to their specific context or application.

\subsection{Working Prior}
Unless otherwise specified, we adopt the same prior specification as in the \textit{gg-Mix} method:
\begin{equation*}
    \sigma_i^2 \overset{\text{i.i.d.}}{\sim} G(\cdot), \quad \mu_i \mid \sigma_i^2 \overset{\text{i.i.d.}}{\sim} (1 - \pi)\, \delta_0(\cdot) + \pi\, f(\cdot),    
\end{equation*}
where $ G $ is an arbitrary distribution supported on $ [0, \infty) $, and $ f(\cdot) $ is a generic density function with no structural assumptions, such as unimodality or symmetry.

\subsection{Conformity Score Function}
Based on the working prior specified above, we define the conformity score function $u: \mathbb{R} \times [0, \infty) \to [0, \infty)$, mapping each observation $(x, s^2)$ to:
\begin{equation*}
    u(x, s^2) = \frac{p(x \mid \mu = 0,\, s^2)}{p(x \mid s^2)}.    
\end{equation*}
This conformity score is proportional to the Lfdr:
\begin{equation*}
    \text{Lfdr}(x, s^2) = \frac{(1 - \pi) \cdot p(x \mid \mu = 0,\, s^2)}{p(x \mid s^2)},
\end{equation*}
which is known to be optimal for ranking hypotheses under our stated prior assumption \citep{sun2007oracle}. The only distinction between $u(x, s^2)$ and $\text{Lfdr}(x, s^2)$ is the constant factor $(1 - \pi)$. Given that our procedure is rank-based, the two scores yield identical results within the framework considered.

\subsection{Estimation Algorithm}
Note that the conformity score does not require explicit estimation of the non-null proportion, as this proportion is only indirectly reflected in the denominator. Thus, instead of estimating the full prior distribution and using a plug-in approach, we directly estimate the conformity score function itself. Specifically, we employ a likelihood-based strategy to estimate both $p(x \mid \mu = 0, s^2)$ and $p(x \mid s^2)$, which serve as the numerator and denominator of the conformity score function, respectively.

First, the numerator $p(x \mid \mu = 0, s^2)$ is expressed as:
\begin{equation*}
    p(x \mid \mu = 0, s^2) = \frac{\int p(x \mid \mu = 0, \sigma^2) \cdot p(s^2 \mid \sigma^2) \, dG(\sigma^2)}{\int p(s^2 \mid \sigma^2) \, dG(\sigma^2)},
\end{equation*}
where the only unknown component is the prior distribution $G$. We estimate $G$ using the NPMLE as described in equation \eqref{NPMLE_G}, and denote the resulting estimator by $\hat{G}$. Substituting $\hat{G}$ into the expression above yields the estimate of $p(x \mid \mu = 0, s^2)$.\footnote{More precisely, for practical implementation, rather than maximizing over the class of all distributions supported on $[0,\infty)$, we approximate the NPMLE by restricting $G$ to a discrete class supported on a logarithmically equispaced grid of 50 points, spanning from the 1\% percentile to the maximum of the observed values $\{S_{i,2}^2\}_{i=1}^{m^{\mathcal{D}_2}}$.}

For the denominator $ p(x \mid s^2) $, we approximate the non-null distribution $ f(\cdot) $ of $ \mu_i $ using a Gaussian location-scale mixture:
\begin{equation*}
    f(\cdot) \approx \sum_{k=1}^{k_1} \pi_k^\prime \, \mathcal{N}(\cdot ; \gamma_k, \zeta^2) + \sum_{k=k_1+1}^{k_1+k_2} \pi_k^\prime \, \mathcal{N}(\cdot ; 0, \eta_k^2) \ \eqqcolon \ \tilde{f}(\cdot),
\end{equation*}
In this formulation, the mixture weights $\{\pi_k^\prime\}_{k = 1}^{k_1 + k_2}$, which sum to one, are treated as free parameters to be estimated, while the remaining hyperparameters are fixed in advance. Specifically, for the location mixture component, we recommend setting $ k_1 $ to a sufficiently large value (e.g., $ k_1 = 30 $) to ensure adequate modeling flexibility. Given a fixed $ k_1 $, we define the lower and upper 1\% percentiles of the observed $ X_i $ values as $ a $ and $ b $, respectively, and construct $ k_1 $ equally spaced grid points between $ a $ and $ b $. These grid points are assigned as the location parameters: $ \gamma_1 = a $, $ \gamma_2 $ the second smallest, and so on, with $ \gamma_{k_1} = b $. By default, the variance $ \zeta^2 $ is fixed at 1. For the scale mixture component, we specify $ k_2 $ and the variances $ \eta_k^2 $ following the recommendations provided in the supplementary material of \citet{stephens2017false}. We emphasize that, by employing sufficiently large values for $ k_1 $ and $ k_2 $, the mixture $ \tilde{f} $ can provide a close approximation to the true underlying distribution $ f $.

Based on the approximation $ f(\cdot) \approx \tilde{f}(\cdot) $ and the estimated prior $ \hat{G} $, we approximate the conditional distribution of $ X_i $ given $ S_i^2 $ as follows:
\begin{equation*}
\begin{split}
    p(x \mid s^2) &=  (1-\pi) \cdot \frac{\int p(x \mid \mu = 0, \sigma^2) \cdot p(s^2 \mid \sigma^2) \, dG(\sigma^2)}{\int p(s^2 \mid \sigma^2) \, dG(\sigma^2)} \\
    &\quad\quad +  ~\pi \cdot \frac{\int p(x \mid \mu, \sigma^2) \cdot p(s^2 \mid \sigma^2) \cdot f(\mu) \, d\mu \, dG(\sigma^2)}{\int p(s^2 \mid \sigma^2) \, dG(\sigma^2)}\\
    &\approx (1-\pi) \cdot \frac{\int p(x \mid \mu = 0, \sigma^2) \cdot p(s^2 \mid \sigma^2) \, d\hat{G}(\sigma^2)}{\int p(s^2 \mid \sigma^2) \, d\hat{G}(\sigma^2)} \\
    &\quad\quad +  ~\pi \cdot \frac{\int p(x \mid \mu, \sigma^2) \cdot p(s^2 \mid \sigma^2) \cdot \tilde{f}(\mu) \, d\mu \, d\hat{G}(\sigma^2)}
    {\int p(s^2 \mid \sigma^2) \, d\hat{G}(\sigma^2)} =: \tilde{p}(x \mid s^2).
    \end{split}
\end{equation*}
Treating the approximated conditional distribution $ \tilde{p}(x \mid s^2) $ as the true distribution, we estimate the unknown parameters $ \pi $ and $ \{ \pi_k^\prime : k \in [k_1 + k_2] \} $ by maximizing the approximate conditional likelihood of the observed summary statistic pairs. The resulting optimization problem is solved using the sequential quadratic programming method proposed by \citet{kim2020fast}.

\begin{remark}
    In our implementation, we adopt the same estimation procedure as \textit{gg-Mix} for $p(x \mid \mu = 0, s^2)$, but modify the estimation of $p(x \mid s^2)$ to better align with our objective. This modification reflects a key difference in methodological goals. In the original \textit{gg-Mix} framework, the conformity score function explicitly depends on the null proportion, and a conservative estimate of it is essential for controlling the FDR. Accordingly, a regularization step is introduced to ensure conservativeness, even at the cost of some bias in estimating $p(x \mid s^2)$. In contrast, our conformity score does not require an explicit estimate of the null proportion, as this quantity is only implicitly involved through the term $p(x \mid s^2)$. Therefore, accurate estimation of $p(x \mid s^2)$ becomes more important in our framework. To this end, we omit the regularization step and estimate $p(x \mid s^2)$ directly without penalization.
\end{remark}

\subsection{Additional Implementation Details for \textit{COIN-FS}}

\subsubsection{Choice of Hyperparameters}
For \textit{COIN-FS}, several hyperparameters need to be specified: the number of folds $ K $, the fold-specific target FDR levels $ \alpha^{(k)} $, and the multiple testing procedure used in the final decision step (either $ e $BH or U-$ e $BH). Throughout this paper, we fix the number of folds at $ K = 5 $, set the fold-specific FDR levels as $ \alpha^{(k)} = 0.9 \times \alpha_{e\text{BH}} $ for all $ k \in [K] $, and adopt the U-$ e $BH procedure to boost power. 

\paragraph{Choice of fold-specific target FDR levels $\alpha^{(k)}$ for \textit{COIN-FS}.}

The \textit{COIN-FS} method requires specification of fold-specific FDR levels $\alpha^{(k)}$. To simplify this choice while maintaining flexibility, we introduce a simple parameterization scheme as follows: we set $\alpha^{(k)} = c \times \alpha_{e\text{BH}}$, where $c > 0$ is a tuning constant. To investigate the impact of this parameter, we conducted a simulation study by varying $c$ over the grid $\{0.7, 0.8, 0.9, 1\}$ and evaluated the resulting FDR and TPR across the simulation settings considered in the main manuscript. Although this procedure may not yield the optimal choice, it provides a practical guideline for implementation. A full investigation of the optimal choice lies beyond the scope of this work.

The simulation results are summarized in Figures~\ref{fig:sim1_fdr_c}--\ref{fig:sim2_tpr_c} and lead to the following observations. First, as expected, the FDR is well controlled across all simulation settings regardless of the choice of $c$. Second, the power obtained with $c = 0.9$ generally exceeds that of alternative choices across most simulation settings. While $c = 1$ occasionally yields higher power than $c = 0.9$, this advantage is not consistently observed. Based on these findings, we fix $c = 0.9$ for all simulation studies.

\begin{figure}[H]
    \centering
    \includegraphics[width=\linewidth]{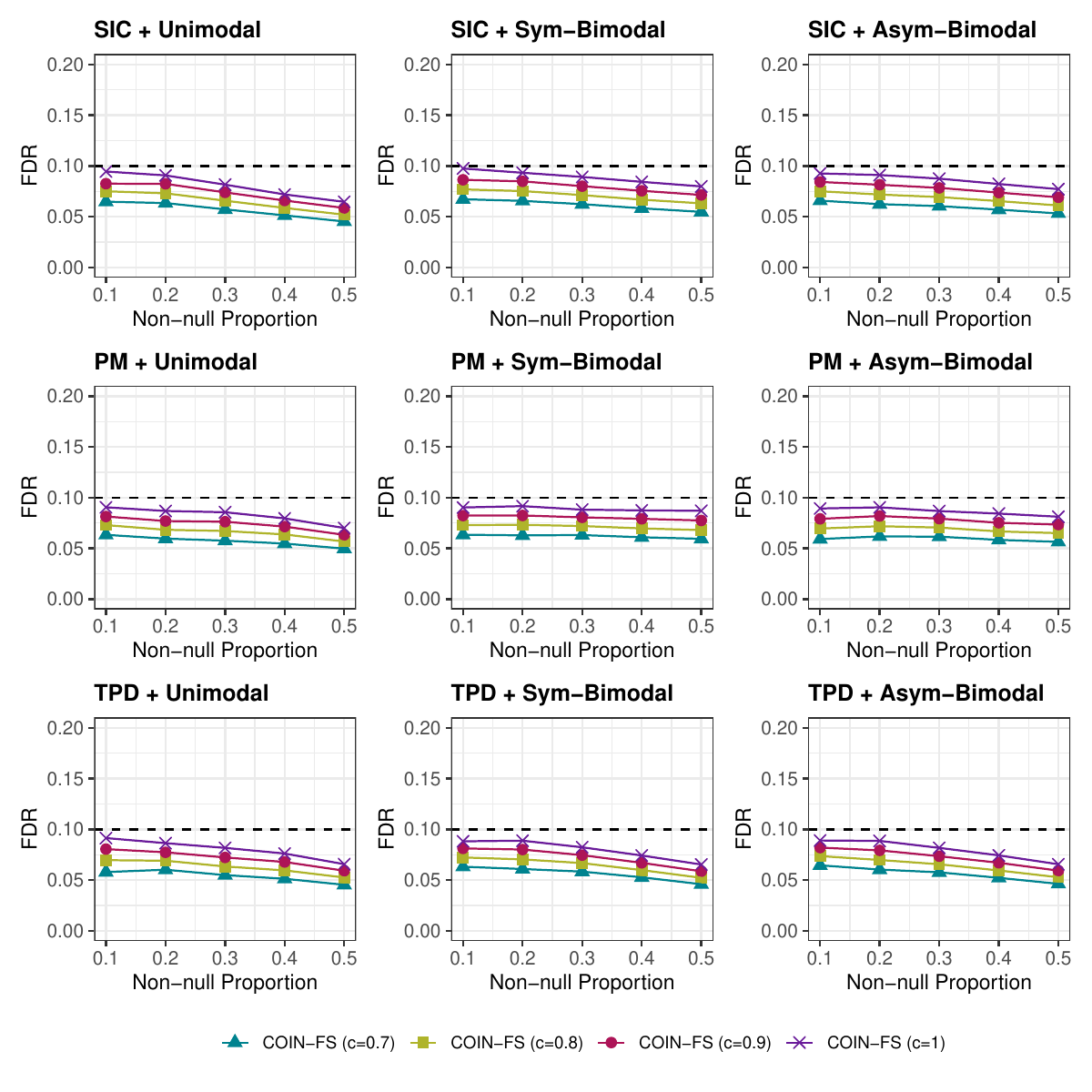}
    \caption{False discovery rates (FDRs) for four multiple testing methods, each corresponding to a different choice of $c$. The simulation settings are identical to those used in Scenario~1 of the main manuscript. The dashed horizontal line indicates the target FDR level of 0.1.}
    \label{fig:sim1_fdr_c}
\end{figure}

\begin{figure}[H]
    \centering
    \includegraphics[width=\linewidth]{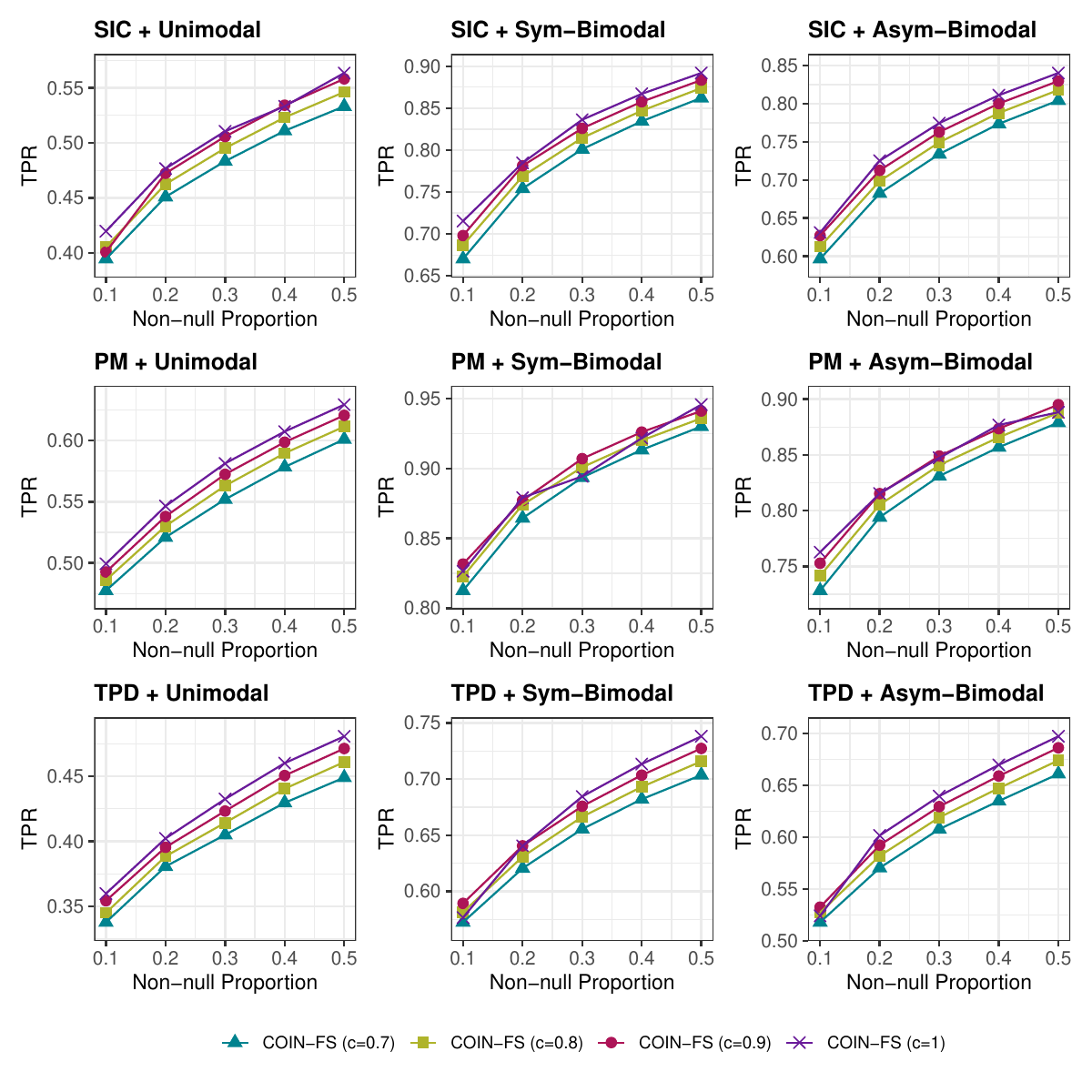}
    \caption{True postive rates (TPRs) for four multiple testing methods, each corresponding to a different choice of $c$. The simulation settings are identical to those used in Scenario~1 of the main manuscript.}
    \label{fig:sim1_tpr_c}
\end{figure}

\begin{figure}[H]
    \centering
    \includegraphics[width=\linewidth]{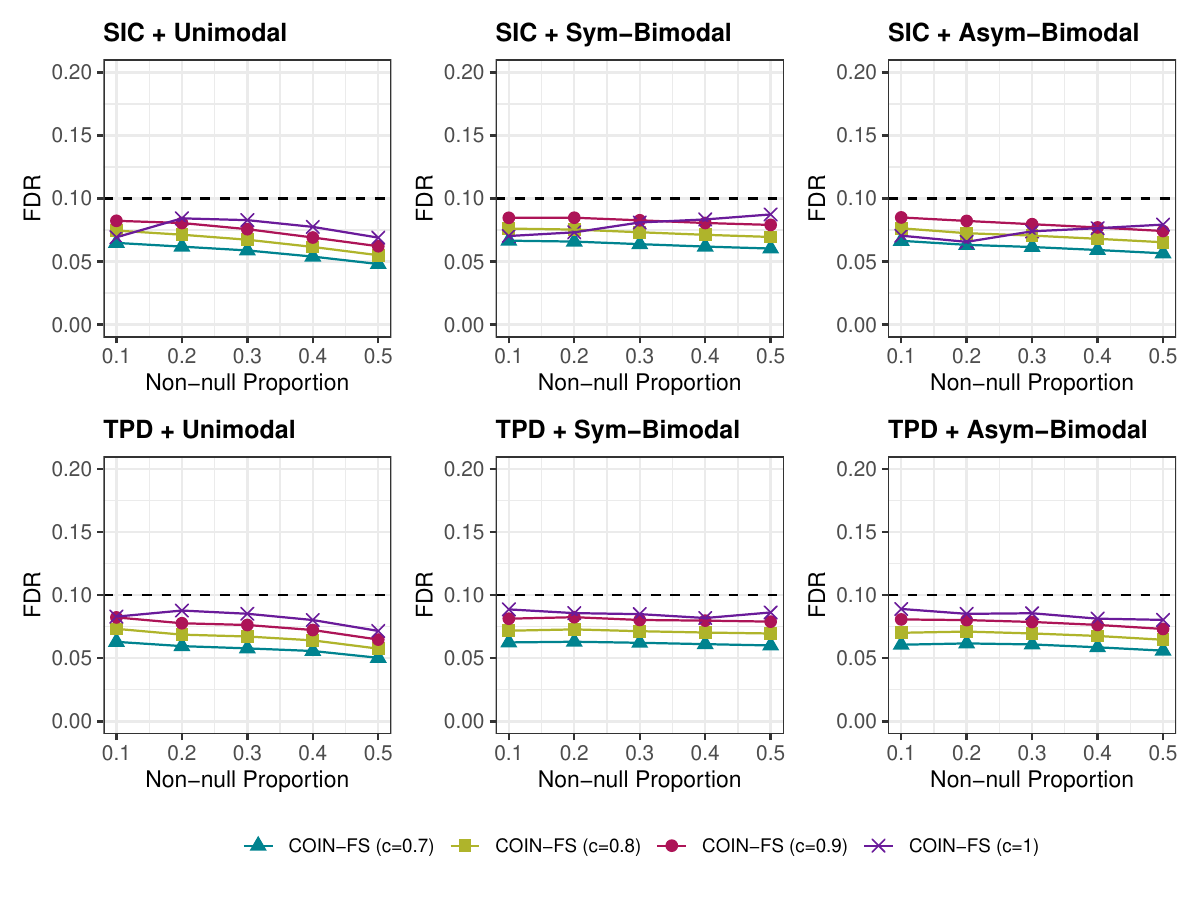}
    \caption{False discovery rates (FDRs) for four multiple testing methods, each corresponding to a different choice of $c$. The simulation settings are identical to those used in Scenario~2 of the main manuscript. The dashed horizontal line indicates the target FDR level of 0.1.}
    \label{fig:sim2_fdr_c}
\end{figure}

\begin{figure}[H]
    \centering
    \includegraphics[width=\linewidth]{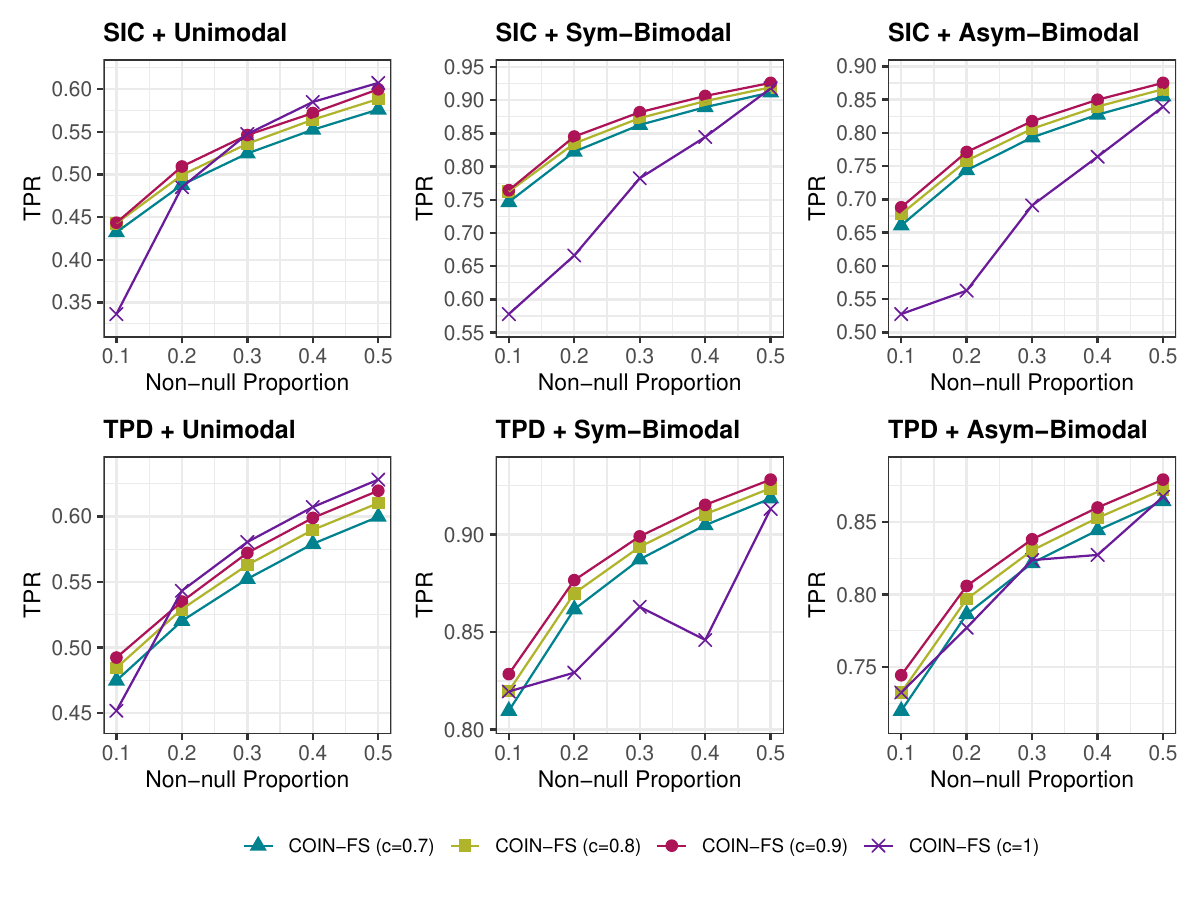}
    \caption{True postive rates (TPRs) for four multiple testing methods, each corresponding to a different choice of $c$. The simulation settings are identical to those used in Scenario~2 of the main manuscript. }
    \label{fig:sim2_tpr_c}
\end{figure}

\subsubsection{A Uniform Improvement on Data-Adaptive Threshold}
Following the approach of \cite{ren2024derandomised}, the power of \textit{COIN-FS} can be uniformly improved by refining the data-adaptive threshold.

To establish the validity of the feature-splitting method, we rely solely on the fact that the data-adaptive threshold~$\tau^{(k)}$, defined as
\begin{equation}\label{tau_k}
    \tau^{(k)} 
    = \max \left\{ 
        t \in \mathcal{U}^{(k)} \cup \tilde{\mathcal{U}}^{(k)} : 
        \frac{1 + \sum_{i \in \mathcal{H}^{(k)}} (1 - \xi_i^{(k)}) \, \mathbb{I}\!\left(s_i^{(k)} \le t\right)}
        {\bigl[\sum_{i \in \mathcal{H}^{(k)}} \xi_i^{(k)} \, \mathbb{I}\!\left(s_i^{(k)} \le t\right)\bigr] \vee 1}
        \le \alpha^{(k)}
    \right\},
\end{equation}
is a stopping time with respect to the filtration~$\{ \mathcal{F}_i^{(k)} \}$, where
\begin{equation*}
    \mathcal{F}_i^{(k)} 
    = \sigma\!\left(
        \{ \xi_j^{(k)} : j \in \mathcal{H}_1^{(k)} \},
        \{ s_j^{(k)} : j \in \mathcal{H}^{(k)} \},
        \{ V_j^{(k)}, \tilde{V}_j^{(k)} : i \le j \le |\mathcal{H}^{(k)}| \}
    \right).
\end{equation*}
Following the approach of \citet{ren2024derandomised}, we can universally enhance the power of the feature-splitting method 
by refining the data-adaptive threshold as
\begin{equation}\label{refined_tau_k}
\begin{split}
    \tau^{(k)} 
    = \max \Bigg\{ 
        t \in \mathcal{U}^{(k)} \cup \tilde{\mathcal{U}}^{(k)} :\,
        &\frac{1 + \sum_{i \in \mathcal{H}^{(k)}} (1 - \xi_i^{(k)}) \, \mathbb{I}\!\left(s_i^{(k)} \le t\right)}
        {\bigl[\sum_{i \in \mathcal{H}^{(k)}} \xi_i^{(k)} \, \mathbb{I}\!\left(s_i^{(k)} \le t\right)\bigr] \vee 1}
        \le \alpha^{(k)} \\
        &\quad \text{or} \quad 
        \sum_{i \in \mathcal{H}^{(k)}} 
        \xi_i^{(k)} \, \mathbb{I}\!\left(s_i^{(k)} \le t\right)
        < \frac{1}{\alpha^{(k)}}
    \Bigg\}.
\end{split}
\end{equation}
It is evident that the refined threshold in~\eqref{refined_tau_k} 
is a stopping time with respect to the filtration~$\{\mathcal{F}_i^{(k)}\}$ and always yields a value greater than or equal to that of~\eqref{tau_k}, thereby universally improving the statistical power of the procedure. 

The refinement of the data-adaptive threshold from~\eqref{tau_k} to~\eqref{refined_tau_k} becomes effective only when
\begin{equation*}
\begin{split}
    &\max \Bigg\{ 
        t \in \mathcal{U}^{(k)} \cup \tilde{\mathcal{U}}^{(k)} :
        \frac{1 + \sum_{i \in \mathcal{H}^{(k)}} (1 - \xi_i^{(k)}) \, \mathbb{I}\!\left(s_i^{(k)} \le t\right)}
        {\bigl[ \sum_{i \in \mathcal{H}^{(k)}} \xi_i^{(k)} \, \mathbb{I}\!\left(s_i^{(k)} \le t\right) \bigr] \vee 1}
        \le \alpha^{(k)} 
    \Bigg\} \\
    &\quad < 
    \max \Bigg\{ 
        t \in \mathcal{U}^{(k)} \cup \tilde{\mathcal{U}}^{(k)} :
        \sum_{i \in \mathcal{H}^{(k)}} 
        \xi_i^{(k)} \, \mathbb{I}\!\left(s_i^{(k)} \le t\right) 
        < \frac{1}{\alpha^{(k)}}
    \Bigg\}.
\end{split}
\end{equation*}
This condition indicates that the threshold defined in~\eqref{tau_k} 
results in $ \tau^{(k)} = -\infty $, thereby assigning zero to all $ e $-variables corresponding to hypotheses in $ \mathcal{H}^{(k)} $. Hence, when the original data-adaptive threshold in~\eqref{tau_k} fails to reject any hypotheses in $ \mathcal{H}^{(k)} $, the refined threshold in~\eqref{refined_tau_k} can enhance the detection power. 

The refinement in~\eqref{refined_tau_k} guarantees that the number of hypotheses assigned non-zero $ e $-values is bounded below by
\begin{equation*}
    \left\lceil \frac{1}{\alpha^{(k)}} \right\rceil - 1,
\end{equation*}
where $ \left\lceil \cdot \right\rceil $ denotes the ceiling function, that is, the smallest integer greater than or equal to the input. For example, when $ \alpha^{(k)} = 0.05 $, we have $ \left\lceil 1/\alpha^{(k)} \right\rceil = 20 $, and thus at least $ 20 - 1 = 19 $ hypotheses receive non-zero $ e $-values under the refined threshold in~\eqref{refined_tau_k}.

\section{Additional Simulation Results}\label{add_sim}
In this section, we compare the performance of the method proposed by \cite{ignatiadis2025empirical}, hereafter referred to as \textit{IS}, with our proposed method \textit{COIN-FS}. We consider the same simulation settings that are considered in the main manuscript. 

The \textit{IS} method assumes that the variance parameters $\sigma_i^2$ are independently drawn from an unknown prior distribution $G$ supported on $[0,\infty)$, while the mean parameters $\mu_i$ are treated as fixed but unknown constants. As a result, \textit{IS} does not fall within the empirical Bayes (EB) framework but instead belongs to the empirical partially Bayes (EPB) framework. Nevertheless, it has been theoretically established that \textit{IS} can asymptotically control the FDR even when the data are generated under an EB model. This theoretical robustness motivates a comparison between \textit{IS} and \textit{COIN-FS}.

The simulation results are summarized in Figures~\ref{fig:sim1_FDR_IS}--\ref{fig:sim2_TPR_IS}. Across all simulation settings, both \textit{IS} and \textit{COIN-FS} consistently control the FDR, in agreement with their theoretical guarantees. In terms of power, when the non-null proportion is small, \textit{IS} tends to exhibit slightly higher power than \textit{COIN-FS}. However, as the non-null proportion increases, the power of \textit{COIN-FS} eventually surpasses that of \textit{IS}.

The observed differences in power can be attributed to several factors. 
One key factor is that \textit{IS} relies on the BH procedure, which effectively corresponds to the assumption $1-\pi = 1$. This results in a conservative rejection threshold and leads to reduced power, particularly when the non-null proportion is non-negligible. 
Another important factor is the ability to incorporate structural information about the non-null distribution $f$. The proposed \textit{COIN-FS} adopts an EB approach and explicitly models a prior distribution for $\mu_i$, thereby leveraging structural information in $f$. In contrast, \textit{IS}, as an EPB method, treats $\mu_i$ as fixed but unknown constants and cannot exploit such structure. Consequently, as the structure of $f$ becomes more informative, \textit{COIN-FS} tends to exhibit superior power relative to \textit{IS}.
Finally, \textit{COIN-FS} relies on data splitting, which inevitably leads to a loss of information for estimating the conformity score function. 
As a result, in sparse signal settings, although \textit{COIN-FS} can exploit structural information in $f$, the estimation accuracy may be limited. 
This limitation prevents the method from fully extracting the available information, thereby leading to a reduction in power.

\begin{figure}[!htb]
    \centering
    \includegraphics[width=\linewidth]{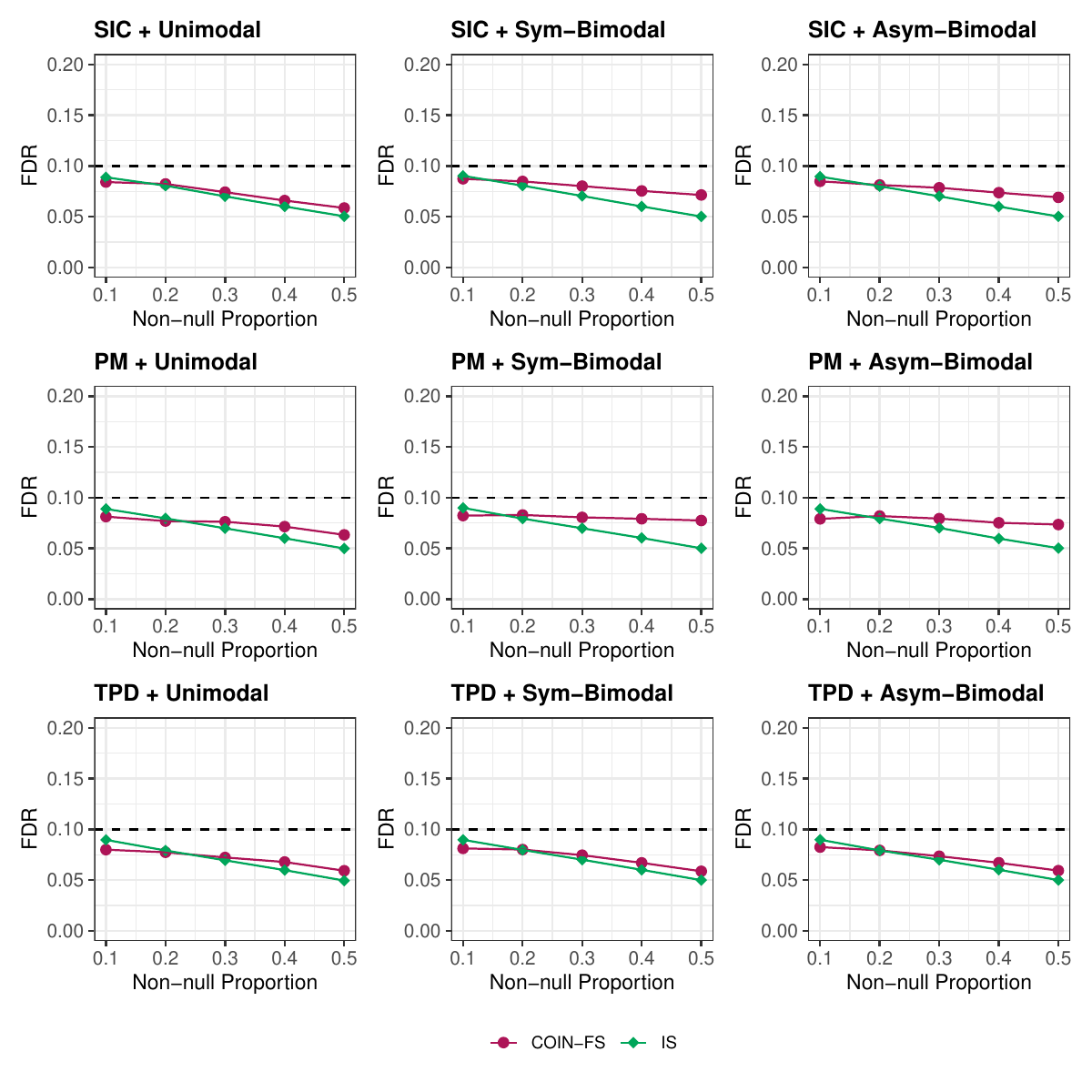}
    \caption{False discovery rates (FDRs) of \textit{IS} and \textit{COIN-FS}. The simulation settings are identical to those used in Scenario~1 of the main manuscript. The dashed horizontal line indicates the target FDR level of 0.1.}
    \label{fig:sim1_FDR_IS}
\end{figure}

\begin{figure}[!htb]
    \centering
    \includegraphics[width=\linewidth]{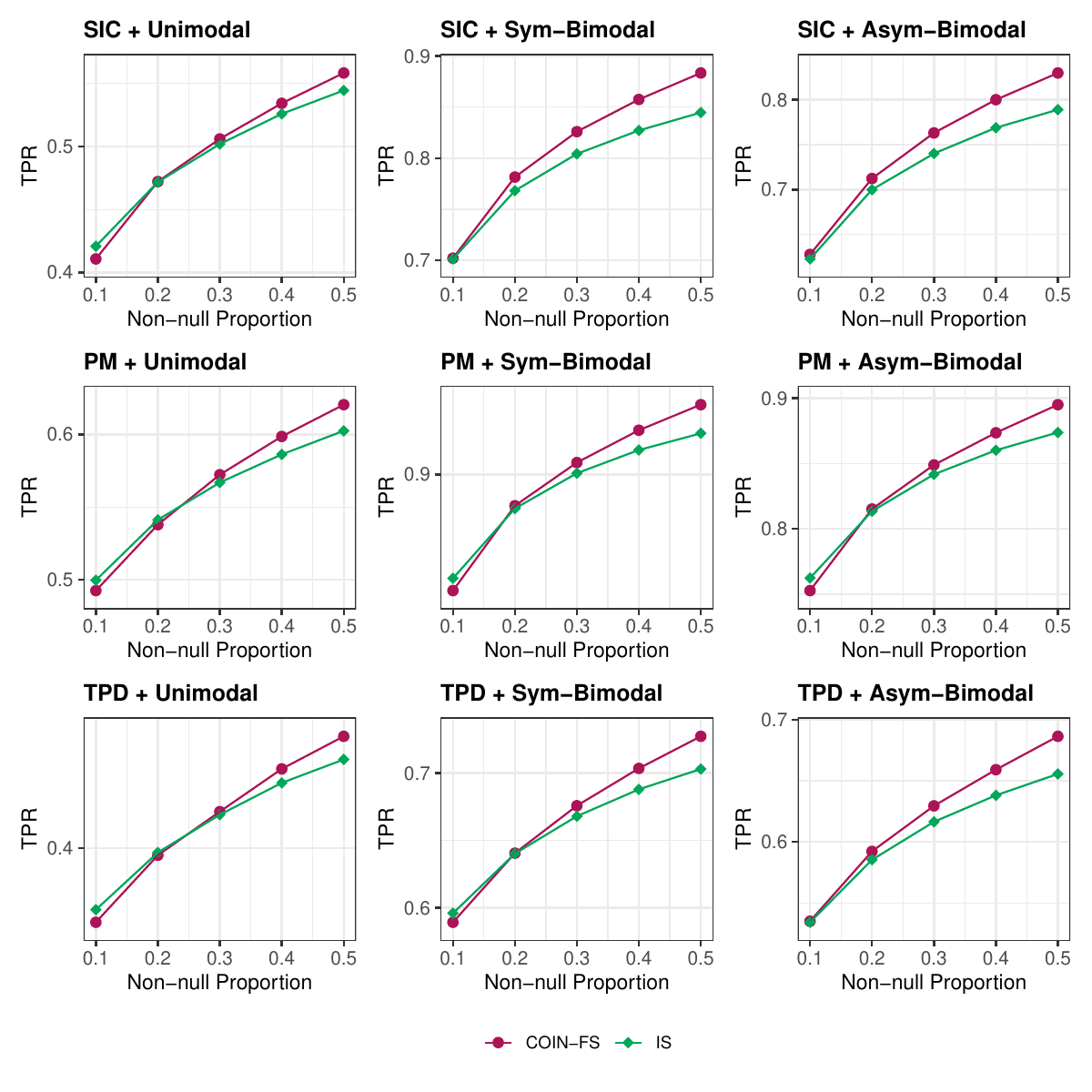}
    \caption{True positive rates (TPRs) of \textit{IS} and \textit{COIN-FS}. The simulation settings are identical to those used in Scenario~1 of the main manuscript.}
    \label{fig:sim1_TPR_IS}
\end{figure}

\begin{figure}[!htb]
    \centering
    \includegraphics[width=\linewidth]{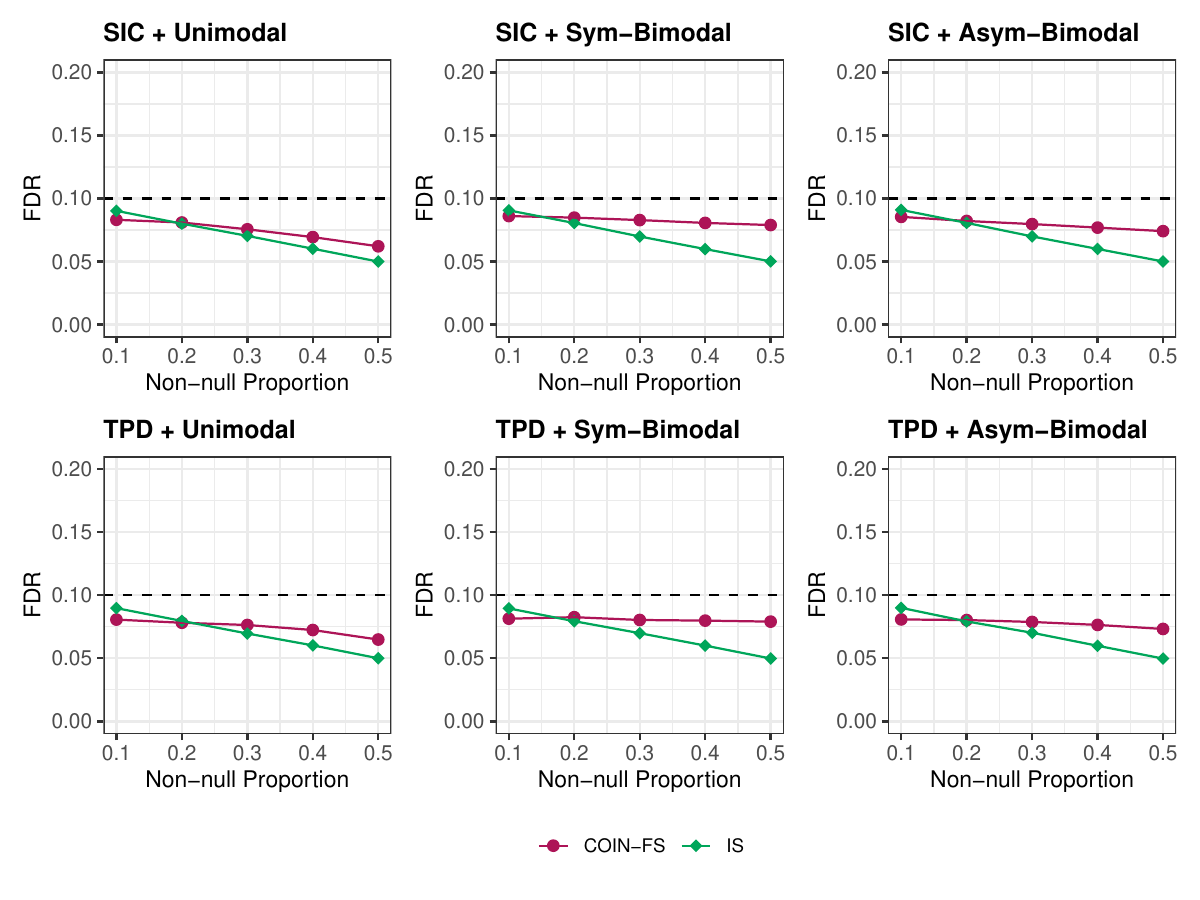}
    \caption{False discovery rates (FDRs) of \textit{IS} and \textit{COIN-FS}. The simulation settings are identical to those used in Scenario~2 of the main manuscript. The dashed horizontal line indicates the target FDR level of 0.1.}
    \label{fig:sim2_FDR_IS}
\end{figure}

\begin{figure}[!htb]
    \centering
    \includegraphics[width=\linewidth]{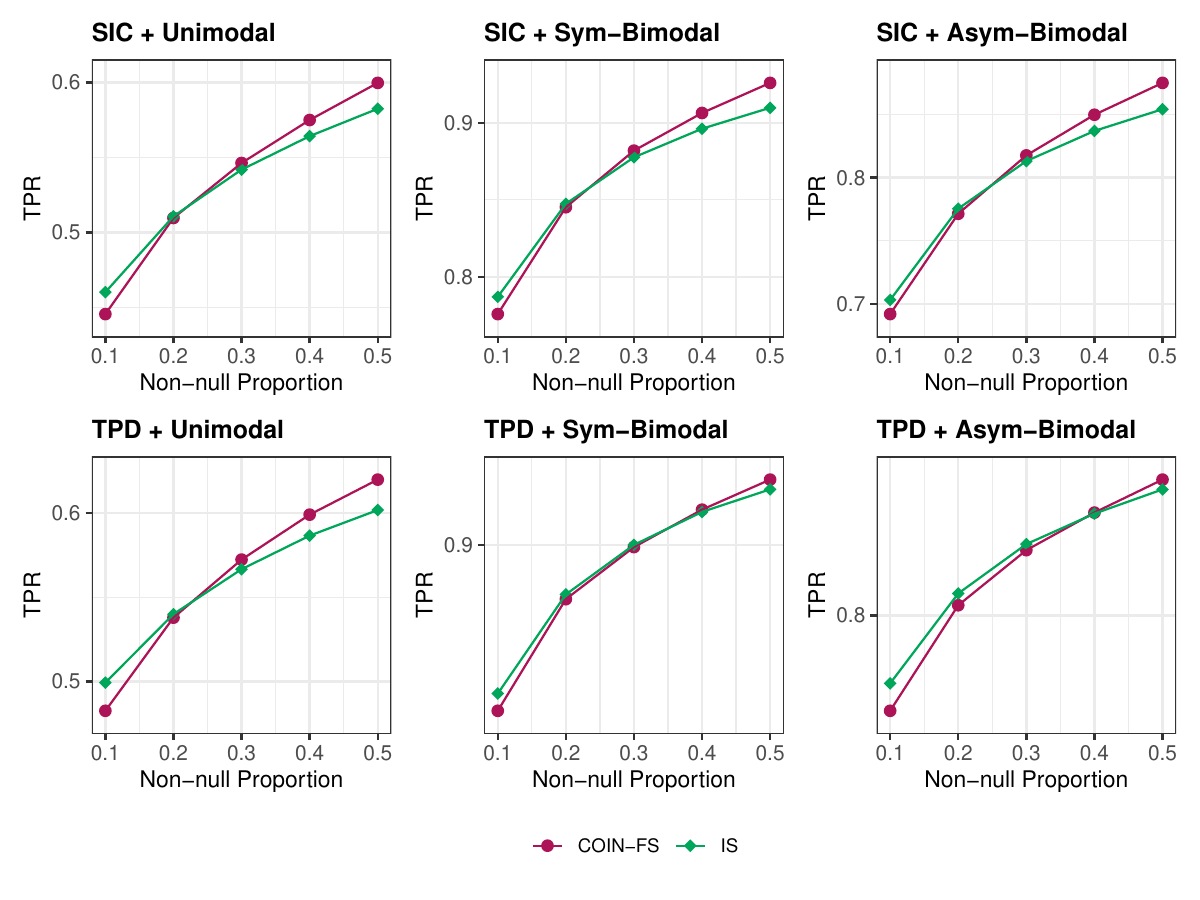}
    \caption{True positive rates (TPRs) of \textit{IS} and \textit{COIN-FS}. The simulation settings are identical to those used in Scenario~2 of the main manuscript.}
    \label{fig:sim2_TPR_IS}
\end{figure}
    
\FloatBarrier
\section{Additional Details on the Real Data}\label{detail_real}
\paragraph{Example 1: DNA Methylation Dataset.}
We analyze genome-wide DNA methylation data from \cite{zhang2013genome}, which consist of log$_2$ ratios of methylated to unmethylated signal intensities measured at 439,918 CpG sites. The dataset includes 10 samples collected from 3 individuals and encompasses 4 immune cell types, including \textit{naive T-cells} and \textit{antigen-activated naive T-cells}. The primary objective of this analysis is to identify CpG sites where DNA methylation levels significantly differ between these two biologically related cell types.

The summary statistics $ (X_i, S_i^2) $ were computed from the linear regression model fitted at each CpG site $ i $, defined as:
\begin{equation*}
    y_{ij} = \mu_i + \alpha_{i,k(j)} + \beta_{i,c(j)} + \varepsilon_{ij},
\end{equation*}
where $ y_{ij} $ denotes the log-transformed methylation level at site $ i $ for sample $ j $, $ \mu_i $ is a site-specific intercept, $ \alpha_{i,k(j)} $ represents the effect of individual $ k(j) $, $ \beta_{i,c(j)} $ denotes the effect of cell type $ c(j) $, and $ \varepsilon_{ij} $ is an independent error term.

To assess whether methylation levels differ between \textit{naive T-cells} and \textit{antigen-activated naive T-cells}, the following hypotheses were tested at each site:
\begin{equation*}
    H_{0,i} : \beta_{i,\text{naive}} = \beta_{i,\text{activated}} 
    \quad \text{vs.} \quad
    H_{1,i} : \beta_{i,\text{naive}} \neq \beta_{i,\text{activated}}
\end{equation*}
The test statistic $ X_i = \widehat{\beta}_{i,\text{naive}} - \widehat{\beta}_{i,\text{activated}} $ and its estimated variance $ S_i^2 = \widehat{\mathrm{Var}}(X_i) $ were obtained from the least squares estimates of the regression coefficients and their estimated variance. Each regression model was based on all 10 samples and included 6 parameters in total: one intercept term, two indicator variables for the three individuals, and three indicator variables for the four cell types. Accordingly, the residual degrees of freedom used for variance estimation was $ \nu = 10 - 6 = 4 $.

\paragraph{Example 2: CLL Dataset.}
We analyze gene expression data from \cite{dietrich2018drug}, which includes paired measurements from 12 patients diagnosed with chronic lymphocytic leukemia (CLL). For each patient, gene expression levels were profiled before and after a 12-hour incubation with Ibrutinib, a targeted therapy for CLL, resulting in measurements for 48,107 genes. The primary objective of this analysis is to identify genes whose expression levels significantly differ between the pre-treatment and post-treatment conditions.

The summary statistics $ (X_i, S_i^2) $ were computed from the paired differences in expression levels for each gene $ i $, defined as:
\begin{equation*}
    d_{ij} = y_{ij}^{\text{post}} - y_{ij}^{\text{pre}},
\end{equation*}
where $ d_{ij} $ denotes the difference in expression levels of gene $ i $ for patient $ j $, calculated by subtracting the pre-treatment value from the post-treatment value.

To assess whether the treatment had a significant effect on gene expression, the following hypotheses were tested at each gene:
\begin{equation*}
    H_{0,i} : \mu_{i,\mathrm{post}} = \mu_{i,\mathrm{pre}} 
    \quad \text{vs.} \quad
    H_{1,i} : \mu_{i,\mathrm{post}} \neq \mu_{i,\mathrm{pre}}.
\end{equation*}
The test statistic $ X_i = \bar{d}_i = \frac{1}{12} \sum_{j=1}^{12} d_{ij} $ and its estimated variance $ S_i^2 = \widehat{\mathrm{Var}}(X_i) $ summarize the evidence for differential expression at gene $ i $.  Because the analysis is based on 12 matched pairs, the residual degrees of freedom used for variance estimation is $ \nu = 12 - 1 = 11 $.

\paragraph{Example 3: Breast Cancer Proteomics Dataset.}
We analyze proteomics data from \cite{terkelsen2021high}, in which liquid chromatography-tandem mass spectrometry (LC-MS/MS) was applied to 34 tumor interstitial fluid samples collected from breast cancer patients. The dataset includes three breast cancer subtypes---luminal, Her2, and triple-negative---and the samples were processed across four experimental batches. Protein expression levels were measured for 6,763 proteins. The primary objective of this analysis is to identify proteins whose expression levels significantly differ between the luminal and Her2 subtypes.

The summary statistics $ (X_i, S_i^2) $ were computed from the linear regression model fitted at each protein $ i $, defined as:
\begin{equation*}
    y_{ij} = \mu_i + \alpha_{i,b(j)} + \beta_{i,s(j)} + \varepsilon_{ij},
\end{equation*}
where $ y_{ij} $ denotes the log$_2$-transformed expression intensity of protein $ i $ in sample $ j $, $ \mu_i $ is a protein-specific intercept, $ \alpha_{i,b(j)} $ represents the effect of experimental batch $ b(j) $, $ \beta_{i,s(j)} $ denotes the effect of cancer subtype $ s(j) $, and $ \varepsilon_{ij} $ is an independent error term.

To assess whether protein expression levels differ between the luminal and Her2 subtypes, the following hypotheses were tested at each protein:
\begin{equation*}
    H_{0,i} : \beta_{i,\text{luminal}} = \beta_{i,\text{Her2}} 
    \quad \text{vs.} \quad
    H_{1,i} : \beta_{i,\text{luminal}} \neq \beta_{i,\text{Her2}}.
\end{equation*}
The test statistic $ X_i = \widehat{\beta}_{i,\text{luminal}} - \widehat{\beta}_{i,\text{Her2}} $ and its estimated variance $ S_i^2 = \widehat{\mathrm{Var}}(X_i) $ summarize the evidence for differential expression at protein $ i $. These quantities are derived from the least squares estimates of the regression coefficients and their estimated variance. Each regression model is based on 34 samples and includes 6 parameters: one intercept term, three indicator variables for the four experimental batches, and two indicator variables for the three cancer subtypes. Accordingly, the residual degrees of freedom used for variance estimation is $ \nu = 34 - 6 = 28 $.

\section{Summary of the Proposed Algorithms}\label{algorithms}

\begin{algorithm}[!htb]
    \caption{COIN}\label{alg_1}
    
    \textbf{Input:} Test dataset $\mathcal{D}_1 = \{(X_{i,1}, S_{i,1}^2)\}_{i = 1}^{m^{\mathcal{D}_1}}$, training dataset $\mathcal{D}_2 = \{(X_{i,2}, S_{i,2}^2)\}_{i = 1}^{m^{\mathcal{D}_2}}$, target FDR level $\alpha$
    
    \textbf{Output:} Set of rejected hypotheses $\mathcal{R} \subset \mathcal{H}^{\mathcal{D}_1}$
    
    \begin{algorithmic}[1]
        \State Construct the pseudo calibration dataset $ \tilde{\mathcal{D}}_1 = \{(\tilde{X}_{i,1}, S_{i,1}^2)\}_{i \in \mathcal{H}^{\mathcal{D}_1}} $,  
        where the calibration variables $ \tilde{X}_{i,1} $ are generated as in \eqref{pseudo_cal}.
        
        \State Specify a working prior distribution and conformity score function $u(\cdot, \cdot)$, and estimate them based on training dataset $\mathcal{D}_2$.
        \State Compute conformity scores for test and pseudo calibration datasets:
        \begin{equation*}
            \mathcal{U} = \{u_i = \hat{u}(X_{i,1}, S_{i,1}^2) : i \in \mathcal{H}^{\mathcal{D}_1}\}, \quad
            \tilde{\mathcal{U}} = \{\tilde{u}_i = \hat{u}(\tilde{X}_{i,1}, S_{i,1}^2): i \in \mathcal{H}^{\mathcal{D}_1}\}.
        \end{equation*}
        
        \State Compute the threshold $\tau$ as defined in \eqref{thresh_unknownG}, and define the decision rules:  
        \begin{equation*}
            \delta_i = \mathbb{I}(u_i \leq \tilde{u}_i \wedge \tau), \quad i \in \mathcal{H}^{\mathcal{D}_1}
        \end{equation*}
        
        \State \textbf{Return:} The rejection set $\mathcal{R} = \{i \in \mathcal{H}^{\mathcal{D}_1} : \delta_i = 1\}$.
    \end{algorithmic}
\end{algorithm}

\begin{algorithm}[!htb]
    \caption{Sample-Splitting Method}\label{alg_2}
    
    \textbf{Input:} Individual-level test dataset $\mathcal{D}^{\text{raw}} \in \mathbb{R}^{m^{\mathcal{D}} \times n}$, target FDR level $\alpha$ 
    
    \textbf{Output:} Set of rejected hypotheses $\mathcal{R} \subset \mathcal{H}^\mathcal{D}$
    
    \begin{algorithmic}[1]
        \State Split the $n$ samples (columns) of $\mathcal{D}^{\text{raw}}$ into two equal-sized subsets, denoted by $\mathcal{D}_1^{\text{raw}}$ and $\mathcal{D}_2^{\text{raw}}$, while preserving the underlying data structure.
        
        \State Compute summary statistics from each split:
        \begin{equation*}
            \mathcal{D}_1^{\text{split}} = \{(X_{i,1}^{(\text{split})}, S_{i,1}^{2(\text{split})})\}_{i \in \mathcal{H}^\mathcal{D}}, \quad 
            \mathcal{D}_2^{\text{split}} = \{(X_{i,2}^{(\text{split})}, S_{i,2}^{2(\text{split})})\}_{i \in \mathcal{H}^\mathcal{D}}
        \end{equation*}
        
        \State Apply Algorithm~\ref{alg_1} using $\mathcal{D}_1^{\text{split}}$ as the test dataset and $\mathcal{D}_2^{\text{split}}$ as the training dataset
        
        \State \textbf{Return} The rejection set $\mathcal{R} = \{i \in \mathcal{H}^\mathcal{D} : \delta_i = 1\}$
    \end{algorithmic}
\end{algorithm}

\begin{algorithm}[!htb]
    \caption{Feature-Splitting Method}\label{alg_3}
    
    \textbf{Input:} Summary-level test dataset $\mathcal{D} = \{(X_i, S_i^2)\}_{i \in \mathcal{H}^{\mathcal{D}}}$; number of folds $K$; fold-specific FDR levels $\alpha^{(k)}$; target FDR level for $e$BH procedure $\alpha_{e\text{BH}}$
    
    \textbf{Output:} Rejection set $\mathcal{R} \subset \mathcal{H}^{\mathcal{D}}$
    
    \begin{algorithmic}[1]
            \State Partition $\mathcal{H}^{\mathcal{D}}$ into $K$ disjoint subsets $\mathcal{H}^{(1)}, \ldots, \mathcal{H}^{(K)}$
            
            \For{$k = 1$ to $K$}
                \State Define $\mathcal{D}^{(k)} \coloneqq \{(X_i, S_i^2)\}_{i \in \mathcal{H}^{(k)}}$, $\mathcal{D}^{(-k)} \coloneqq \{(X_i, S_i^2)\}_{i \in \mathcal{H}^{\mathcal{D}} \setminus \mathcal{H}^{(k)}}$
                
                \State Estimate $\hat{G}^{(k)}$ from $\mathcal{D}^{(-k)}$
                
                \State Generate calibration variables $\tilde{X}_i^{(k)}$ for $i \in \mathcal{H}^{(k)}$ as in \eqref{pseudo_cal}
                
                \State Estimate conformity score function $\hat{u}^{(k)}(\cdot, \cdot)$ based on $\mathcal{D}^{(-k)}$
                
                \State Compute scores: $u_i^{(k)} = \hat{u}^{(k)}(X_i^{(k)}, S_i^{2(k)})$, $\tilde{u}_i^{(k)} = \hat{u}^{(k)}(\tilde{X}_i^{(k)}, S_i^{2(k)})$ for $i \in \mathcal{H}^{(k)}$
                
                \State Compute threshold $\tau^{(k)}$ as in \eqref{thresh_unknownG} using fold-specific FDR level $\alpha^{(k)}$
                
                \State Compute test statistic $E_i^{(k)}$ as in \eqref{e_variable} for $i \in \mathcal{H}^{(k)}$
                
                \State Store $\mathcal{E}^{(k)} \coloneqq \{E_i^{(k)} : i \in \mathcal{H}^{(k)}\}$
            \EndFor
        \State Aggregate: $\mathcal{E} \coloneqq \bigcup_{k = 1}^K \mathcal{E}^{(k)}$
        
        \State (Optional) Draw $U \sim \text{Unif}(0,1)$, define $\mathcal{E}^* \coloneqq \{ E_i^{(k)} / U : i\in \mathcal{H}^{(k)}, k \in [K]\}$
        
        \State Apply $e$BH procedure to $\mathcal{E}$ (or $\mathcal{E}^*$) with level $\alpha_{e\text{BH}}$
        
        \State \Return The rejection set $\mathcal{R}$ 
    \end{algorithmic}
\end{algorithm}

\FloatBarrier
\section{The Oracle Case where the G is Known}\label{oracle_case}
In the main manuscript, we have developed the algorithms and their theoretical properties under the assumption that the prior distribution $G$ is unknown. In this section, we turn to the oracle case where $G$ is known, and explicitly describe how the algorithms and theoretical results are adjusted under this condition. 

\begin{remark}
    In the normal means inference problem, it is common to develop methods under the assumption that the variances $\sigma_i^2$ are known \citep{stephens2017false, fu2022heteroscedasticity}. This assumption is motivated by the fact that the variance estimates $S_i^2$ converge to $\sigma_i^2$ as the degrees of freedom $\nu \to \infty$. In contrast, in this section, we assume that the prior distribution $G$ of the variances $\sigma_i^2$ is known. Within our framework, this assumption is justified by the fact that, as the size of the training dataset $m^{\mathcal{D}_2}$ increases, the estimate $\hat{G}$ converges to the true distribution $G$.
\end{remark}

\paragraph{Algorithms Adapted to the Known $G$.}
In the main manuscript, we proposed three algorithms for the normal means inference problem under the assumption that $ G $ is unknown (Algorithms~\ref{alg_1}--\ref{alg_3}). In the current oracle setting where the prior distribution $ G $ is known, these algorithms can be modified by directly leveraging the knowledge of $ G $. Specifically, since $ G $ is available, the estimation step of $ G $ is no longer required in any of the algorithms, and the calibration variables are generated according to \eqref{oracle_cal} instead of \eqref{pseudo_cal}. By applying these two modifications while keeping the remaining steps unchanged, we obtain three algorithms tailored to the setting where $ G $ is known. These are referred to as Algorithms~\ref{alg_1}$^*$, \ref{alg_2}$^*$, and \ref{alg_3}$^*$---modified counterparts of Algorithms~\ref{alg_1}, \ref{alg_2}, and \ref{alg_3}, respectively.

\paragraph{Theory for the Algorithm \ref{alg_1}$^*$ and \ref{alg_2}$^*$.}
In this section, we establish that the decision rules derived from Algorithms~\ref{alg_1}$^*$ and~\ref{alg_2}$^*$ control the FDR at the prescribed level in finite samples. Since Algorithm~\ref{alg_2}$^*$ essentially shares the same theoretical properties as Algorithm~\ref{alg_1}$^*$, it suffices to present the theory for Algorithm~\ref{alg_1}$^*$. We begin by introducing two lemmas that will be used to prove our main result.

\begin{lemma}\label{lemma8}
    Suppose the hierarchical model described in \eqref{hierar_model} holds, and assume that there are no ties between $ u_i $ and $ \tilde{u}_i^* $\footnote{We use the superscript $^*$ to denote algorithms that rely on the true prior distribution $G$, such as Algorithm~\ref{alg_1}$^*$. Similarly, for other quantities that depend on the knowledge of $G$, we also attach a superscript $^*$ to distinguish them from those constructed without knowledge of $G$. For example, the calibration variable generated from $h_G(x \mid \mu = 0, s^2)$ is denoted by $\tilde{X}_{i,1}^*$, and the corresponding conformity score is denoted by $u_i^*$.} for all $ i \in \mathcal{H}^{\mathcal{D}_1} $. Given the data $ \mathcal{D}_2 $, we have
    \begin{equation*}
        \left(\left\{1+\text{sgn}(u_i - \tilde{u}_i^*) \right\}/2\mid u_k \wedge \tilde{u}_k^* :  k \in \mathcal{H}^{\mathcal{D}_1}\right) \overset{\text{i.i.d.}}{\sim} \text{Bernoulli}(1/2), \quad \forall i \in \mathcal{H}_0^{\mathcal{D}_1},
    \end{equation*}
    where $ \text{sgn}(\cdot) $ denotes the sign function.
\end{lemma}

\begin{proof}
    The proof of Lemma~\ref{lemma8} can be viewed as a special case of the proof of Lemma~\ref{lemma5}, derived under the additional assumption that the prior distribution $ G $ is known.
    From the result in Lemma \ref{lemma5}, we have
    \begin{equation*}
        \left(\text{sgn}(u_i - \tilde{u}_i^*) \,\middle|\, u_k \wedge \tilde{u}_k^* : k \in \mathcal{H}^{\mathcal{D}_1}\right) \overset{\text{i.i.d.}}{\sim} \text{Bernoulli}(\rho), \quad \forall i \in \mathcal{H}_0^{\mathcal{D}_1},
    \end{equation*}
    where $ \rho $ is a random variable defined as
    \begin{equation*}
        \rho \coloneqq \mathbb{P}\left(\text{sgn}(u_i - \tilde{u}_i^*) = 1 \,\middle|\, u_i \wedge \tilde{u}_i^* \right), \quad \text{for } i \in \mathcal{H}_0^{\mathcal{D}_1}.
    \end{equation*}
    In contrast to Lemma~\ref{lemma5}, we considers the case where the calibration variables are generated based on the known distribution $ G $. This implies that the conditional exchangeability property between $ \tilde{X}_{i,1} $ and $ X_{i,1} $ given $ S_{i,1}^2 $ holds for all $ i \in \mathcal{H}_0^{\mathcal{D}_1} $. Consequently, we can obtain stronger conclusion that $ \rho = 1/2 $.
\end{proof}

Before proceeding to the next result, we introduce some necessary notations. Let $\xi_i^* = \mathbb{I}(u_i \leq \tilde{u}_i^*)$, $s_i^* = u_i \wedge \tilde{u}_i^*$, and let $s_{(i)}^*$ denote the $i$-th smallest value among $\{s_j^*\}_{j \in \mathcal{H}^{\mathcal{D}_1}}$. Based on these notations, define the cumulative counts $V_i^*$ and $\tilde{V}_i^*$ as
\begin{equation*}
    V_i^* = \sum_{j \in \mathcal{H}_0^{\mathcal{D}_1}} \xi_j^* \cdot \mathbb{I}(s_j^* \leq s_{(i)}^*), \quad 
    \tilde{V}_i^* = \sum_{j \in \mathcal{H}_0^{\mathcal{D}_1}} (1 - \xi_j^*) \cdot \mathbb{I}(s_j^* \leq s_{(i)}^*).
\end{equation*}
Furthermore, for a target FDR level $\alpha$, define
\begin{equation*}
    \hat{k}^* = \hat{k}^*(\alpha) \coloneqq \max \left\{k \in [m^{\mathcal{D}_1}] : \frac{1 + \sum_{i \in \mathcal{H}^{\mathcal{D}_1}} (1 - \xi_i^*) \cdot \mathbb{I}(s_i^* \leq s_{(k)}^*)}{\sum_{i \in \mathcal{H}^{\mathcal{D}_1}} \xi_i^* \cdot \mathbb{I}(s_i^* \leq s_{(k)}^*) \vee 1} \leq \alpha \right\}.
\end{equation*}

\begin{lemma}\label{lemma9} 
    Suppose the hierarchical model described in \eqref{hierar_model} holds, and assume that there are no ties between $ u_i $ and $ \tilde{u}_i^* $ for all $ i \in \mathcal{H}^{\mathcal{D}_1} $.
    Given the data $ \mathcal{D}_2 $, the following inequality holds:
    \begin{equation*}
        \mathbb{E}\left[ \frac{V_{\hat{k}^*}^*}{1 + \tilde{V}_{\hat{k}^*}^*} \right] \leq 1.
    \end{equation*}
\end{lemma}

\begin{proof}
    Since the proof closely parallels that of Lemma~\ref{lemma6}, we omit the details here.
\end{proof}

Theorem~\ref{theorem4} establishes that the decision rule derived from Algorithm~\ref{alg_1}$^*$ controls the FDR at the nominal level in finite samples.

\begin{theorem}\label{theorem4}
     Suppose the hierarchical model described in \eqref{hierar_model} holds, and assume that there are no ties between $ u_i $ and $ \tilde{u}_i^* $ for all $ i \in \mathcal{H}^{\mathcal{D}_1} $. Then, the decision rule, $\boldsymbol{\delta}^*$, derived from Algorithm~\ref{alg_1}$^*$, controls the FDR at pre-specified level $\alpha$ in finite samples, i.e.,
    \begin{equation*}
        \text{FDR}(\boldsymbol{\delta}^*) \leq \alpha.
    \end{equation*}
\end{theorem}

\begin{proof}
Following the same argument as in the proof of Theorem~\ref{theorem1}, it can be shown that the decision rule based on the threshold $\tau^*$, obtained as in \eqref{thresh_unknownG}, is equivalent to that based on the threshold $s_{(\hat{k}^*)}^*$.

By definition of $\text{FDP}(\boldsymbol{\delta}^*)$ and the fact established above, we have 
\begin{equation*}
\begin{split}
    \text{FDP}(\boldsymbol{\delta}^*)
    &= \frac{1 + \sum_{i \in \mathcal{H}^{\mathcal{D}_1}} (1-\xi_i^*) \cdot \mathbb{I}(s_i^* \leq s_{(\hat{k})}^*)}{[\sum_{i \in \mathcal{H}^{\mathcal{D}_1}} \xi_i^* \cdot \mathbb{I}(s_i^* \leq s_{(\hat{k})}^*)]\vee 1}\\
    &\quad\quad\times \frac{1 + \sum_{i \in \mathcal{H}_0^{\mathcal{D}_1}} (1-\xi_i^*) \cdot \mathbb{I}(s_i^* \leq s_{(\hat{k})}^*)}{1 + \sum_{i \in \mathcal{H}^{\mathcal{D}_1}} (1-\xi_i^*) \cdot \mathbb{I}(s_i^* \leq s_{(\hat{k})}^*)} \\
    &\quad\quad\times \frac{\sum_{i \in \mathcal{H}_0^{\mathcal{D}_1}} \xi_i^* \cdot \mathbb{I}(s_i^* \leq s_{(\hat{k})}^*)}{1 + \sum_{i \in \mathcal{H}_0^{\mathcal{D}_1}} (1-\xi_i^*) \cdot \mathbb{I}(s_i^* \leq s_{(\hat{k})}^*)}\\
    &\leq \alpha \cdot 1 \cdot \frac{\sum_{i \in \mathcal{H}_0^{\mathcal{D}_1}} \xi_i^* \cdot \mathbb{I}(s_i^* \leq s_{(\hat{k}^*)})}{1 + \sum_{i \in \mathcal{H}_0^{\mathcal{D}_1}} (1-\xi_i^*) \cdot \mathbb{I}(s_i^* \leq s_{(\hat{k})}^*)}\\
    &= \alpha \cdot \frac{V_{\hat{k}}^*}{1 + \tilde{V}_{\hat{k}}^*}.
    \end{split}
\end{equation*}
The inequality follows from the definition of $ \hat{k}^* $, along with the fact that $ \mathcal{H}_0^{\mathcal{D}_1} \subset \mathcal{H}^{\mathcal{D}_1} $. By taking expectations on both sides and applying Lemma~\ref{lemma9}, we obtain that $\mathrm{FDR}(\boldsymbol{\delta}^*) \leq \alpha$.
\end{proof}

\paragraph{Generalized/compound $e$-variables and the $e\text{BH}$ Procedure.}\label{sub_sec_3.4}
Before presenting the theory for Algorithm~\ref{alg_3}$^*$, we first introduce the concepts of generalized/compound $ e $-variables and the $ e\text{BH} $ procedure, which play a central role in the theoretical development of Algorithm~\ref{alg_3}$^*$.

Consider a collection of hypotheses indexed by the set $ \mathcal{H} $, with $ \mathcal{H}_0 \subseteq \mathcal{H} $ denoting the subset corresponding to the true null hypotheses. For each hypothesis $ i \in \mathcal{H} $, let $ E_i $ be a non-negative random variable. We say that $ \{E_i : i \in \mathcal{H}\} $ forms a set of generalized/compound $ e $-variables if the expected sum of $ E_i $ over the true nulls is bounded above by the total number of hypotheses, that is,
\begin{equation*}
    \mathbb{E}\left[\sum_{i \in \mathcal{H}_0} E_i\right] \leq |\mathcal{H}|.
\end{equation*}

Based on a set of generalized/compound $ e $-variables, \citet{wang2022false} proposed a multiple testing procedure known as the $ e\text{BH} $ procedure. The $ e\text{BH} $ procedure can be view as a counterpart to the classical BH procedure by \citet{benjamini1995controlling}, with the key difference being that hypotheses are ranked according to the $ e $-values rather than $ p $-values. Specifically, let $ E_{(i)} $ denote the $ i $-th largest $ e $-variable among all $ m $ $e$-variables:
\begin{equation*}
    E_{(1)} \geq E_{(2)} \geq \cdots \geq E_{(m)} \geq 0.    
\end{equation*}
The $ e\text{BH} $ procedure rejects the $ i $-th hypothesis if $ E_i \geq E_{(\hat{k})} $, where
\begin{equation*}
    \hat{k} = \max\left\{i : E_{(i)} \geq \frac{m}{\alpha \cdot i} \right\}.    
\end{equation*}
Accordingly, the decision rule based on the generalized $ e $-variables is given by
\begin{equation*}
    \boldsymbol{\delta}_{e\text{BH}} = \{\delta_{i, e\text{BH}} : i \in \mathcal{H}\}, \quad \text{where} \quad \delta_{i, e\text{BH}} = \mathbb{I}(E_i \geq E_{(\hat{k})}).
\end{equation*}
\citet{wang2022false} show that the FDR level of the decision rule $ \boldsymbol{\delta}_{e\text{BH}} $ is controlled at the target level in finite samples.

Now, returning to our setting, consider the random variables $ E_i $ defined as:
\begin{equation}\label{e-var}
    E_i = \frac{|\mathcal{H}^{\mathcal{D}_1}| \cdot \xi_i^* \cdot \mathbb{I}(s_i^* \leq s_{(\hat{k})}^*)}{1 + \sum_{i \in \mathcal{H}^{\mathcal{D}_1}} (1 - \xi_i^*) \cdot \mathbb{I}(s_i^* \leq s_{(\hat{k})}^*)}, \quad \forall i \in \mathcal{H}^{\mathcal{D}_1}.
\end{equation}
It can be verified that the random variables $ E_i $ satisfy the definition of the generalized/compound $ e $-variable:
\begin{equation*}
    \mathbb{E}\left[\sum_{i \in \mathcal{H}_0^{\mathcal{D}_1}} E_i\right]  = |\mathcal{H}^{\mathcal{D}_1}| \cdot  \mathbb{E}\left[ \frac{\sum_{i \in \mathcal{H}_0^{\mathcal{D}_1}}  \xi_i^* \cdot \mathbb{I}(s_i^* \leq s_{(\hat{k})}^*)}{1 + \sum_{i \in \mathcal{H}^{\mathcal{D}_1}} (1 - \xi_i^*) \cdot \mathbb{I}(s_i^* \leq s_{(\hat{k})}^*)}\right] \leq |\mathcal{H}^{\mathcal{D}_1}|,
\end{equation*}
where the last inequality follows from Lemma~\ref{lemma9}. Therefore, applying the $ e\text{BH} $ procedure to these generalized/compound $ e $-variables ensures finite-sample FDR control.

Theorem~\ref{theorem5} shows that the decision rule $ \boldsymbol{\delta}_{e\text{BH}} $, obtained by applying the $ e\text{BH} $ procedure to the generalized/compound $ e $-variables defined in \eqref{e-var}, is equivalent to the decision rule $ \boldsymbol{\delta} $ obtained by applying Algorithm~\ref{alg_1}$^*$.

\begin{theorem}[Proposition 1 in \cite{zhao2025conformalized}]\label{theorem5}
    Suppose that there are no ties between $u_i$ and $\tilde{u}^*$ for all $i \in \mathcal{H}^{\mathcal{D}_1}$. Let $ \{E_i : i \in \mathcal{H}^{\mathcal{D}_1} \} $ be a set of generalized $ e $-variables defined in \eqref{e-var}, and let $ \boldsymbol{\delta}_{e\text{BH}} $ denote the decision rule obtained by applying the $ e\text{BH} $ procedure to this set. Then, the decision rule $ \boldsymbol{\delta}_{e\text{BH}} $ is equivalent to the decision rule $ \boldsymbol{\delta} $ produced by Algorithm~\ref{alg_1}$^*$.
\end{theorem}

\paragraph{Theory for the Algorithm \ref{alg_3}$^*$.}
We are now ready to establish the theoretical guarantee of Algorithm~\ref{alg_3}$^*$. In Algorithm~\ref{alg_3}$^*$, for each fold $k = 1, 2, \ldots, K$, we compute the non-negative random variables
\begin{equation*}
    E_i^{(k)} = \frac{|\mathcal{H}^{(k)}| \cdot \xi_i^{*(k)} \cdot \mathbb{I}\left(s_i^{*(k)} \leq \tau^{*(k)}\right)}{1 + \sum_{j \in \mathcal{H}^{(k)}} \left(1 - \xi_j^{*(k)}\right) \cdot \mathbb{I}\left(s_j^{*(k)} \leq \tau^{*(k)}\right)}, \quad \forall i \in \mathcal{H}^{(k)}.
\end{equation*}
After aggregating the random variables from all folds, we define
\begin{equation*}
    \mathcal{E} = \bigcup_{k = 1}^K \mathcal{E}^{(k)}, \quad \mathcal{E}^{(k)} = \left\{ E_i^{(k)} : i \in \mathcal{H}^{(k)} \right\}.
\end{equation*}
It is important to note that the combined set $\mathcal{E}$ forms a collection of generalized/compound $e$-variables with respect to the hypothesis set $\mathcal{H}^{\mathcal{D}}$:
\begin{equation*}
    \mathbb{E}\left[\sum_{k = 1}^K \sum_{i \in \mathcal{H}_0^{(k)}} E_i^{(k)}\right] 
    = \sum_{k = 1}^K \mathbb{E}\left[\sum_{i \in \mathcal{H}_0^{(k)}} E_i^{(k)}\right] 
    \leq \sum_{k = 1}^K |\mathcal{H}^{(k)}| = |\mathcal{H}^{\mathcal{D}}|.
\end{equation*}
The inequality holds because each collection $\mathcal{E}^{(k)}$ forms a set of generalized/compound $e$-variables by construction and Lemma~\ref{lemma9}. Therefore, applying the theoretical results established in \citet{wang2022false} and \citet{xu2023more}, we conclude that the decision rule derived from Algorithm~\ref{alg_3}$^*$ controls the FDR at the target level in finite samples.

\end{document}